\documentclass[3p,12pt,authoryear]{elsarticle}

\usepackage{psfrag}
\usepackage{amsmath}
\usepackage{amssymb}

\usepackage{amsfonts}
\usepackage{latexsym}
\usepackage{graphicx}

\usepackage{colortbl}
\usepackage{fancyhdr}
\usepackage{epstopdf}
\usepackage{float,bbm}
\usepackage{subfigure}
\usepackage{subfig}

\usepackage{verbatim}
\usepackage{enumerate}
\usepackage{multicol}
\usepackage{graphicx}

\newcommand{\ssection}[1]{\section[#1]{\bfseries  \centering #1}  }

\usepackage[usenames,dvipsnames,svgnames,table]{xcolor}

\newcommand{\thetaL} {\theta}

\newtheorem{theorem}{Theorem}
\newtheorem{definition}{Definition}
\newtheorem{assumption}{Assumption}

\newtheorem{lemma}{Lemma}
\newtheorem{corollary}{Corollary}
\newtheorem{remark}{Remark}
\newcommand {\be}{\begin{equation}}
\newcommand {\ee}{\end{equation}}

\newcommand {\bes}{\begin{equation*}}
\newcommand {\ees}{\end{equation*}}

\newcommand{\N}{\mathbb{N}}
\newcommand{\R}{\mathbb{R}}

\newcommand{\q}{z}

\usepackage{multirow}
\usepackage{rotating}
\newcommand{\vertiii}[1]{{\left\vert\kern-0.25ex\left\vert\kern-0.25ex\left\vert #1 \right\vert\kern-0.25ex\right\vert\kern-0.25ex\right\vert}}
\newcommand{\leb}{L^2([0,1])}
\newcommand{\lebn}{{L^2([0,1];\mathbb{R}^n)}}
\newcommand{\lebm}{{L^2([0,1];\mathbb{R}^m)}}
\newcommand{\ls}{L_\mathcal{S}}

\usepackage{soul}

\newcommand{\lt}{L^2}
\newcommand{\type}{t}
\newcommand{\lnn}{{L^2;\mathbb{R}^n}}
\newcommand{\one}{1_N}

\definecolor{forestGreen}{RGB}{20,170,40}
\definecolor{blue}{RGB}{0,0,0}

\def\blue{\color{black}}

\def \J{U}

\newenvironment{proof}{\paragraph{Proof:}}{\hfill$\square$}

\usepackage{thmtools}
\declaretheorem[style=definition]{example}
\renewcommand\thmcontinues[1]{Continued}

\usepackage{textpos}

\begin{document}


\vspace{-2cm}

\begin{frontmatter}

\title{ Graphon games: A statistical framework for network games and interventions}
\tnotetext[mytitlenote]{Research partially supported by the SNSF grant  P300P2$\_$177805,   by ARO MURI W911NF-12-1-0509 and by ARO W911NF-18-1-0407. The authors thank  D. Acemoglu, M. Avella-Medina, M.T. Schaub, S. Segarra,  R. Carmona, A. Jadbabaie as well as the participants to the ``Fourth and Fifth Annual Conference on Network Science and Economics'' for helpful conversations.}

\author{Francesca Parise}
\ead{parisef@mit.edu}
\author{Asuman Ozdaglar}
\ead{asuman@mit.edu}
\address{Laboratory for Information and Decision Systems,\\ Massachusetts Institute of Technology, Cambridge, MA, USA.}

\begin{abstract}
{\small \blue
In this paper, we present a unifying framework for analyzing equilibria and designing interventions for large network games sampled from a stochastic network formation process represented by a graphon. We first introduce a new class of infinite population games, termed \textit{graphon games}, where a continuum of  heterogeneous agents interact according to a graphon. After studying properties of equilibria in graphon games, we show that graphon equilibria can approximate equilibria of large network games sampled from the  graphon.  We next show that, under some  regularity assumptions, the graphon approach enables the design of  asymptotically optimal interventions  via the solution of an optimization problem with much lower dimension than  the one based on the entire network structure. We illustrate our framework on a synthetic dataset of rural villages and show that the graphon intervention can be computed efficiently and based solely on aggregated relational data. 
}
\end{abstract}


\end{frontmatter}

\thispagestyle{plain}
\pagestyle{plain}

\vspace{0.5cm}

\ssection{\textbf{Introduction}}
\label{sec:intro}

{\blue Recent decades have witnessed tremendous progress in the theory of network games, which have been used widely to model, understand and predict behavior in a range of settings  involving strategic interactions of agents embedded in  networked environments.
 Despite this progress,  several issues remain when considering interventions  or regulation of economic behavior over \textit{large scale} networks.  First, in this case the optimization problem that  the central planner needs to solve for determining the optimal intervention is very high dimensional, often scaling with the size of the network. Second, assuming that the central planner has access to detailed information about the network structure is not  a good approximation of reality since collection of exact network data is either extremely costly or, in many settings,  not at all  possible due to proprietary and privacy concerns.\footnote{\blue \blue \cite{breza2017using} estimated that conducting network surveys in $120$ Indian villages would cost approximately $\$190, 000$ and take over eight months. Moreover, proprietary and privacy concerns may arise, for example, when measuring high-risk populations or transactions between networks of financial intermediaries.}

To overcome these issues,  in this paper we  develop a unifying  framework for analysis of equilibria and efficient design of interventions in    \textit{large sampled network games} where agents interact according to a network drawn from a stochastic network formation process, which we represent by a graphon.
 A graphon is a general nonparametric random graph model\footnote{\blue Introduced by \cite{lovasz2006limits,Lovasz2012,Borgs2008}.} which includes commonly used Erd\H{o}s-R\'enyi  and stochastic block models  as special cases  and can be used to formally define the limit of a sequence of graphs when the number of nodes tends to infinity. Exploiting this limit characterization,  we start our analysis of sampled network games by proposing a new class of infinite population games, which we  term \textit{graphon games},   where a continuum of agents interact according to a graphon. After providing existence, uniqueness and continuity results for the equilibrium of a graphon game, we turn to the analysis of equilibria in sampled network games drawn from a graphon. Our key contribution is to provide a characterization of sampled network game equilibria,  in the limit of large populations, by   showing that such equilibria can be approximated by the equilibria of the corresponding graphon game. We provide bounds on the distance between sampled and graphon equilibria as a function of the network size and prove that this distance vanishes as the number of agents grows.

 In addition to enabling a unified analysis of sampled network games, graphon games become particularly useful in designing interventions precisely because they deal with the two problems highlighted above. First we show that, under some regularity assumptions on the graphon - most importantly when the graphon is finite rank\footnote{\blue A graphon  is finite rank if  the corresponding  operator has a finite number of eigenvalues different from zero. While a refinement, finite rank graphons are general enough to nest a large number of random graph models.  For example stochastic block models are finite rank graphons with rank equal to the number of blocks, while randomly grown ranked attachment graph sequences as described in \cite{borgs2011limits} converge to a graphon that has rank 2, see \cite[Section 4.2]{graphons}, and  uniform attachment graph sequences converge to a graphon  which can be very well approximated with a rank 5 graphon.} - the optimization problem faced by the central planner can be approximated by  a low dimensional problem  (with size corresponding to the rank of the graphon instead of the number of agents). Second, under the same assumptions, graphon  interventions can be designed with much less information than the entire network structure. To illustrate this second point, we consider the use of aggregated relational data (ARD), as suggested in \cite{breza2017using}, instead of exact network  data. In other words, we consider  the use of data collected through questions such as \textit{``how many of the agents you interact with have trait $k$?''}, instead of  questions of the form \textit{``what is the identity of all the agents you interact with?''}. Using real world data on households across villages in India,  \cite{breza2017using} showed that, in addition to being much easier to collect\footnote{\blue For the villages in Karnataka, India,  \cite{breza2017using} shows using J-PAL South Asia cost estimated that collecting ARD leads to a 70-80\% cost reduction with respect to the cost of data collected in \cite{banerjee2013diffusion}.},  ARD even on $30\%$ of the individuals suffices to obtain reasonable estimates of many network features of economic interest.  We complement these results by showing through a case study, that ARD can be used to efficiently estimate the parameters of a  network game sampled from a stochastic block model (which is  a  widely used type of graphon),  thus allowing the  design of policy interventions from ARD using the graphon approach. For this case study  the suggested procedure results in an optimization problem with dimension equal to the number of  blocks (communities) instead of  number of agents,  leading to a computationally tractable approach even for large populations.}

\subsection{Detailed contributions}

Our contributions are as follows. First, we formalize the notion of a ``graphon game'' in terms of a continuum of agents indexed in $[0,1]$ and a graphon, represented by a bounded symmetric measurable function $W:[0,1]^2\rightarrow [0,1]$ with $W(x,y) $ denoting the influence of agent $y$'s strategy on agent $x$'s payoff function. We assume that agent $x$'s payoff function depends on his strategy $s(x)\in\R^n$ as well as a \textit{local  aggregate}  of the other agents' strategies computed according to the graphon $W$.
We define the  Nash equilibrium for a graphon game as a strategy profile $s$ at which  no agent can unilaterally  increase its payoff  given  the  fixed local aggregate of the other agents' strategies.\footnote{Similar to the notion of Wardrop equilibrium used in nonatomic routing games  where each agent uses routes of least cost  given the aggregate congestion level, see \cite{wardrop1900some,smith1979existence}.} We then study fundamental properties of such equilibrium and derive sufficient conditions in terms  of the payoff functions, strategy sets and  the underlying graphon to guarantee  existence and uniqueness. Under the same assumptions, we additionally derive a {\blue continuity} result quantifying the effect of graphon changes on the equilibrium outcome.

Our second main contribution is to relate the equilibria of the infinite population graphon game to equilibria of finite  network games sampled from the graphon. We start by showing that any network game can be rewritten as a graphon game, hence graphon games are a generalization of network games. 
We then show that, with high probability,  the graphon game equilibrium  is a good  approximation of the equilibrium of any sampled network game  and we provide a precise mathematical bound for the approximation error  in terms of the size of the sampled network.  {\color{blue}  Using this bound, we show that sampled network equilibria converge almost surely to the graphon equilibrium.
Such a characterization of the limiting strategies allows us to extract fundamental features of equilibrium  in  large network games  which can then in turn be used for analysis or planning of interventions. }
 {\blue For simplicity of exposition,  we first present our convergence results  for the case of \textit{dense undirected networks}, where the number of neighbors grows linearly with the population size and   the network aggregate is defined as the sum of neighbors actions  (normalized by the network size).  We then show in Section~\ref{sec:ext} that our results can be generalized from undirected to directed networks and from games where the sum of neighbors' actions is normalized by the population size to games where it is normalized by each agent's degree. Most importantly, we show that our results can be extended to sparser classes of networks where  the number of neighbors grows sublinearly (but still faster than logarithmically in the population size $N$). This is an oft encountered condition used in random graph theory to ensure that nodes have enough links so that concentration inequality bounds apply, but the required rate of growth is very slow, only being of order larger than $\log(N)$. A similar condition is used  for example in \cite{jackson2019behavioral}. We show within our  case study, that these results lead to useful insights even when the average degree in a network of $1000$ agents is around $20$, illustrating the applicability of our framework to  realistic networks. }

{\blue As a third main contribution, we  turn to the problem of designing targeted interventions in linear quadratic network games, as recently considered in \cite{galeotti2017targeting}. Because of the high dimensionality of the corresponding optimization problem, \cite{galeotti2017targeting} proposed and analyzed the performance of heuristics  based on  spectral properties of the network.   Instead we here suggest an alternative approach  based on a novel optimization problem in the graphon space  which, 
through sampling, provides interventions for  finite sampled network games. We show that such graphon-based interventions are close to optimal and we provide a bound on the distance from  optimality  which decreases as a function of the network size. Additionally, we show 
that for finite rank graphons, the graphon optimization problem is a tractable finite dimensional problem with as many variables as the rank of the underlying graphon.

 To illustrate the computational and informational gains obtained with the graphon approach, we consider a  case study on a simulated dataset of $80$ different networks, drawn as independent realizations of a stochastic block model with $4$ communities, which can for example model interactions among the inhabitants of $80$  different rural villages.  For this case study, the graphon approach leads to a $4$ dimensional optimization problem whereas the optimal intervention and the network heuristic of  \cite{galeotti2017targeting} necessitate solving a problem of dimension equal to the size  $N$ of the network (we set $N=300,600,1200$ in our simulations). Moreover, the graphon approach results in a near optimal solution which provides significant gains over the network heuristic. Finally, within this case study we suggest how our framework can be used to estimate peer effects under partial network data. This is a topic of  recent  interest, as discussed for example in \cite{chandrasekhar2011econometrics,  de2018recovering, boucher2019estimating,lewbel2019social}. Estimating peer effect is not the subject of our work, hence we do not develop this aspect of our theory beyond the intuition given in the case study and a preliminary analysis given in online Appendix \ref{sec:id}. We however believe this could be an interesting future direction enabled by the suggested graphon framework. 
 }

The results discussed so  far are derived under the assumption that agents have perfect information about the sampled network. 
{\color{blue} In online Appendix \ref{bayes}},  we analyze an incomplete information version of sampled network games and develop a close relation between the corresponding Bayesian Nash equilibrium and the graphon equilibrium discussed above. We show that, under suitable regularity conditions and under the assumption that the agents know the graphon generating the sampled network (but not the realization), the graphon equilibrium is an $\varepsilon$-Bayesian Nash equilibrium for the incomplete information game.

\subsection{Related literature}

{\blue Our work complements  results derived for  complete information network games (see e.g. in \cite{ballester2006s,bramoulle2007public,bramoulle2014strategic,jackson2014games,bramoulle2015games,galeotti2017targeting})
by considering a setting where agents have complete information, as in the works above, while the central planner has only access to the stochastic  network  formation model.

  We note that stochastic  network  formation models have been used before in the literature  for the study of \textit{diffusion dynamics and related optimal seeding problems}. This includes 
\cite{golub2012does} and \cite{golub2012homophily}  who studied 
DeGroot dynamics over a variation
of a  stochastic block model and  characterized the  time to consensus in terms of a  measure of clustering, called spectral homophily, that depends only on large-scale linking patterns among groups and not on idiosyncratic details of network realizations.
In more recent work,  \cite{akbarpour2018just} focused on  linearly independent contagion models and showed that randomly seeding a few nodes more leads to asymptotically comparable performances as  seeding based on detailed network information for the case of Erd\H{o}s-R\'enyi  models.\footnote{\blue The analysis is also extended to networks with power-law degree, generalized version of Erd\H{o}s-R\'enyi  model  with high clustering and to different contagion models beyond  linearly independent, see \cite{akbarpour2018just} for more details.}   \cite{jackson2019behavioral} introduced the concept of ``behavioral communities'' (i.e. agents who adopt the same strategy in every possible equilibrium) in the context of  linear threshold dynamics over random networks generated from a stochastic block model and studied their asymptotic properties.  A game-theoretic model of diffusion is considered in \cite{sadler2020diffusion}; therein agents only know their realized degree hence the focus is on Bayesian strategies. Finally, \cite{banerjee2019using} suggested a gossip approach for identifying agents with high diffusion centrality when the network is unknown. 
In contrast with these works, the focus of our paper is to \textit{provide a characterization of the limiting equilibrium strategies} for a large class of  network games \textit{with continuous strategies}\footnote{\blue While contagion models have discrete (typically $0-1$) strategies, we focus here on games with continuous strategies. The type of continuous games considered in our framework has been broadly used in the literature, both in theoretical and empirical works, for example to model applications where  agents  need to decide on their level of effort or investment in a certain activity (see \cite{vives2005complementarities,ballester2006s,acemoglu2015networks,bramoulle2007public,bramoulle2014strategic,allouch2015private}).}  and a \textit{broad range of random graph models} (graphons  include the network formation models mentioned above as special cases) in terms of a new infinite population game. As summarized above, knowledge of such a limiting behavior can be very useful to inspire new approaches to design of interventions or estimation of peer effects.  
 }

{\blue Our results on incomplete information network games, reported in online Appendix \ref{bayes}, are related to two previous works: \cite{galeotti2010network} and \cite{kalai2004large}.   \cite{galeotti2010network} focused on network games over  random networks with fixed number of  agents that  only know  their degree. Properties of the corresponding  Bayesian Nash equilibrium are derived, but no asymptotic analysis is provided.
 \cite{kalai2004large}  proved that the Bayesian Nash equilibrium on a  game with anonymous payoffs (which depend only on how many players select each type-action) is an  $\varepsilon$-Nash equilibrium of the complete information game with $\varepsilon$ going to zero when the number of agents tends to infinity. Two points are noteworthy in relating  this paper to the network game literature and our paper in particular: first, network games capture heterogeneous interactions hence do to satisfy the anonymity assumption in \cite{kalai2004large}; second \cite{kalai2004large} shows that the Bayesian Nash  equilibrium is an $\varepsilon$-Nash equilibrium, instead  we prove that the Bayesian Nash  equilibrium converges (in strategies) to the equilibrium of the corresponding graphon game, thus providing a characterization of the limiting behavior.}

While our goal is to use graphon games to approximate equilibria of sampled network games, we note that graphon games can also be of independent interest as a new model of nonatomic games. In this context, our work complements previous  models by incorporating heterogeneous local effects in infinite population games. A widely considered infinite population model is that of mean field games as introduced in \cite{lasry:lions:07, huang:caines:malhame:07} which, while focusing on  more general dynamic stochastic interactions, assumes that each agent is influenced by the same aggregate (i.e. the mean) of the whole population.
Another common model is that of population games,  \cite{sandholm2010population}, where a continuum of agents select their strategy among a finite set of options (instead of a continuous set) and the game dynamics are typically described in terms of the total mass of agents playing each strategy.
The behavior of infinite but countable populations has also been studied in aggregative games where each agent is influenced by the same aggregate of the strategies of the rest of the population, as discussed in \cite{kukushkin:04,jensen2010aggregative,acemoglu2013aggregate,cornes2012fully,dubey2006strategic,ma:callaway:hiskens:13,altman2006survey}.
With respect to all these works, graphon games capture settings that include heterogeneous local interactions.

 We finally remark that the idea of using graphons as a support for large population analysis has been successfully applied recently  in different areas such as  community detection in \cite{eldridge2016graphons}, crowd-sourcing in \cite{Lee2017}, signal processing in~\cite{Morency2017}  and  optimal control of dynamical systems in \cite{caines2017graphons}.  The concurrent work by \cite{caines2018graphon} suggests the use of graphons to extend the setup of mean-field games (which differently from network games are dynamic and stochastic games) to heterogeneous settings. Moreover,  the idea of interpreting observed graphs as random realizations from an underlying random graph model has recently been used in the study of centrality measures in \cite{Dasaratha2017} for stochastic block models and in \cite{graphons} for  graphon models. The authors of these papers study among others Bonacich centrality, which  is known to coincide with the equilibrium of a specific type of network games with scalar nonnegative strategies,  quadratic payoff functions and strategic complements.  

\subsection{Organization}

 {\blue The rest of the paper is organized as follows. In Section \ref{graphon_games} we introduce graphon games, we define the graphon equilibrium and we study its properties. In Section \ref{sec:network_formation} we formalize the notion of   network games sampled from a graphon and
in Section~\ref{sec:asymptotic} we investigate the relation between the equilibria of such sampled network games and  graphon games. In Section \ref{sec:ext} we extend our theory to directed and sparser networks and we discuss normalization of the network aggregate by agent's degree instead of population size. In Section \ref{sec:int} we turn to targeted interventions and we study optimality and computability of interventions based on graphon information. In Section \ref{sec:case_study} we present a case study illustrating our approach from data acquisition to  design of interventions.
 Finally, Section \ref{conc} concludes the paper and presents a number of future directions. Appendix \ref{app:graphon_games} presents an equivalent reformulation of the graphon equilibrium as a fixed point of a best response operator and studies the properties of such operator, as needed to prove the results of Section \ref{graphon_games}. 
Appendix \ref{appC} and online Appendix \ref{aux_all} contain omitted proofs and auxiliary lemmas.  In online Appendix \ref{bayes} we extend our results to incomplete information and in online Appendix \ref{sec:id} we briefly comment on identification of unknown payoffs parameters (such as peer effect)  based on graphon information. For simplicity of exposition in the main text we consider games with scalar strategies, all the proofs in the Appendix are instead provided for the vector case. A summary of  notation is provided at the beginning of the Appendix.}

 \ssection{\textbf{Graphon games}}
\label{graphon_games}

{\color{blue} We start by recalling the definition of network games for a finite number of agents. We then show how this concept can be extended to a continuum of agents by introducing the new class of graphon games. We define an equilibrium notion for  graphon games and analyze its existence, uniqueness and continuity properties. }

\subsection{\textbf{Recap on network games}}
\label{sec:network_games}
{\color{blue} We start by formally defining a network game 
as a game with $N$ agents interacting over a network with adjacency matrix $P^{[N]}\in\mathbb{R}^{N\times N}$, where  $P^{[N]}_{ij}$ denotes the level of interaction between agents $i$ and $j$. For simplicity we assume  that the network is  undirected so that   $P^{[N]}$ is symmetric. The extension to directed networks will be discussed in Section \ref{sec:ext}. } In a network game each agent $i\in\{1,\ldots,N\}$ selects a  strategy {\color{blue} $ s^i\in\R$} in its feasible set $\mathcal{S}^i\subseteq \R$  to \textit{maximize} a  payoff function

\begin{equation}\label{gamey}
\J(s^i,\q^i(s), {\blue \theta^i})
\end{equation}
where {\color{blue} $s:=[s^i]_{i=1}^N\in\R^{N}$},  $\q^i(s):=\frac1N\sum_{j=1}^N [P^{[N]}]_{ij}s^j$ denotes the local aggregate\footnote{In network games typically there is no factor $\frac1N$ in the definition of $z^i(s)$. Since we study the behavior when $N$ changes we find it useful to consider this factor explicitly. {\color{blue} A different normalization in terms of agents degree instead of population size is discussed in Section \ref{sec:ext}.}} computed according to the  network $P^{[N]}$ {\blue and $\theta^i\in\R$ is a parameter modeling heterogeneity in the payoff functions of different agents. For simplicity of exposition in the main text we consider games where both $s^i$ and $\theta^i$ are scalars, the extension to the vector case is immediate (as presented in the Appendix). }  
We denote compactly  a network game with the notation   $\mathcal{G}^{[N]}(\{\mathcal{S}^i\}_{i=1}^N, \J, {\blue \{{\theta}^i\}_{i=1}^N}, P^{[N]})$ and we say ``a network game $\mathcal{G}^{[N]}$ with network $P^{[N]}$'' if we need to stress the role of the network.  

{\color{blue} \begin{example}[Linear quadratic network games]\label{ex:lq}
One of the simplest examples of network games is obtained when agents have scalar non-negative strategies $s^i\in\mathbb{R}_{\ge0}$ and  the payoff $\J$ is linear in the network aggregate $\q^i$ and quadratic in the strategy $s^i$, so that

\begin{equation}\label{costL}
\J(s^i,z^i,{\blue \thetaL^i})=-\frac12(s^i)^2+({\blue \thetaL^i}+\alpha z^i)s^i.
\end{equation}
 The  parameter $\alpha\in\R$ in \eqref{costL} captures how much the local aggregate affects each agent's marginal return, which could either be an increasing (strategic complements) or decreasing (strategic substitutes) function of $z^i$ depending on the sign of $\alpha$. The parameter $\thetaL^i>0$ represents the \textit{standalone marginal return} that does not depend on other's actions. This model has been studied e.g. in \cite{jackson2014games,bramoulle2015games}. The model is homogeneous when $\thetaL^i=\thetaL$ for all agents.
\end{example}}
\subsection{\textbf{Graphon games:  The model}}
\label{sec:graphon_model}
Consider a continuum of agents where  each agent  is indexed by the variable $x\in[0,1]$ instead of the finite index $i\in\{1,\ldots,N\}$ and has a scalar strategy   denoted by $s(x)\in\R$ instead of $s^i\in\R$. As in the finite population case, we assume local constraints of the form $s(x)\in\mathtt{S}(x)$, where $\mathtt{S}(x):[0,1]\rightarrow 2^{\R}$ is  a set-valued function. In finite  network games, each agent computes its best response to the local aggregate $z^i(s):=\frac1N\sum_{j=1}^N P_{ij}^{[N]}s^j$ according to the weights of the underlying graph $P^{[N]}$. In the infinite population case,  the natural mathematical object to describe the network of interactions is a graphon. Mathematically, a graphon is a bounded symmetric  measurable function $W:[0,1]^2 \mapsto [0,1]$.   Graphons have originally been introduced as the limit of a sequence of graphs when the number of nodes tends to infinity \cite{Lovasz2012}. In this sense, $W(x,y)$ can be interpreted as measuring  the level of interaction between two infinitesimal agents $x$ and $y$ belonging to the $[0,1]$ interval, exactly as $P^{[N]}_{ij}$ denotes the level of interaction between agents $i$ and $j$ in $\{1,\ldots,N\}$.  For any graphon $W$, we can then define the local aggregate experienced by  agent $x$ as the ``weighted average'' of the other agents actions according to the graphon:

\begin{align*}
z(x\mid s)&:=\int_0^1 W(x,y)s(y)dy.
\end{align*}

\begin{remark}
\textit{ Note that for graphon games a strategy profile $s:[0,1]\rightarrow \R$ is  a \textit{function}. 
In the following, we require that any strategy profile is square integrable, that is $s(x)\in\leb$, where $\leb$ denotes the space of square integrable functions defined on $[0,1]$.
}

\end{remark}
As in network games, the goal of each agent  in a graphon game is  to select the strategy $s(x)\in\mathtt{S}(x)$ that maximizes  its payoff given by
\begin{equation}
\label{eq:br_naJ}
 \J( s(x),z(x \mid s), {\blue \theta(x)}).
\end{equation}
Similar to network games, we assume that the payoff function of an agent $x$ depends on his strategy $s(x)$, on his local aggregate $z(x\mid s)$ {\blue and on a heterogeneity parameter $\theta(x)$}. Note that such a payoff function  has  the same structural form as in network games.
The  difference in the two setups is the way in which the local aggregate ($z^i(s)$ for network games and $z(x \mid s)$ for graphon games) is evaluated. In a graphon game each agent aims at  computing  its best response to the local aggregate induced by the strategy profile $s$ as follows
\begin{equation}
\label{eq:br_na}
s_{\textup{br}}(x \mid  s):=\arg\max_{\tilde s\in\mathtt{S}(x)} \J(\tilde s,z(x \mid  s),{\blue \theta(x)}).
\end{equation}
Note that such a best response might in general  be set-valued. Moreover,  since there is a continuum of agents, the contribution of agent $x$ to the aggregate $z(x\mid s)$ is negligible.  
Consequently,  the decision variable $\tilde s$ affects only the first argument in the payoff function in~\eqref{eq:br_na}.
We summarize the previous discussion in the following definition.
 \begin{definition}[Graphon game]
\textit{A graphon game $\mathcal{G}$ is defined in terms of a continuum set of agents indexed by $[0,1]$, a graphon $W$, a payoff function $\J$ as  in  \eqref{eq:br_naJ},  and for each agent $x\in[0,1]$ a parameter  {\blue $\theta(x)$} and a strategy set $\mathtt{S}(x)$. }
\end{definition}
In the following, we say ``a graphon game $\mathcal{G}$ with graphon $W$'' if we need to stress the role of the graphon and we explicitly write $\mathcal{G}(\mathtt{S},\J,  {\blue \theta}, W)$ is we want to stress the role of all the game primitives.

\subsection{\textbf{Graphon games: Equilibrium concept}}
 Paralleling the  literature on nonatomic  games  (see e.g., \cite{schmeidler1973equilibrium,khan1986equilibrium,wardrop1900some,smith1979existence}), one can extend the concept of Nash equilibrium to graphon games. 
\begin{definition}[Nash equilibrium]\label{nash}
\textit{A function $\bar s\in \leb$ with associated local aggregate $\bar z(x):=z(x\mid \bar s)=\int_0^1 W(x,y)\bar s(y)dy $ is  a Nash equilibrium for the graphon game $\mathcal{G}(\mathtt{S},\J, {\blue \theta},W)$ if  for all $x\in[0,1]$, we have  $\bar s(x)\in\mathtt{S}(x)$ and   
\begin{align*}
\textstyle \J(\bar s(x),\bar z(x), {\blue \theta(x)})\textstyle \ge  \J(\tilde s,\bar z(x),{\blue \theta(x)}) \mbox{ for all } \tilde s\in\mathtt{S}(x).
\end{align*}}
\end{definition}\

In other words, a function $\bar s$  is  a Nash equilibrium if, for each agent $x$, the strategy $\bar s(x)$ is a best response of that agent to the strategies of the other agents. 
In the rest of the section we study Nash equilibrium properties under the following assumptions.

\begin{assumption}[Payoff]\label{ass:cost}
{\blue
\textit{The function $\J(s,z,\theta)$ in \eqref{eq:br_naJ} is continuously differentiable and strongly concave in $s$ with uniform constant $\alpha_\J$  for each value of $z,\theta$. Moreover,  $\nabla_s \J(s, z,\theta)$ is uniformly Lipschitz in $[z,\theta]$ with constants  $\ell_{\J},\ell_\theta$ for all $s$ meaning that $\|\nabla_s \J(s, z_1,\theta_1)-\nabla_s \J(s, z_2,\theta_2)\|\le \ell_\J\|z_1-z_2\|+\ell_\theta\|\theta_1-\theta_2\|$. For each $x\in[0,1]$ the set $\mathtt{S}(x)$ is convex and closed.  }}

\end{assumption}

{
\blue The  assumption of concave payoffs and convex strategy sets is standard in the game theoretical literature, see e.g. \cite{rosen1965existence}. The  assumption on Lipschitz continuity of 
$\nabla_s \J(s, z,\theta)$ is also  natural and guarantees that the effect of the network aggregate~$z$ and the heterogeneity parameter $\theta$ on the marginal payoff is continuous and bounded. Finally, to guarantee that the strategy at equilibrium will not grow unbounded we  make the following additional assumption.
}

\begin{assumption}[Strategy set]\label{ass:constraint}
\textit{A)  There exists $\hat z$ and $M>0$ such that \\ $\| \arg\max_{\tilde s\in\mathtt{S}(x)} \J(\tilde s,\hat z, {\blue \theta(x)}) \|\le M$ for all $x\in[0,1]$.   B) There exists a compact set $\mathcal{ S}$ such that $\mathtt{S}(x)\subseteq \mathcal{S}$ for all $x\in[0,1]$ so that  $s_\textup{max}:=\max_{s\in\mathcal{S}} \|s\|<\infty$.} \end{assumption}

Assumption \ref{ass:constraint}B) implies Assumption \ref{ass:constraint}A). We consider these assumptions separately since  some of our results hold under the sole Assumption \ref{ass:constraint}A), which is less restrictive.

\subsection{\textbf{Graphon games: Properties of the equilibrium}}
\label{properties}

 To study equilibrium  properties, we    report in Appendix \ref{app:graphon_games} an equivalent characterization of the Nash equilibrium of a graphon game as a  fixed point of a best response operator.   Existence of a Nash equilibrium is then an immediate consequence of Schauder fixed point theorem.
\begin{theorem}[Existence] \label{thm:existence}
 \textit{Suppose that the graphon game $\mathcal{G}(\mathtt{S},\J, {\blue \theta},W)$ satisfies  Assumptions~\ref{ass:cost} and \ref{ass:constraint}B). Then it
 admits at least one Nash equilibrium.}
 \end{theorem}

{\color{blue} Uniqueness on the other hand is not always guaranteed. In fixed point theory it is well known that a sufficient condition for uniqueness is  contractiveness. To study contractiveness properties of the best response operator we need to introduce the so-called graphon operator, see also \cite[Section 7.5]{Lovasz2012}.}

\begin{definition}[Graphon operator]\label{def:g_op}
\textit{ For a given graphon $W$, we define the associated graphon operator $\mathbb{W}$ as the  integral operator $\mathbb{W}:\leb\mapsto \leb$ given by}

$$ f(x) \mapsto (\mathbb{W}f)(x)=\int_0^1W(x,y)f(y)\mathrm{d}y.$$
\end{definition}

Intuitively, the graphon operator plays the same role that the adjacency matrix of a graph plays in network analysis. 
Specifically, the graphon operator $\mathbb{W}$
is a linear operator mapping functions to functions, exactly as the adjacency matrix of a network is a linear operator mapping vectors to vectors. One can then introduce spectral properties of the graphon operator, \cite[Definition 3.7.2]{hutson2005applications}.
\begin{definition}[Eigenvalues and eigenfunctions]\label{def:eigen}
  \textit{  A complex number $\lambda$ is an eigenvalue of the operator $\mathbb{W}$ if there exists a {nonzero} function $\psi \in\leb$, called the eigenfunction, such that}
    
    \begin{equation}\label{E:def_eigenfunction}
    (\mathbb{W}\psi)(x) =\lambda \psi(x).
    \end{equation}
    
\end{definition}
{\color{blue} As summarized in Lemma \ref{lem:wn} in Appendix \ref{app:graphon_games}, all the eigenvalues of the graphon operator $\mathbb{W}$ are real and the operator norm, defined as $\vertiii{\mathbb{W}}:= \sup_{f\in\leb, \|f\|_{L^2}=1} \|\mathbb{W}f\|_{L^2}$ coincides with the largest eigenvalue of  $\mathbb{W}$ which we denote by   $\lambda_{\textup{max}}(\mathbb{W})$. We next show that if $\lambda_{\textup{max}}(\mathbb{W})$ is not too large, as formalized in Assumption \ref{cond}, then the best response operator is a contraction, guaranteeing uniqueness of the graphon equilibrium, as shown in Theorem \ref{thm:unique}. }

\begin{assumption}[Contraction] \label{cond}
\textit{Suppose that 
\[\frac{\ell_\J}{\alpha_\J}\cdot \lambda_{\textup{max}}(\mathbb{W})<1,\]
{\blue where $\ell_\J$ and $\alpha_{\J}$ are Lipschitz constants as defined in Assumption \ref{ass:cost}, while} $\lambda_{\textup{max}}(\mathbb{W})$ is the largest eigenvalue of the graphon operator $\mathbb{W}$.}
\end{assumption}

 {\color{blue}\begin{remark}  Assumption \ref{cond} is similar to  assumptions commonly used to obtain uniqueness in  finite network games, see for example \cite{ballester2006s}, and guarantees that the effect of the neighbors aggregate on an agent's marginal payoff, quantified by $\ell_\J \lambda_{\textup{max}}(\mathbb{W})$ is not too large with respect to  effect of its own strategy, quantified by $\alpha_\J$. The only difference is that while in the network game literature  the effect of the network is captured by the maximum eigenvalue of the finite network $P^{[N]}$,  in the case of graphon games the corresponding role is played by the dominant eigenvalue of the graphon, that is, $\lambda_{max}(\mathbb{W})$. In both cases this quantity captures the maximum amount by which the network/graphon can amplify a unitary vector/function.\footnote{\color{blue} In \cite{parise2019variational} conditions for uniqueness based on different network quantities such as the minimum eigenvalue or the infinity norm (i.e. the maximum degree) are discussed. We believe that a similar analysis is possible and interesting also for graphon games. }
 \end{remark} }

\begin{theorem}[Uniqueness]\label{thm:unique}
 \textit{Suppose that the graphon game $\mathcal{G}(\mathtt{S},\J, {\blue \theta},W)$ satisfies Assumptions \ref{ass:cost}, \ref{ass:constraint}A) and \ref{cond}. Then it 
 admits a unique Nash equilibrium.}
 \end{theorem}
 
 Note that in Theorem \ref{thm:unique},  Assumption \ref{ass:constraint}B) is  not needed. In other words the strategy sets $\mathtt{S}(x)$ do not need to be bounded. This is because  for contraction mappings existence and uniqueness of the fixed point can be guaranteed under the sole assumption that the domain is closed and convex, without the need for compactness.\footnote{On the other hand, Assumption \ref{ass:constraint}A) is needed to guarantee that the best response to any strategy profile in $\leb$ belongs to the same space (i.e., it is square integrable).} 
 
 To illustrate our results we consider the familiar framework of linear quadratic games. 
 
{\blue \begin{example}[Linear quadratic graphon games]\label{lqgg}
Building on  Example \ref{ex:lq}, consider a linear quadratic \textit{graphon} game  where the strategy of each agent is scalar and nonnegative so that $\mathtt{S}(x)=\R_{\ge 0}$ for all $x\in[0,1]$ and the payoff function of  an arbitrary agent playing strategy $s$ and subject to the local aggregate $z$ is  quadratic  in $s$ and linear in $z$
\begin{equation}\label{eq:cost_quadratic}
\J(s,z,\thetaL)=-\frac12 s^2+ s[\alpha z+\thetaL],
\end{equation}
{\blue for $\theta(x)\equiv \thetaL$ and $\alpha,\thetaL$ as defined in Example \ref{ex:lq}}. The best response for each agent $x$ is  given by
\begin{equation}
\label{eq:br_scalar}
s_{\textup{br}}(x \mid  s)=\max\{0,[\alpha z(x\mid  s)+\thetaL]\}.
\end{equation}
It therefore follows that $\J$ satisfies Assumption \ref{ass:cost} with $\alpha_\J=1$, $\ell_\J=|\alpha|$.  Note also that Assumption~\ref{ass:constraint}A) is satisfied (take e.g. $\hat z=0, M=\thetaL$). 
Consequently, by Theorem \ref{thm:unique} a unique graphon equilibrium exists if 
$$|\alpha|<\frac{1}{\lambda_{\textup{max}}(\mathbb{W})},$$
which is a similar condition as the one derived in \cite{ballester2006s} for finite network games. 
If additionally $\alpha>0$, we can immediately see  from \eqref{eq:br_scalar} that  the best response of each agent is an increasing function of the local aggregate $z(x\mid s)$,  i.e., this is a game of strategic complements [\cite{ballester2006s}] and the unique Nash equilibrium  $\bar s $ is  \textit{internal} (i.e., it satisfies $\bar s (x)>0$ for all $x\in[0,1]$). From \eqref{eq:br_scalar}  it then must hold

\begin{align}\label{step_lq}
\bar s(x)=\alpha z(x\mid \bar s)+\thetaL &\quad \Rightarrow  \quad  \bar s(x)    =\alpha(\mathbb{W}\bar s)(x)+\thetaL\\
&\quad \Rightarrow \quad (\mathbb{I} \bar s)(x)    =\alpha(\mathbb{W}\bar s)(x)+\thetaL\quad \Rightarrow \quad  ((\mathbb{I}-\alpha\mathbb{W})\bar s)(x)     =\thetaL {1}_{[0,1]} (x).\notag
\end{align}

The condition $|\alpha| \lambda_{\textup{max}}(\mathbb{W})<1$  implies invertibility of the operator $(\mathbb{I}-\alpha\mathbb{W})$. Hence  
\begin{align}\label{eq:bon}
\bar s(x)     &=\thetaL ((\mathbb{I}-\alpha\mathbb{W})^{-1}{1}_{[0,1]}) (x) = \thetaL \sum_{k=0}^\infty \alpha^k(\mathbb{W}^k{1}_{[0,1]})(x)
\end{align}
which corresponds to the Bonacich centrality of agent $x$ in the graphon $W$, as defined in \cite{graphons}. 
\end{example}}

 Finally, for graphon games satisfying the assumptions of Theorem~\ref{thm:unique},  so that the Nash equilibrium is unique,  we   study the effect of graphon perturbations. {\blue To this end,  for any linear integral operator $\mathbb{O}$, we denote by $\vertiii{\mathbb{O}}:=\max_{f\in\leb} \|\mathbb{O}f\|_{L^2}$  its operator norm.}

\begin{theorem}[{\blue Continuity}]\label{thm:comp}
\textit{ Suppose that the graphon game $\mathcal{G}(\mathtt{S},\J, {\blue \theta},W)$ satisfies Assumptions \ref{ass:cost}, \ref{ass:constraint}B), \ref{cond} and let $\bar s$ be its unique Nash equilibrium. Consider a perturbed graphon $\tilde W$, {\blue a perturbed function $\tilde \theta$} and let $\tilde s$ be any Nash equilibrium of the graphon game $\mathcal{G}(\mathtt{S},\J, {\blue \tilde \theta},\tilde W)$. Then it holds 
\begin{equation}\label{eq:K}
\blue \|\bar s -\tilde s\|_{L^2} \le \frac{1/\alpha_\J }{1-\ell_\J/\alpha_\J \lambda_{\textup{max}}(\mathbb{W})}\left(\ell_\J\vertiii{\mathbb{W} -\tilde{\mathbb{W}}}s_\textup{max} +\ell_\theta \|\theta-\tilde \theta\|_{L^2}\right).
\end{equation}}
\end{theorem}

 The result in Theorem \ref{thm:comp}, besides  being  of  interest on its own,  is   fundamental for the finite population analysis performed in the next sections.

 \ssection{\textbf{Sampled network games: Definition and examples }}
\label{sec:network_formation}

{\color{blue} Graphon games describe strategic interactions among a continuum of agents. In this section we show how one can sample finite networks from a graphon and define  \textit{sampled network games}.  In the next section we will then study the relation between equilibria of sampled network games and graphon games. }

\subsection{\textbf{Graphons as a stochastic network formation model}}
 \label{step1}

 In the next definition, we illustrate how  a graphon can be used to describe a probability distribution over the space of networks and how one can sample from this distribution to construct a sampled network, see Figure \ref{fig:example},  \cite[Chapter 10]{Lovasz2012}.

\begin{figure}
\begin{center}
a) \hspace{0.4cm} $W(x,y)$ \hspace{2.6cm} b) \hspace{0.4cm} $P_w^{[5]}$   \hspace{0.8cm} \hspace{2.4cm} c) \hspace{0.4cm}$P^{[5]}_s$ \hspace{0.9cm}  \\[0.4cm]
\includegraphics[height=3.4cm]{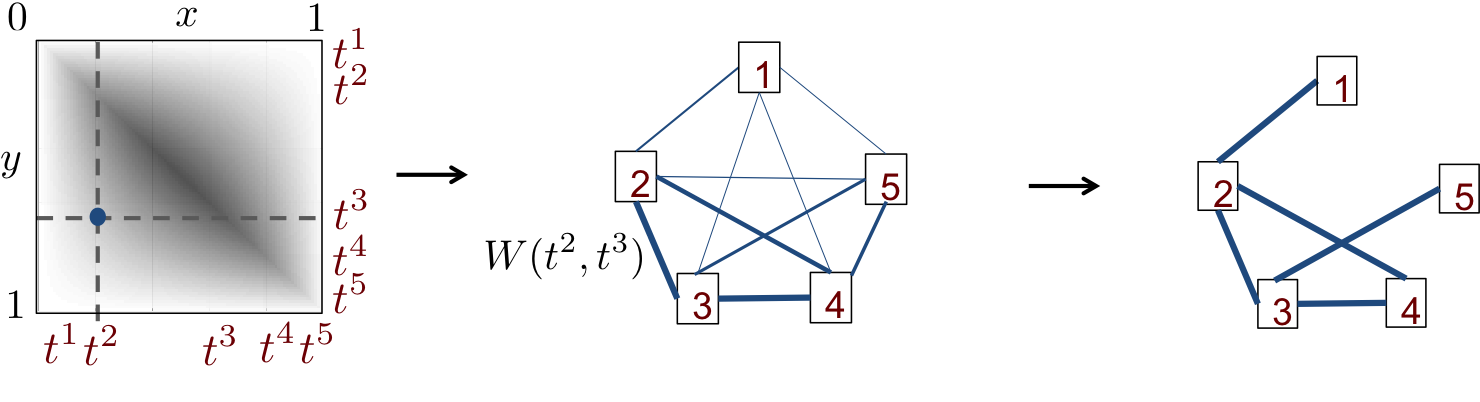}
\end{center}
\caption{Illustration of the sampling procedure described in Definition \ref{sample} for $N=5$. a) The  graphon. b) The \textit{weighted adjacency matrix}  $P_w^{[5]}$ associated with  the random sample $[t^1,\ldots,t^5]=[0.03, 0.31,0.69,0.82,0.95]$. c)  A realization of the \textit{$0$-$1$ adjacency matrix}  $P_s^{[5]}$.  For the graphon a linear grayscale colormap is used with  white associated to  $W=0$ and black to $W=1$. For $P_w^{[5]}$ the width of the line is proportional to the weight of the corresponding edge. In $P_s^{[5]}$ any edge has weight $1$.}
\label{fig:example}
\end{figure}

\begin{definition}[Sampling procedure]\label{sample}
\textit{ Given any graphon $W$ and any desired number $N$ of nodes, uniformly and independently sample $N$ points $\{\type^i\}_{i=1}^N$ from $[0,1]$ and 
define a \textit{weighted adjacency matrix} $P_w^{[N]}$ as follows}

$$[P_w^{[N]}]_{ij}=\begin{cases} W(\type^i,\type^j), & \mbox{if}\ i\neq j, \\ 0, &  \mbox{if}\ i= j.  \end{cases}$$
\textit{Starting from $P_w^{[N]}$,   define the \textit{$0$-$1$ adjacency matrix} $P_s^{[N]}$ as the adjacency matrix corresponding to a graph with $N$ nodes obtained  by  randomly connecting nodes $i,j \in[1,N]$ with Bernoulli probability   $[P_w^{[N]}]_{ij}$.}
\end{definition}\
\begin{remark}\label{order}
The random points $\{t^i\}_{i=1}^N$ can be interpreted as  \textit{agents types}  (e.g, an agent's type may represents the community to which the agent belongs or its geographical  location, as discussed in the following Examples 3 and 4). The   graphon value $W(t^i,t^j)$ is then encoding information about the level of interaction  between two arbitrary agents of type $t^i$ and $t^j$.
From here on we are going to assume that the  $\{\type^i\}_{i=1}^N$ are ordered such that $\type^i\le \type^{i+1}$ for all $i\in\{1,\ldots, N-1\}$. This is without loss of generality, since it simply corresponds to a relabeling of the nodes.  Figure \ref{fig:example} illustrates the sampling procedure described in Definition~\ref{sample}. 
Note that both $P^{[N]}_{w}$ and $P^{[N]}_{s}$ are stochastic matrices. The difference between the two is that $P^{[N]}_{w}\in[0,1]^{N\times N}$ while $P^{[N]}_{s}\in\{0,1\}^{N\times N}$. {\color{blue} Finally note that an agent of type $t^i$ has an expected number of neighbors that grows as $N\int_0^1 W(t^i,t^j) dt^j$. Hence networks sampled according to Definition \ref{sample} are dense. The generalization to sparser networks is discussed in Section~\ref{sec:sparse}. }
\end{remark}

To develop more intuition on the framework of graphons and its connection to other well-known stochastic network formation processes we start by noting that 
for any $p\in[0,1]$, the constant graphon $W(x, y) \equiv p$ coincides with the  Erd\H{o}s-R\'enyi random graph model where each pair of agents in connected with  probability $p$. In the next example, we  show how graphons can be used to encode stochastic block models, which can be seen as an extension of  
Erd\H{o}s-R\'enyi  models to a setting with finitely many communities.

\begin{example}[label=ex:sbm] \textup{{(Community structure)}}
Consider networks where   agents are divided into $K$ communities and let $\pi_k$  be the probability that a random agent belongs to community $k$, with $\sum_{k=1}^K \pi_k=1$. 
Additionally,  assume that agents belonging to the same community form a link  with Bernoulli probability $g_{\textup{in}}$ while agents from  different communities  form a link  with  probability $g_{\textup{out}}$ (typically smaller than $g_{\textup{in}}$).\footnote{{\blue The parameters $\pi_k$ are exogenous and model the probability that  agents are born with type $k$, e.g. male or female.}  {\blue The exogenous parameters $g_{in}$ and $g_{out}$ are instead a result of the different costs borne by each agent when forming a link to someone from the same and from the other community (see for example \cite{jackson2005economics}).}} To generate such a community structure  from a graphon, one can partition  $[0,1]$ into $K$ disjoint intervals $\{\mathcal{C}_k\}_{k=1}^K$, with $|\mathcal{C}_k|=\pi_k$, and use the piecewise constant graphon

$$W_{\textup{SBM}}(x,y)=\begin{cases} g_{\textup{in}} & \mbox{if there exists } k \mbox{ s.t. } x\in \mathcal{C}_k ,y\in \mathcal{C}_k, \\ g_{\textup{out}} & \mbox{otherwise}. \end{cases}$$

 We  denote this graphon with the label ``SBM'' because of its relation to Stochastic Block Models.
Figure \ref{fig:graphon} (left) illustrates  an SBM graphon of this type with $K=2$ communities (e.g. red and blue agents) of   size $[w_1,w_2]=[0.75, 0.25]$ and with $g_{\textup{in}}=0.8$, $g_{\textup{out}}=0.1$. In this case, we selected $\mathcal{C}_1=[0, 0.75]$, $\mathcal{C}_2=(0.75, 1]$.

\end{example}

 In the previous example agents are partitioned into a finite number of different communities. Graphons can also be used to model processes  where agents types may take infinitely many values.  
  The next example illustrates one such case where an agent's type is given by its location.

\begin{example}[label=ex:minmax] \textup{{(Location model)}}
Consider a  model where  $N$ agents  are independently located uniformly at random along a line segment represented by the interval [0,1] (e.g., homeowners along a street) and
 assume that the level of interaction between  agent $i$ and $j$  is  a decreasing function of  their spatial distance, {\blue capturing the natural observation that the cost of forming links increases with  agents geographical distance, as motivated in \cite{johnson2003spatial}.} This type of interaction  can be represented for example by using the  ``minmax'' graphon  

$$W_{\textup{MM}}(x,y)=\min(x,y)(1-\max(x,y)),$$
where $x\in[0,1]$ denotes the agents position along the line, see Figure~\ref{fig:graphon} (right). 
\begin{figure}
\begin{center}
\includegraphics[width=0.18\textwidth]{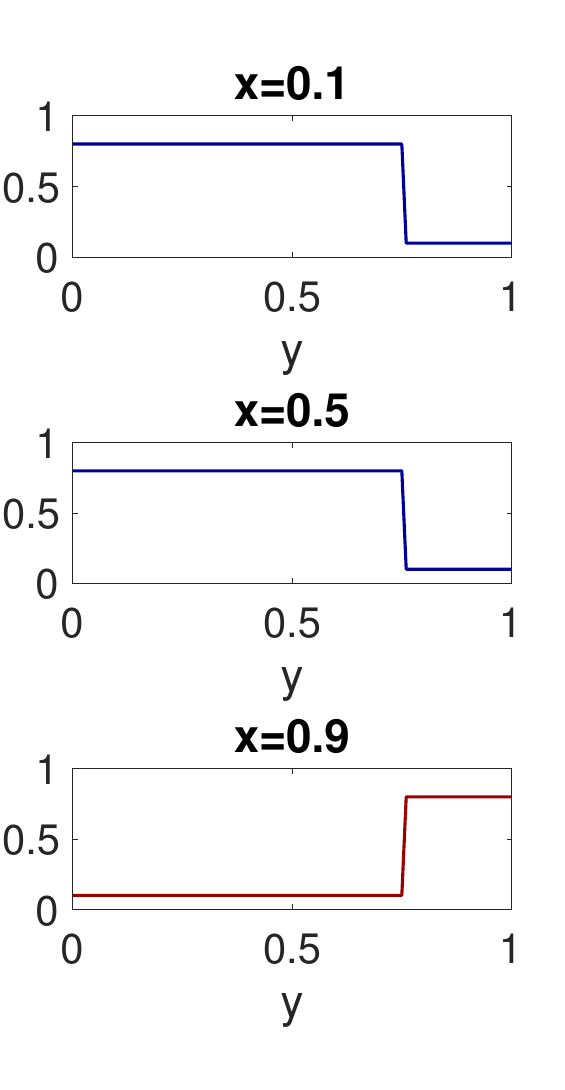} \quad
\includegraphics[width=0.23\textwidth]{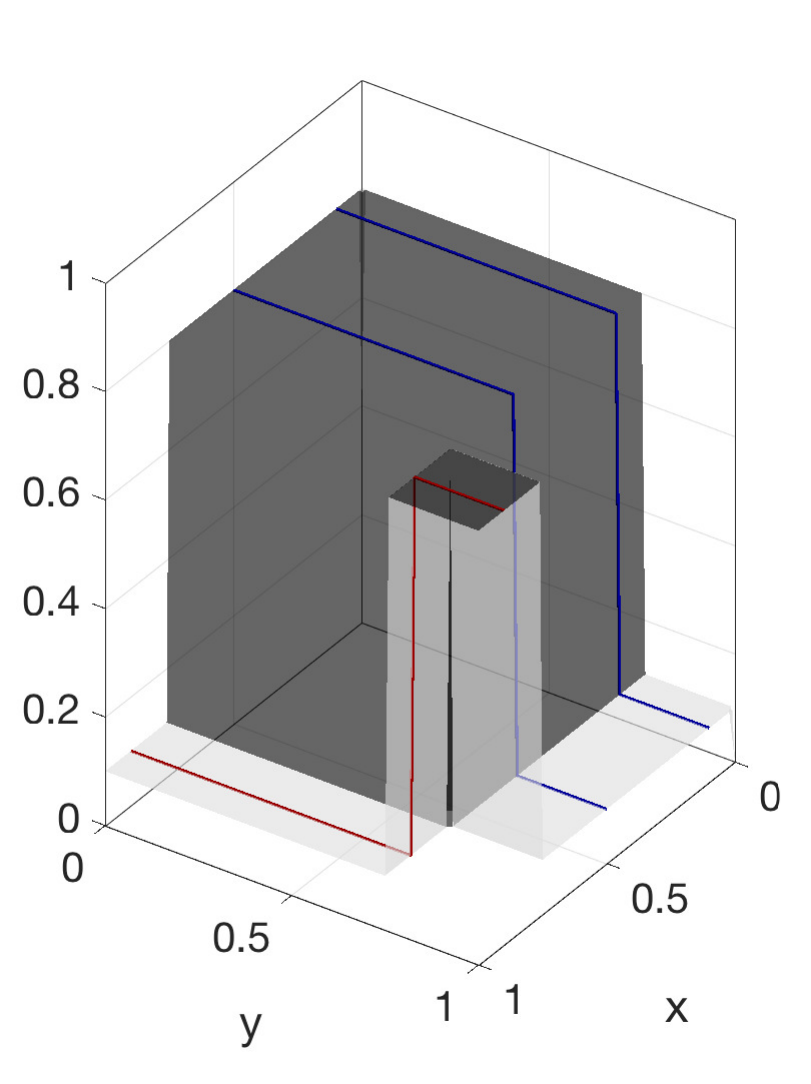} \qquad
\includegraphics[width=0.18\textwidth]{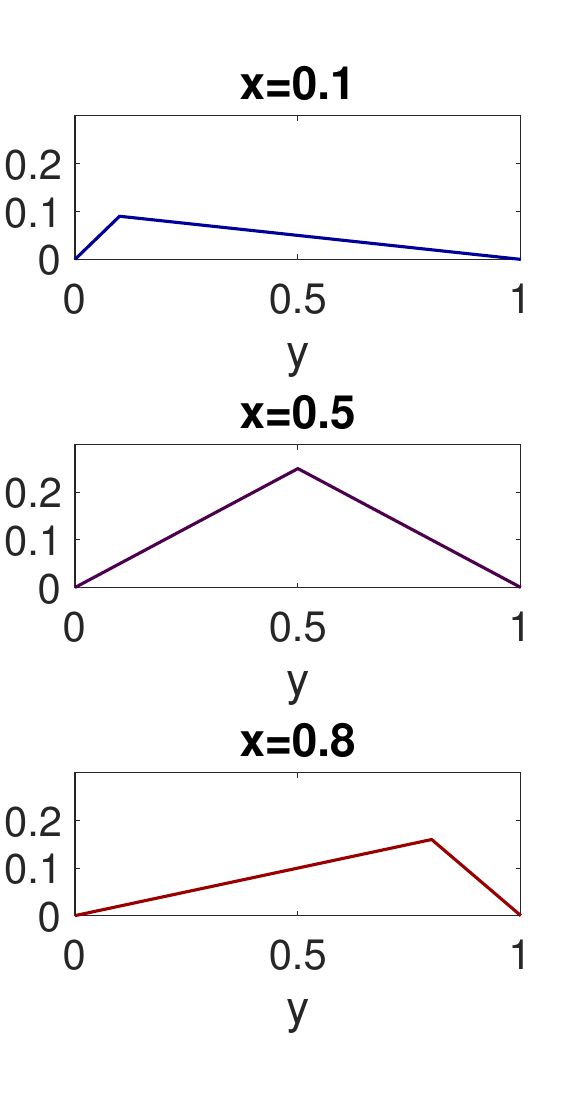} \quad
\includegraphics[width=0.23\textwidth]{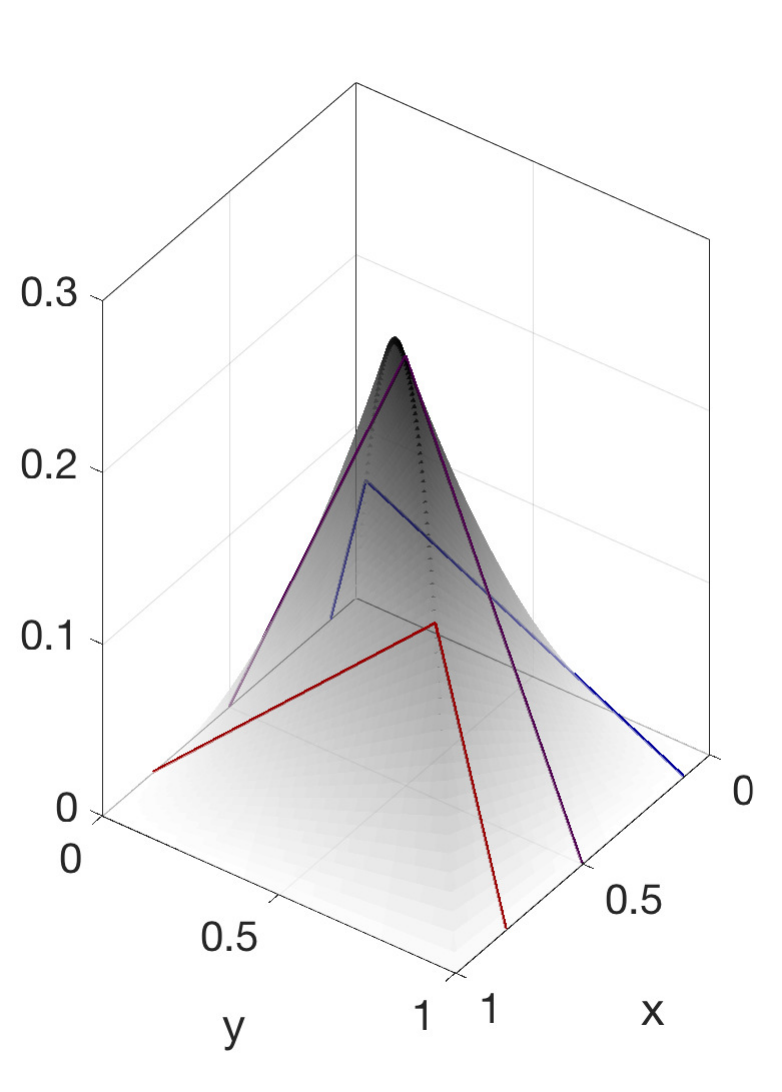} 
\end{center}
\caption{Left:   The SBM graphon of Example \ref{ex:sbm}. Note that the agents in position $x=0.1$ and $x=0.5$ belong to the same community and are thus affected in the same way by the rest of the agents. Right: The minmax graphon of Example \ref{ex:minmax}. In both cases $W(x,\cdot)$ for three different values of $x$ and $W(x,y)$ as a function of both $x$ and $y$ are shown. A linear grayscale colormap is used with  white associated to  $W=0$ and black to $W=1$.}
\label{fig:graphon}
\end{figure}
\end{example}

 \subsection{\textbf{ Sampled network games }}
 \label{step2}
 
We here specialize the definition of network games introduced in Section \ref{sec:network_games}  to games where the network of interactions is sampled from a graphon. Intuitively, we define a sampled network game as a game where $N$ agents of type $\{\type^i\}_{i=1}^N$, randomly sampled in $[0,1]$, interact over a network  formed according to the process described in Definition~\ref{sample}. Note that, we consider both games played over  the weighted adjacency matrix  $P^{[N]}_{w}\in[0,1]^{N\times N}$ and the $0$-$1$ adjacency matrix  $P^{[N]}_{s}\in\{0,1\}^{N\times N}$ and we use the symbol $P^{[N]}_{w/s}$ for statements that hold in both cases. 
 
 \begin{definition}[Sampled network game]\label{sng}
 Given a graphon $W$, a payoff function $U$, a set valued function $\mathtt{S}$ {\blue and a parameter function $\theta$}, we define a sampled network game among $N$ agents of type $\{\type^i\}_{i=1}^N$  as $\mathcal{G}^{[N]}(\{\mathtt{S}(\type^i)\}_{i=1}^N,\J,{\blue \{{\theta}(t^i)\}_{i=1}^N}, P_{w/s}^{[N]})$, where the types $\{\type^i\}_{i=1}^N$ are sampled uniformly and independently at random from $[0,1]$ and $P_{w/s}^{[N]}$ is as in  Definition~\ref{sample}.  
 \end{definition}

{\color{blue} Figure \ref{fig:example1} and \ref{fig:example2} show the equilibria of three realizations of  sampled network games with LQ payoffs as in Example \ref{ex:lq}, when the networks are sampled from the graphons described in Example \ref{ex:sbm} and \ref{ex:minmax}, for different values of $N$.} {\color{blue} In both examples, one can observe similarities between equilibria of different sampled network games. For instance in Example \ref{ex:sbm} red agents tend to exert lower efforts at equilibrium than  blue agents, while in Example \ref{ex:minmax} agents at more central locations exert higher efforts at equilibrium.  This trend becomes sharper and ``more deterministic'' as the population size increases. In the next section we formalize these observations by showing that equilibria of sampled network games converge to the equilibrium of the corresponding graphon game, as defined in Section \ref{graphon_games}.}

\begin{figure}
\begin{center}
 \includegraphics[width=4.8cm, angle =180]{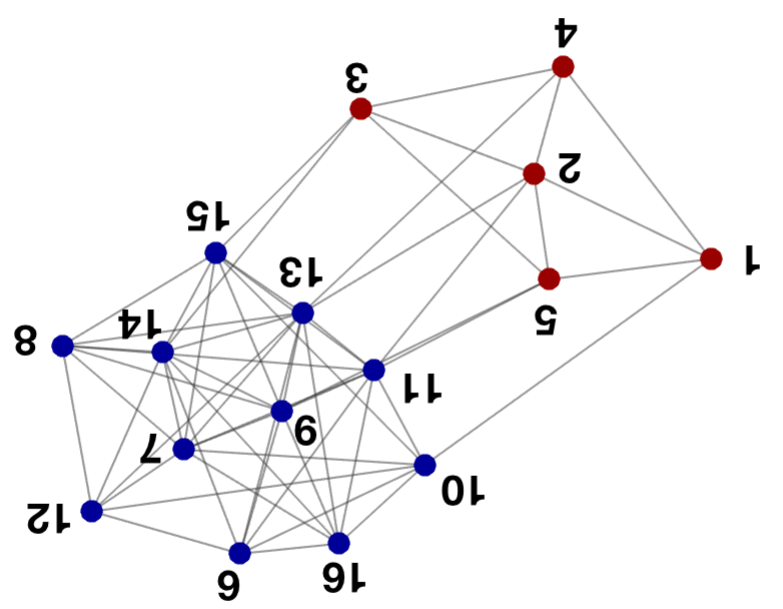}
\includegraphics[width=5cm, angle =180]{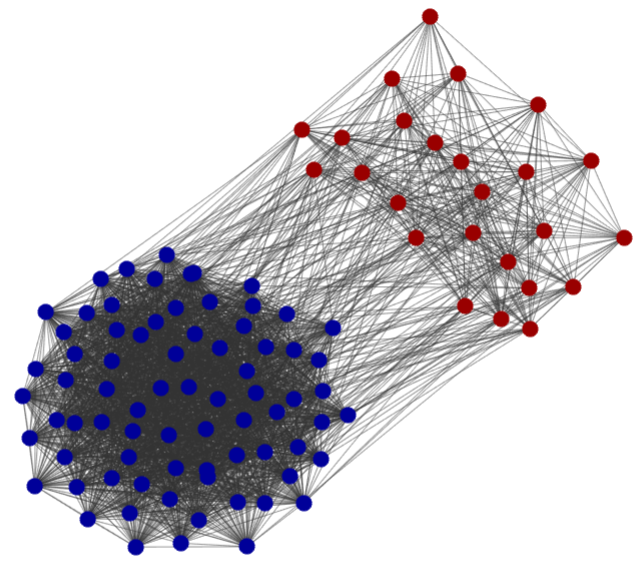}
\includegraphics[width=5cm, angle =180]{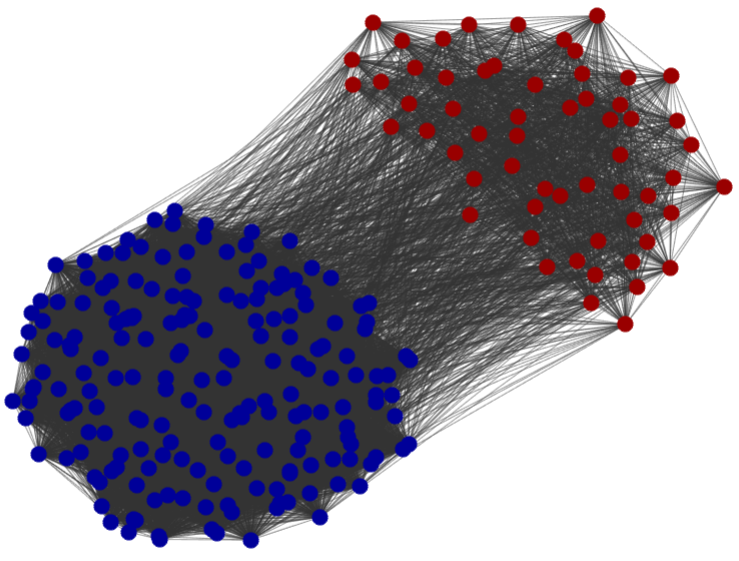}\\
\includegraphics[width=4.8cm]{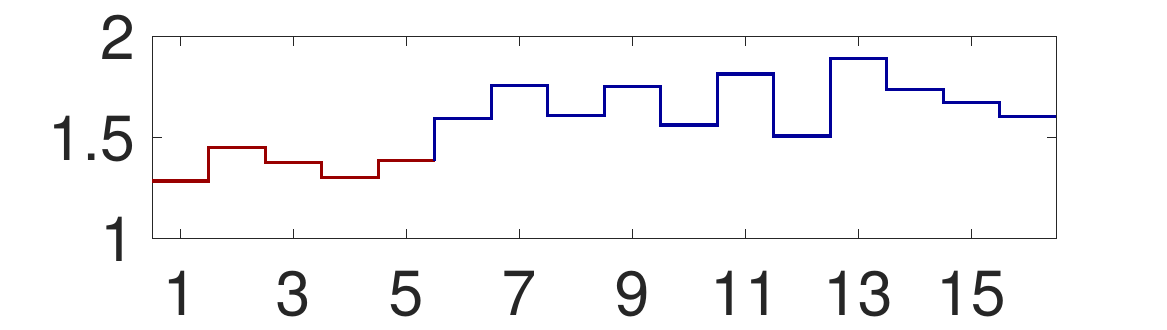}
\includegraphics[width=4.8cm]{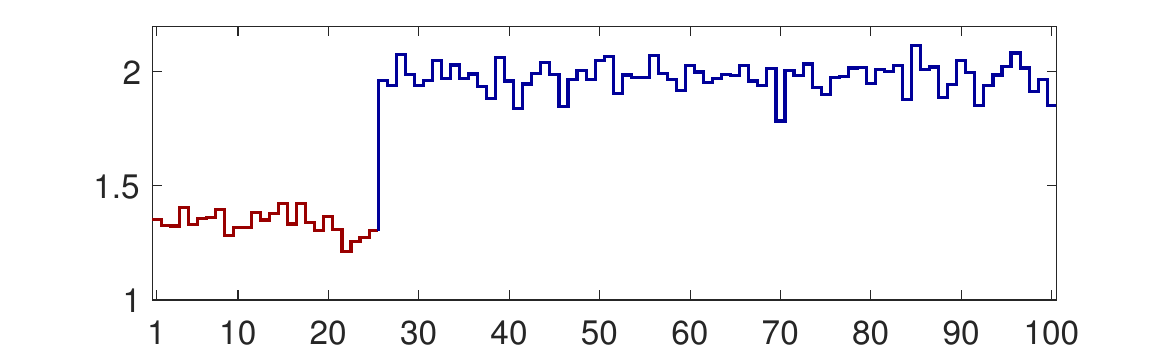}
\includegraphics[width=4.8cm]{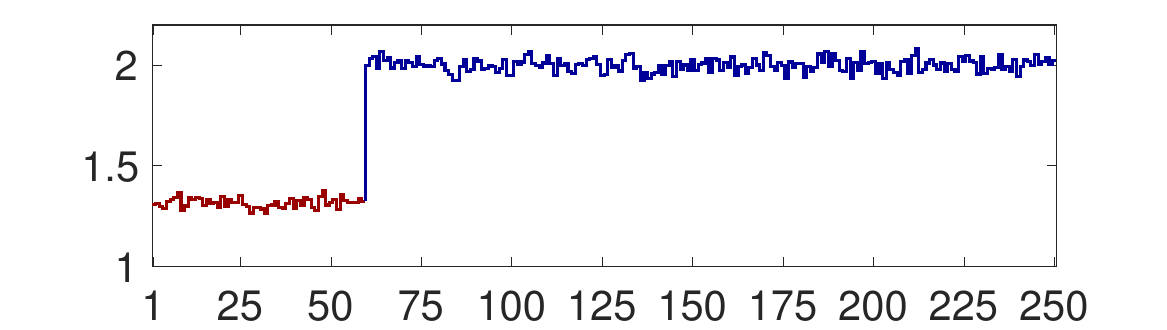}
\end{center}
\caption{ Three realizations of networks formed according to the  two community model described in  Example~\ref{ex:sbm} (with $\pi_{\textup{red}}=0.25$, $\pi_{\textup{blue}}=0.75$, $g_{in}=0.8$ and $g_{out}=0.1$) for $N=10,100,250$  and their corresponding equilibria (for payoff as in \eqref{costL} with $\alpha=0.8$, $\thetaL=1$).}
\label{fig:example1}
\end{figure}

 \begin{figure}
\begin{center}
\includegraphics[width=4.8cm]{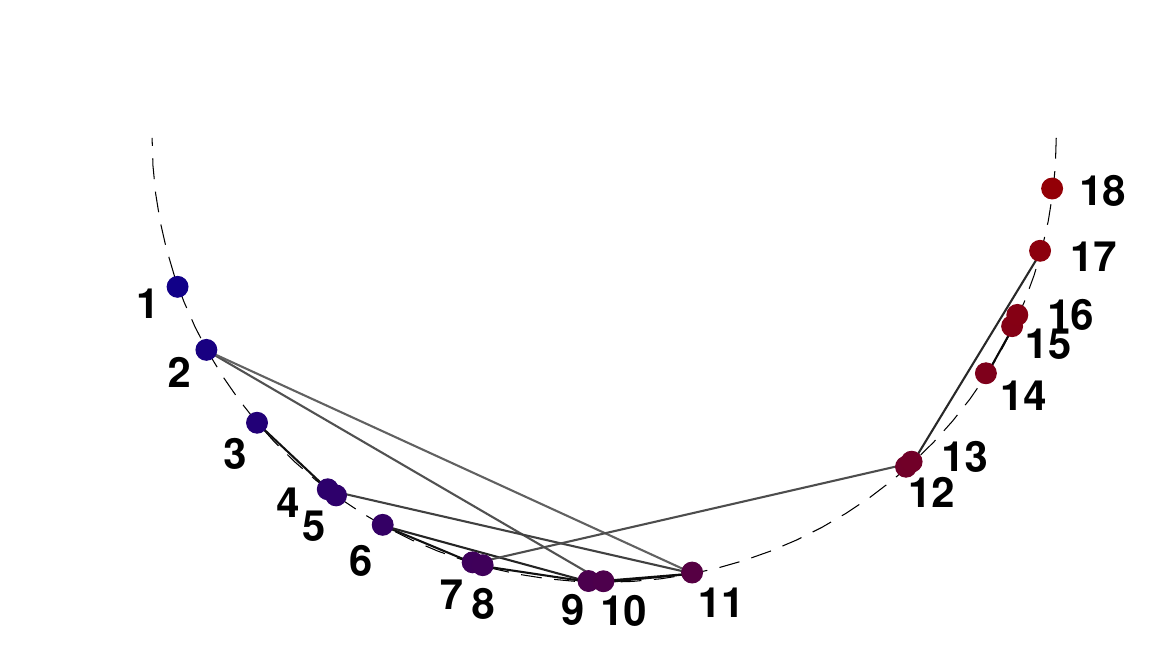}
\includegraphics[width=4.8cm]{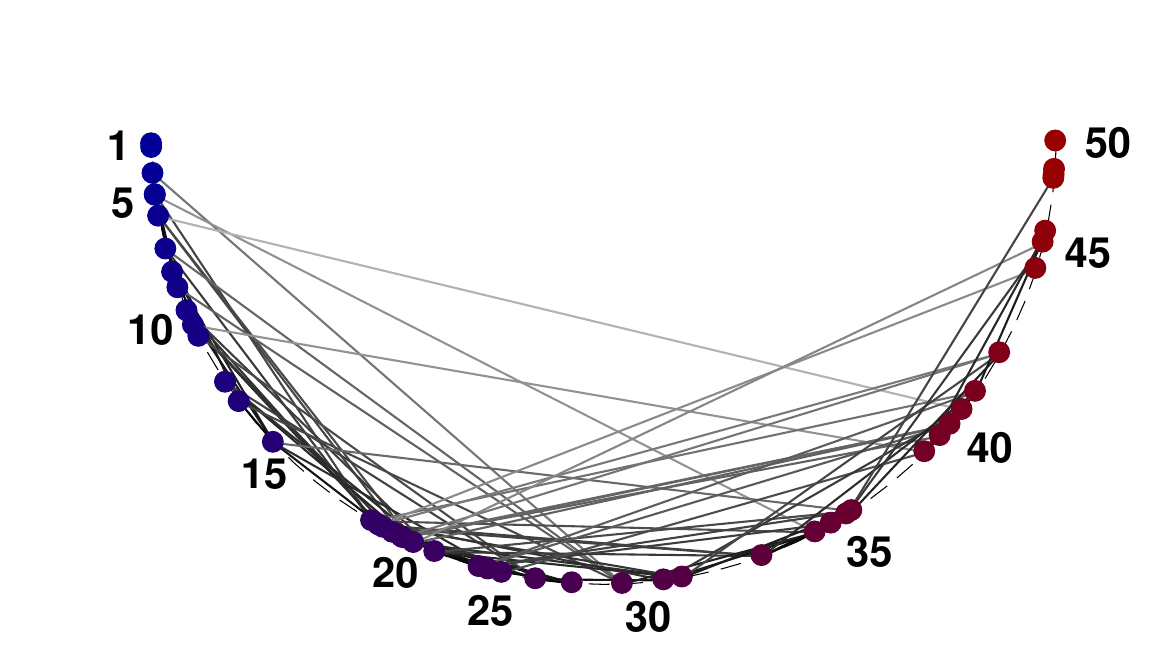}
\includegraphics[width=4.8cm]{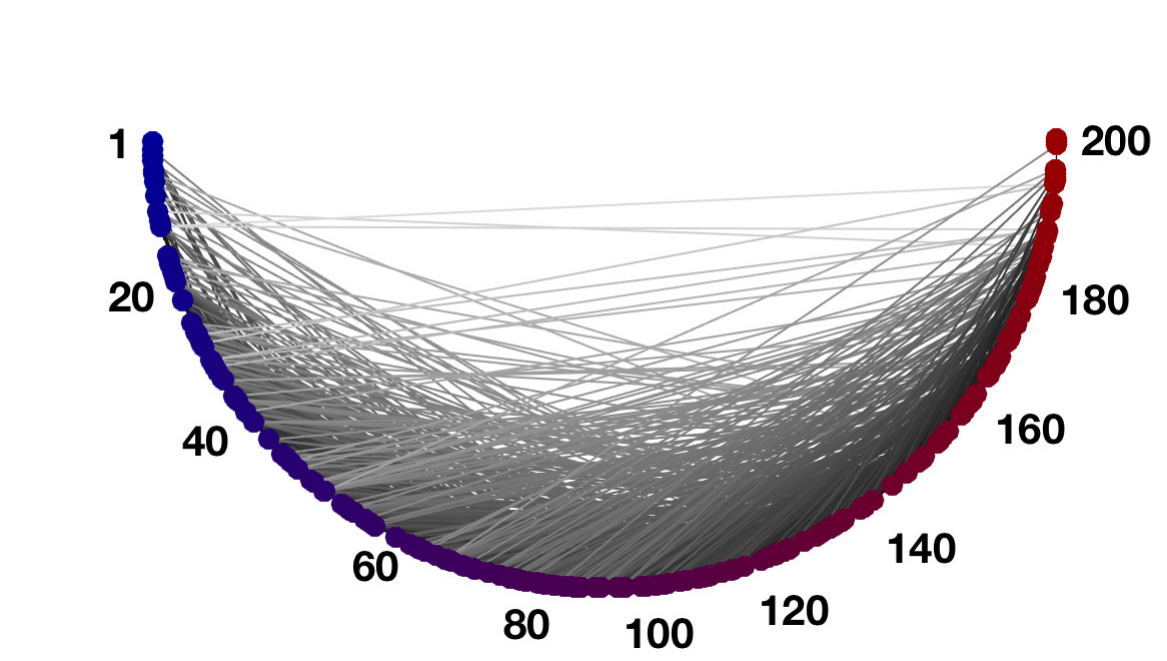}\\
\includegraphics[width=4.8cm]{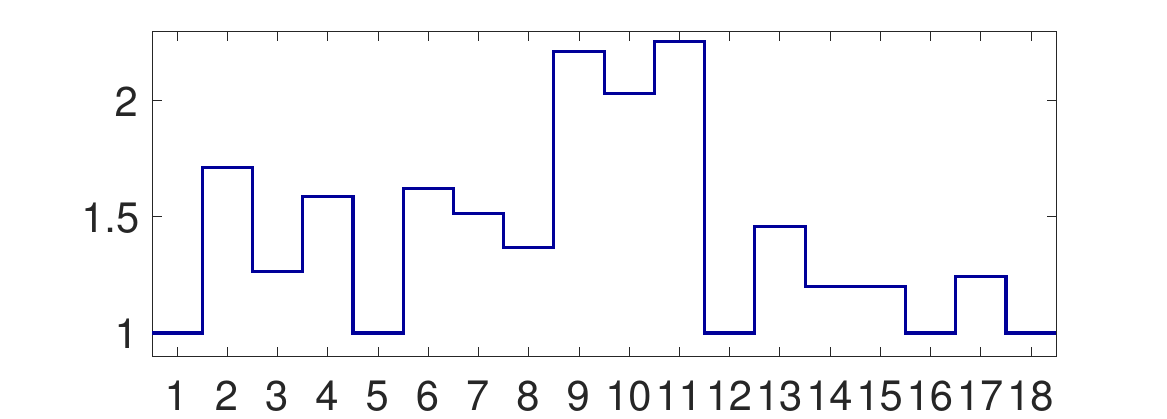}
\includegraphics[width=4.8cm]{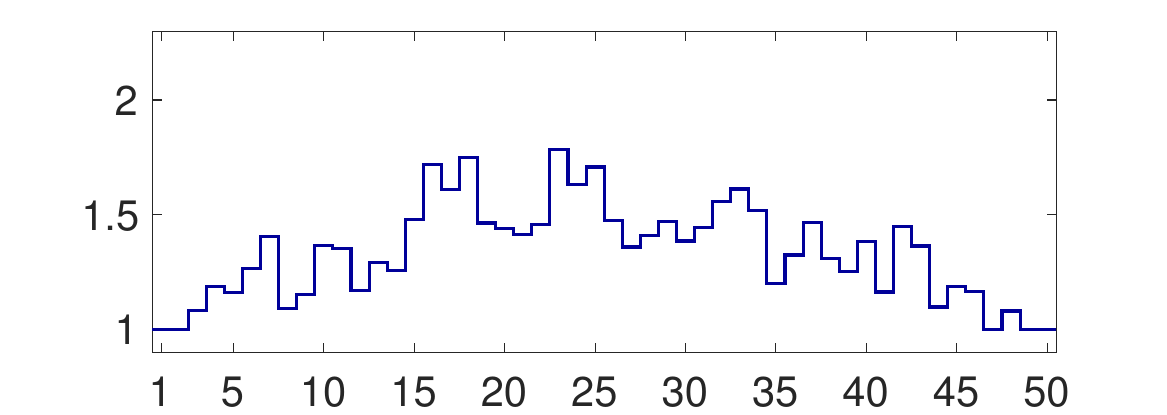}
\includegraphics[width=4.8cm]{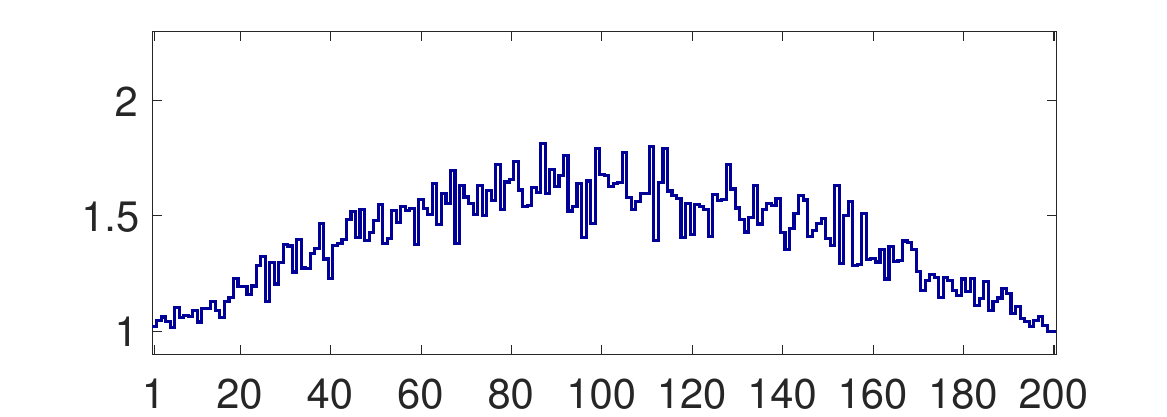}
\end{center}
\caption{ Three realizations of networks formed according to the location model described in Example \ref{ex:minmax} and their corresponding equilibria (for  payoff as in \eqref{costL} with $\alpha=3$ and $\thetaL=1$). The line along which agents are located is represented as a semicircle for simplicity of visualization. The color of the nodes is  associated to the agents location along the line (blue being one extreme and red the other extreme). Edges between agents that are further apart are in lighter color.}
\label{fig:example2}
\end{figure}

\ssection{\textbf{Sampled network games: Convergence analysis}}\label{sec:asymptotic}

\subsection{\textbf{Network games are graphon games}}
\label{net_graph}
 We start our analysis by showing that any network game   can be equivalently reformulated as a graphon game.  In  network games Nash equilibria  are vectors of $\R^{N}$  while in  graphon games they   are functions of $\leb$. To compare these two objects, we define a one-to-one correspondence between vector and functions using a   uniform partition $\mathcal{U}^{[N]}=\{\mathcal{U}_1^{[N]},\mathcal{U}_2^{[N]},\ldots,\mathcal{U}_N^{[N]}\}$ of $[0,1]$ obtained by setting $\mathcal{U}_k^{[N]}=[\frac{k-1}{N},\frac kN),$ $k\in\{1,\ldots,N-1\}
$ and $\mathcal{U}^{[N]}_N=[\frac{N-1}{N},1]$. Intuitively, we are going to pair each agent $i$ in a finite network game with the interval $\mathcal{U}^{[N]}_i$. For any $N\in\N$  we can then define   the  \textit{step function equilibrium} $\bar s_{[N]}(x)\in\leb$ corresponding to any equilibrium $\bar s_{[N]}\in\R^{N}$  of a network game as follows
\[\bar s_{[N]}(x):= \bar s_{[N]}^i, \quad \forall x\in \mathcal{U}^{[N]}_i, \ \forall  i\in\{1,\ldots,N\}.\] 
By exploiting this reformulation we can compare the Nash equilibria of graphon and network games (or of network games with different population sizes) by working in the $\leb$ domain. 
Similarly, the   uniform partition $\mathcal{U}^{[N]}$ can be used to define a one-to-one correspondence between  any graph $P^{[N]}\in\R^{N\times N}$ and a corresponding \textit{step function graphon} $W^{[N]}$  obtained by setting 

\begin{equation}\label{Wcor}
W^{[N]}(x,y):= P^{[N]}_{ij}, \quad \forall (x,y)\in \mathcal{U}^{[N]}_i\times \mathcal{U}^{[N]}_j,\quad \forall i,j\in\{1,\ldots,N\}.
\end{equation}

The following theorem shows that  the step function equilibria of any network game with graph $P^{[N]}$ coincide with  the Nash equilibria of the graphon game with step function graphon $W^{[N]}$ corresponding to $P^{[N]}$.

\begin{theorem}\label{thm:oneNash}
\blue
\textit{A vector $\bar s_{[N]}\in \R^{N}$ is a Nash equilibrium of  $\mathcal{G}^{[N]}(\{\mathcal{S}^i\}_{i=1}^N, \J,\{\theta^i\}_{i=1}^N, P^{[N]})$  if and only if the  corresponding  step function equilibrium $\bar s_{[N]}(x)\in \leb$ is a  Nash equilibrium  of the graphon game  $\mathcal{G}(\mathtt{S}^{[N]}, \J, \theta^{[N]}, W^{[N]})$ with   payoff function as in \eqref{eq:br_naJ}, set valued function  $\mathtt{S}^{[N]}(x):=\mathcal{S}^i$ for all $x\in \mathcal{U}_i^{[N]}$, parameter function $\theta^{[N]}(x):=\theta^i$ for all $x\in \mathcal{U}_i^{[N]}$ and step function graphon $W^{[N]}$  corresponding to  $P^{[N]}$.}
\end{theorem}

\subsection{\textbf{Equilibria in sampled network games}}

Our  next result relates the Nash equilibria of sampled network games, as introduced in Section \ref{step2}, to the equilibrium of the corresponding graphon game.
 Specifically, in the following theorem, we derive a bound on the distance between such  equilibria that holds for any graphon satisfying Assumption \ref{cond} and the following additional regularity condition.\footnote{ A more general result that requires only Assumption \ref{cond} is given in \cite{parise2018graphon}. We here focus on Lipschitz graphons  to obtain simpler bounds.}
\begin{assumption}[Lipschitz continuity]\label{lipschitz}
\textit{There exists a constant $L>0$ and a sequence of non-overlapping intervals  $\mathcal{I}_k=[\omega_{k-1},\omega_k)$ defined by $0=\omega_0< \dots < \omega_{\Omega+1} =1$, for a (finite) $\Omega\in\mathbb{N}$ and $k\in\{1,\ldots,\Omega+1\}$, such that for any $k,l\in\{1,\ldots, \Omega+1\}$, any set $\mathcal{I}_{kl}=\mathcal{I}_k\times \mathcal{I}_l$ and pairs $(x,y)\in \mathcal{I}_{kl}$, $(x',y')\in \mathcal{I}_{kl}$ we have that }
\[|W(x,y)-W(x',y')|\leq L(|x-x'|+|y-y'|).\]
{\blue Moreover, 
$|\theta(x)-\theta(x')|\leq L|x-x'|$ for any $x,x'\in \mathcal{I}_k$ and if $\Omega>0$ there exists $\theta_\textup{max}$ such that $\|\theta(x)\|\le \theta_\textup{max}$ for all $x\in[0,1]$.
}
\end{assumption}

Assumption \ref{lipschitz} implies that the  graphon $W$ is piecewise Lipschitz (over the intervals $\mathcal{I}_k\times \mathcal{I}_l$) which  is a common assumption  in the context of graphon estimation, see e.g. \cite{airoldi2013stochastic}, {\blue and that the parameter function is piecewise Lipschitz (over the intervals $\mathcal{I}_k$)}. We note that both the minmax graphon  and any SBM graphon satisfy this assumption. 

{\blue Since the networks are sampled randomly from the graphon, our statements on convergence of equilibria of sampled network games to equilibria of the corresponding graphon game hold in probability.  One can choose the desired probability level, which we denote by $1-2\delta_N$ for a population of size $N$, by defining  an admissible confidence sequence  $\{\delta_N\}_{N=1}^\infty$ as follows.
\begin{definition}[Admissible confidence sequence]
A sequence $\{\delta_N\}_{N=1}^\infty$ is admissible if it is such that $\delta_N \le e^{-1}$ and $\frac{\log( N/\delta_N)}{N} \rightarrow 0$. 
\end{definition}

 \begin{remark}\label{rem:bound} {\blue In general we will be interested in sequences $\delta_N\rightarrow 0$ (so that the probability $1-2\delta_N$ converges to one for large $N$). Hence the requirement $\delta_N \le e^{-1}$  is without loss of generality.  On the other hand, we need to impose that $\delta_N$ does not converge to zero too fast, since we will use  matrix inequalities to bound the distance between a random matrix and its expectation by a quantity that depends on  $\frac{\log( N/\delta_N)}{N}$  and we want this bound to converge to zero for large populations. }
\textit{ To meet this second requirement, one can for example select constant confidence $\delta_N\equiv \delta\in(0,e^{-1})$ or polynomial confidence $\delta_N=\frac{1}{N^k}$ for any $k>0$, since  for $N$ large enough  $\frac{1}{N^k}\le e^{-1}$ and $\frac{\log( N/\delta_N)}{N}= \frac{\log( N^{k+1})}{N}=(k+1)\frac{\log N}{N} \rightarrow 0$.}
\end{remark}}

\begin{theorem}[Distance]\label{thm:dist}
\textit{Consider a graphon game $\mathcal{G}(\mathtt{S}, \J, {\blue \theta}, W)$ where each player has  homogeneous strategy set, i.e.,  $\mathtt{S}(x)=\mathcal{S}$ for all $x\in[0,1]$. Suppose that $\mathcal{G}$ satisfies Assumptions \ref{ass:cost}, \ref{ass:constraint}B), \ref{cond} and \ref{lipschitz}.  Let $\bar s$ be its unique Nash equilibrium  and fix  any {\blue admissible confidence} sequence.  Let $\bar s^{[N]}_{w/s}$ be an arbitrary step function equilibrium of the sampled network game $\mathcal{G}^{[N]}(\{\mathcal{S}\}_{i=1}^N, \J,{\blue \{{\theta}(t^i)\}_{i=1}^N}, P^{[N]}_{w/s})$, as introduced in  Section \ref{step2}.
Then with probability at least $1-2\delta_N$, for $N$ large enough,  it holds
\begin{align*}\|\bar s^{[N]}_{w/s}-\bar s\|_{L^2}  \le K\rho(N)
\end{align*}
for $ K=\frac{{\color{blue}\max\{\ell_\J,\ell_\theta\}}/\alpha_\J}{1-\ell_\J/\alpha_\J \lambda_{\textup{max}}(\mathbb{W})} $ {\blue and $\rho(N)\rightarrow 0$ as $N\rightarrow \infty$}.\footnote{\blue The exact formula for $\rho(N)$ is given in the proof of this statement in Appendix \ref{sec:conv_proof}.}  
Moreover, $\|\bar s^{[N]}_{w/s}-\bar s\|_{L^2} \rightarrow 0$  almost surely when $N\rightarrow \infty$.
}
\end{theorem}

 The proof of Theorem \ref{thm:dist} is given in Appendix \ref{appC} and consists of three main steps. First, by Theorem \ref{thm:oneNash}  one can  compare the  equilibrium of  the graphon game $\mathcal{G}(\mathcal{S}, \J, \theta,W)$ and the equilibria of any sampled network game $\mathcal{G}^{[N]}(\{\mathcal{S}\}_{i=1}^N, \J,{\blue \{{\theta}(t^i)\}_{i=1}^N}, P^{[N]}_{w/s})$ by equivalently comparing the equilibria of two graphon games, one over the original graphon $W$ and one over the step function graphon $W^{[N]}_{w/s}$ corresponding to $P^{[N]}_{w/s}$. Second, by Theorem \ref{thm:comp} the distance of the equilibria in these two graphon games can be upper bounded with a quantity that depends on  the distance of the corresponding graphon operators. Third, the distance of the graphon operators can be upper bounded as shown for example in \cite{Lovasz2012} (for generic graphons) and in \cite{graphons} (for graphons satisfying Assumption \ref{lipschitz}). Almost sure convergence can then be obtained by using the derived bounds with $\delta_N= \frac{1}{N^2}\rightarrow 0$ and  Borel-Cantelli lemma.

  In many practical contexts, it might also be of interest to quantify the distance between the equilibria of two   network games sampled from the same graphon. Such a result can be used to judge the robustness of the equilibrium outcome to stochastic variations in the realized links or in the number of players.  Theorem \ref{thm:dist} can be used to obtain such a bound by triangular inequality. 
  {\blue  Finally we note that Theorem \ref{thm:dist} bounds the distance of   the  equilibria of the sampled network game to  the graphon equilibrium in $\|\cdot\|_{L^2}$. This does not directly imply that  playing the  graphon equilibrium strategy in the sampled network game is an (approximate) Nash equilibrium: we show that this is the case under additional regularity assumptions in Lemma \ref{thm:eps} in online Appendix \ref{aux_all}.}

 {\color{blue}   \ssection{\textbf{Extensions}}\label{sec:ext}
\subsection{\textbf{Sparser networks}}\label{sec:sparse}

As noted in Remark \ref{order}, the sampling procedure given in Definition \ref{sample} generates dense networks, that is, networks where the number of neighbors per agent grows as $N$ (thus  implying that the number of edges grows roughly as the square of the number of nodes). In this subsection, we show that our theory can be  generalized to a class of sparser networks  for which  the number of  neighbors per agent  grows \textit{sublinearly} with $N$ so that
$\frac{\sqrt{\textup{\# edges}}}{\textup{\# nodes}} \rightarrow 0.$
To this end, we introduce a sparsity parameter $\kappa_N$ and consider the following (generalized) procedure to sample networks from a graphon, see e.g. \cite{borgs2019??}. 

\begin{definition}[Sampling procedure - generalized]\label{sampleG}
\textit{Given any graphon $W$, a sequence $\{\kappa_N\}_{N=1}^\infty$ with $0<\kappa_N\le 1$, and any desired number $N$ of nodes, uniformly and independently sample $N$ points $\{\type^i\}_{i=1}^N$ from $[0,1]$ and 
 define the \textit{$0$-$1$ adjacency matrix} $P_s^{[N]}$ as the adjacency matrix corresponding to a graph with $N$ nodes obtained  by  randomly connecting nodes $i,j \in[1,N]$ with Bernoulli probability   $$\kappa_N W(\type^i,\type^j).$$}
\end{definition}
\begin{remark}
Definition \ref{sample} is a special case of  Definition \ref{sampleG} obtained by setting $\kappa_N=1$. It is easy to see that the expected number of neighbors   in $P_s^{[N]}$ is of order $\kappa_N N$. Hence for these sampled networks
$\frac{\sqrt{\textup{\# edges}}}{\textup{\# nodes}}\approx \sqrt{\kappa_N}$ converges to zero if $\kappa_N\rightarrow 0$.
In the following, we will require that $\lim_{N\rightarrow \infty}\frac{\log(N)}{N\kappa_N}=0$. Hence  this generalized framework  allows the number of neighbors to grow sublinearly in $N$ but still requires a growth faster than $\log(N)$. This is a necessary condition for being able to use concentration inequalities guaranteeing accumulation in the neighbors aggregate.\end{remark}

The new  Definition \ref{sampleG} affects only how a sampled network is generated from the graphon but has no repercussions on the limit for infinite number of agents. In other words, the infinite population game is exactly the same graphon game described in Section \ref{graphon_games} and all the same theorems on existence, uniqueness and continuity continue to hold. Instead we need to modify the definition of local aggregate in a sampled network game to account for the fact that the number of neighbors may now be sublinear. In fact, if we were to use as aggregate the quantity
$$z^i(s)=\frac{1}{N} \sum_{j=1}^N [P^{[N]}_s]_{ij}s^j$$
as introduced in Section \ref{step2} then we may have that $z^i(s)\rightarrow 0$ as $N$ grows larger, thus leading to vanishing network effects. To overcome this issue, we need to scale the network effect $\sum_{j=1}^N [P^{[N]}_s]_{ij}s^j$  by  the expected order of neighbors which according to Definition \ref{sampleG} is  $\kappa_N N$ instead of $N$. Overall, we can define a sampled network game exactly as in Section~\ref{step2}, but using as aggregate
$$z^i_\kappa(s)=\frac{1}{\kappa_NN} \sum_{j=1}^N [P^{[N]}_s]_{ij}s^j.$$
 \begin{definition}[Sampled network game - generalized]\label{sng}
 Given a graphon $W$, a payoff function $U$, a set valued function $\mathtt{S}:[0,1]\rightarrow 2^{\R}$, a sparsity parameter $\kappa_N$  and a parameter function $\theta$, we define a sampled network game among $N$ agents of type $\{\type^i\}_{i=1}^N$  as $\mathcal{G}^{[N]}_\kappa(\{\mathtt{S}(\type^i)\}_{i=1}^N, \J, \{{\theta}(t^i)\}_{i=1}^N,  P_{s}^{[N]})$, where the types $\{\type^i\}_{i=1}^N$ are sampled uniformly and independently at random from $[0,1]$, $P_{s}^{[N]}$ is as in  Definition~\ref{sampleG} and each agent $i$ has payoff
  \begin{equation}\label{gameyG}
\J(s^i,\q^i_\kappa(s),  \theta^i).
\end{equation}  
 \end{definition}

 We next informally discuss how  our main convergence result can be extended to this sparser class of sampled networks. The formal statements and proofs can be found in the Appendix. First, following the same arguments as in Theorem  \ref{thm:oneNash} one can show that a vector $\bar s_{[N]}\in \R^{N}$ is a Nash equilibrium of a sampled network game  if and only if the  corresponding  step function equilibrium $\bar s_{[N]}(x)\in \leb$ is a  Nash equilibrium  of a graphon game  with  step function graphon $W_\kappa^{[N]}$  corresponding to  $\frac{1}{\kappa_N}P^{[N]}$.
 Using this fact, it can then be shown that the equivalent of Theorem  \ref{thm:dist} holds for sparse networks as long as the confidence sequence $\{\delta_N\}_{N=1}^\infty$ is such that $\frac{\log( N/\delta_N)}{N{  \kappa_N}} \rightarrow 0$ and the  rate $\rho(N)$ is modified to $\rho_\kappa(N)$ as detailed in the Appendix.

\subsection{\textbf{Average instead of aggregate}}
\label{avg}

In the results derived so far we defined the \textit{local aggregate} as 
$$z^i(s)=\frac{1}{N} \sum_{j=1}^N [P^{[N]}_s]_{ij}s^j,$$
that is, the sum of neighbors strategies normalized by the population size. While this model is of widespread use in both theoretical and empirical works, it is known that for some   applications, a more suitable model is that of \textit{local average}  obtained by normalizing the network effect by the agent's degree.\footnote{\blue See  \cite{patacchini2012juvenile} and \cite{ushchev2020social} for a discussion of the differences of local aggregate and local average models.} This  corresponds to the choice 
  $$z^i_d(s):=\frac{ \sum_{j=1}^N [P^{[N]}_s]_{ij}s^j}{\sum_{j=1}^N [P^{[N]}_s]_{ij}}.$$
  Our results can be  extended to this setting.  
  The first step is to define  the \textit{normalized local aggregate} for a continuum of agents as
  
  \begin{align*}
z_d(x\mid s):=\frac{\int_0^1 W(x,y)s(y)dy}{\int_0^1 W(x,y)dy}.
\end{align*}
For this quantity to be well defined, we  assume from here on that $\int_0^1 W(x,y) dy \ge d_\textup{min}>0$ for all $x\in[0,1]$.
This definition of local aggregate leads  to a graphon game as defined in Section \ref{graphon_games}, played over the \textit{normalized graphon}
    $$W_d(x,y):= \frac{ W(x,y)}{\int_0^1 W(x,y)dy}.$$

    As second step one can define the associated  \textit{normalized graphon operator} $\mathbb{W}_d$ as the   operator $\mathbb{W}_d:\leb\mapsto \leb$ given by

$$ f(x) \mapsto (\mathbb{W}_df)(x)=\frac{\int_0^1W(x,y)f(y)\mathrm{d}y}{\int_0^1W(x,y)\mathrm{d}y}.$$ Under the assumption that $\int_0^1 W(x,y) dy \ge d_\textup{min}>0$, we show in  online Appendix \ref{avgA} that all the results derived  in Section~\ref{graphon_games} about existence and uniqueness of the graphon equilibrium continue to hold.\footnote{\blue Since $W_d$ is not symmetric,   results need to be stated in terms of 
$\vertiii{\mathbb{W}_d}$ instead of $\lambda_{max}(\mathbb{W})$. Note however that the bound $\vertiii{\mathbb{W}_d}\le \frac{ \lambda_{\textup{max}}(\mathbb{W})}{ d_\textup{min}}$ holds (see online Appendix \ref{avgA}).} 
For example, uniqueness holds if
\[\frac{\ell_\J}{\alpha_\J}\cdot  \frac{ \lambda_{\textup{max}}(\mathbb{W})}{ d_\textup{min}}<1.\]
Continuity of the graphon equilibrium can again be shown similar to Theorem \ref{thm:comp}, with the key difference that the upper bound will depend on the distance of the normalized operators $\vertiii{\mathbb{W}_d-\tilde{\mathbb{W}}_d}$. Theorem \ref{thm:oneNash} holds unchanged, hence the key step to prove convergence of sampled network equilibria to graphon equilibria (Theorem \ref{thm:dist}) in this setting is to show that the distance between the normalized operator $\mathbb{W}_d$ and the normalized operator $\mathbb{W}^{[N]}_{sd}$ corresponding to the sampled network $P^{[N]}_s$ converges to zero with high probability. Again this can be obtained under the assumption that  $\int_0^1 W(x,y) dy \ge d_\textup{min}>0$ (a proof for Lipschitz continuous graphons is provided in   online Appendix \ref{avgA}).

  \subsection{\textbf{Directed networks}}\label{sec:dir}
So far we assumed that the graphon is a symmetric function and we thus generated undirected sampled networks. The results of Section \ref{graphon_games} on existence, uniqueness and continuity of the graphon equilibrium continue to hold even when the generating graphon is not symmetric, with the only caveat that the eigenvalues of the corresponding operator are not necessarily real hence one need to use $\vertiii{\mathbb{W}}$ instead of $\lambda_\textup{max}(W)$. Theorem \ref{thm:oneNash} holds unchanged. The only place where symmetry is used  in Theorem 5 is to prove that the matrix $P^{[N]}_s$ accumulates around its expectation $P^{[N]}_w$. To prove this fact we used a matrix concentration result from \cite{chung2011spectra} which holds  for symmetric matrices. However, a similar result can be obtained for the directed case as well (see Lemma \ref{lem:asy} in the online Appendix  \ref{sec:dirA}). Using such a result one can obtain convergence also for directed networks.

\ssection{\textbf{Theory of targeted interventions} }\label{sec:int}
 We consider a central planner (CP)  designing targeted interventions for regulating economic behavior over a network. Our goal is to use the graphon game approximation to design near optimal and computationally tractable interventions for sampled network games (under some regularity assumptions on the underlying graphon).

To this end, we build on  \cite{galeotti2017targeting} which considers linear quadratic network games with scalar  nonnegative strategies and  payoff  as introduced in Equation \eqref{costL}.   For simplicity we focus on games
with strategic complementarities  (i.e. with $\alpha>0$). We assume that the goal of the CP is to maximize the \textit{average social welfare} (defined as the average of the agents payoffs at equilibrium) through  interventions that directly modify the standalone marginal return for an arbitrary agent $i$ from $\thetaL^i$ to $\thetaL^i+\hat{\thetaL}^i$, leading to the modified payoff function
\begin{equation}\label{eq:cost_quadratic2}
\J(s^i,z^i, \thetaL^i+\hat{\thetaL}^i)=-\frac12 (s^i)^2+s^i[\alpha z^i+\thetaL^i+\hat{\thetaL}^i].
\end{equation}

We assume that the planner is subject to a budget constraint which penalizes  interventions in a convex form (to capture the fact that  interventions are increasingly costly), leading to
\[\sum_{i=1}^N (\hat{\thetaL}^i)^2\le CN.\]
Note that we allow  the budget to scale with the population size $N$ to model  the fact that networks with more agents are allocated a proportionally higher budget.
By using the  characterization of equilibrium in linear quadratic games, i.e. $\bar s^i=\alpha \bar z^i+\thetaL^i+\hat{\thetaL}^i$ with $\bar z^i=\frac1N\sum_j P^{[N]}_{ij} \bar s^j$, the  objective function of the central planner can be rewritten as 
\[T^{[N]}(\hat \thetaL):=  \frac1N \sum_{i=1}^N \J(\bar s^i, \bar z^i, \thetaL^i+\hat{\thetaL}^i)= \frac1N \sum_{i=1}^N \left(-\frac12 (\bar s^i)^2+ \bar s^i[\alpha \bar z^i+\thetaL^i+\hat{\thetaL}^i]\right)= \frac{1}{2N} \sum_{i=1}^N (\bar s^i)^2,\]
where $\hat \thetaL:=[\hat{\thetaL}^i]_{i=1}^N$.
This leads to the following optimization problem for the central planner

\begin{equation}\label{finite_opt}
\begin{aligned}
T^{[N]}_{\textup{opt}}:=\max_{\hat\thetaL^{[N]} \in\R^N}&\quad \frac{1}{2N}\|\bar s^{[N]}\|^2,\\
\textup{s.t.}&\quad \bar s^{[N]} =\textup{Nash equilibrium of } \mathcal{G}^{[N]}({\mathbb{R}_{\ge 0}}, \J,\thetaL^{[N]}+\hat \thetaL^{[N]}, P^{[N]}),\\
&\quad \textstyle \frac1N \|\hat {\thetaL}^{[N]}\|^2\le C,
\end{aligned}
\end{equation}
where we added the apex $^{[N]}$ to stress the dependence on the population size.

\subsection{\textbf{Graphon intervention}}

Problem \eqref{finite_opt} scales with the size of the network and  becomes computationally challenging  for networks with more than a few hundreds of agents. We next suggest an alternative approach for sampled network games (i.e. for cases when $P^{[N]}=P^{[N]}_{w/s}$ is a realization from an underlying graphon $W$ and $\thetaL^{[N]}_i=\thetaL(t^i)$) based on the following optimization problem in the graphon space

\begin{equation}\label{infinite_optt}
\begin{aligned}
\thetaL^*\in \arg\max_{\hat \thetaL \in \leb }&\quad \frac12 \|\bar s_{\hat \thetaL}\|_{L^2}^2,\\
\textup{s.t.}&\quad \bar s_{\hat \thetaL} =\textup{Nash equilibrium of } \mathcal{G}(\mathbb{R}_{\ge0}, \J,\thetaL+\hat \thetaL, W) ,\\
&\quad \| \hat \thetaL\|^2_{L^2} \le C.
\end{aligned}
\end{equation}
 In the next theorem we show that a near optimal intervention for the sampled network game can be obtained from the optimal solution $\thetaL^*$ of \eqref{infinite_optt} by allocating  to any sampled agent $i$ (of type $t^i$) an intervention proportional to $\thetaL^*(t^i)$, that is,
$$[\hat \thetaL^{[N]}_\textup{graphon}]_i= \frac{\thetaL^*(t^i)}{\eta^{[N]}},$$
where $\eta^{[N]}$ is a normalization to guarantee that the budget constraint is met with equality (i.e. $\frac1N\| \hat \thetaL^{[N]}_\textup{graphon}\|^2 = C$).

\begin{theorem}\label{thm:inter}
Consider a  network game $\mathcal{G}^{[N]}(\mathbb{R}_{\ge0}, \J, \{\thetaL(t^i)\}_{i=1}^N,P^{[N]}_{w/s})$ sampled from the graphon $W$ according to  the procedure given in Definition \ref{sample}. Suppose that   $\J$ is as in \eqref{eq:cost_quadratic2} with $0<\alpha<\frac{1}{\lambda_\textup{max}(\mathbb{W})}$, that Assumption~\ref{lipschitz} holds and that
 $\thetaL^*$ solution to \eqref{infinite_optt} is piecewise Lipschitz and bounded. For  any admissible confidence sequence  and $N$ large enough, with probability at least $1-2\delta_N$, 
$$T^{[N]}(\hat \thetaL^{[N]}_\textup{graphon})\ge T^{[N]}_{\textup{opt}}-   \rho_T(N),$$
where  $\rho_T(N)\rightarrow 0$ as $N\rightarrow\infty$.\footnote{\blue The explicit formula for $\rho_T(N)$ is given in the Appendix.}
\end{theorem}

Such a graphon intervention offers an advantage if Problem \eqref{infinite_optt} can be solved efficiently. In the next section we show that this is the case for a large class of graphons of practical interest.

\subsection{\textbf{Tractability of Problem \eqref{infinite_optt}}}\label{sec:tractability}

In this section we restrict our attention to graphons in which only a finite number $R$ of eigenvalues  $\{\lambda_r\}_{r=1}^R$  are different from zero (i.e., finite-rank graphons). For this  class of graphons we show that a solution to Problem \eqref{infinite_optt} can be obtained   by solving an equivalent problem in $R+1$ variables. This is a clear advantage with respect to solving Problem \eqref{finite_opt} which instead requires $N \gg R$ variables.\footnote{\blue It can be shown that both  Problem \eqref{infinite_optt} and Problem \eqref{finite_opt}  can be reformulated as an  SDP with two variables and an inequality constraint involving a matrix of dimension $R+2$ and $N+1$ respectively, see \cite[Appendix B.1]{boyd2004convex}  and \cite[Theorem 1]{galeotti2017targeting}.}  

\begin{lemma}\label{solvability}
Suppose that $0<\alpha<\frac{1}{\lambda_\textup{max}(\mathbb{W})}$ and $\mathbb{W}$ has rank $R<\infty$, let $\mathcal{K}$ be the kernel of $\mathbb{W}$ and $\{\psi_r\}_{r=1}^R$ be an orthonormal basis of $\mathcal{K}^\perp$  composed of eigenfunctions of $\mathbb{W}$ corresponding to the eigenvalues $\{\lambda_r\}_{r=1}^R$. Set $b_r = \langle \thetaL, \psi_r \rangle$ for all $r=1,\ldots,R$ and let $b_0 \psi_0$ be the projection of $\thetaL$ in $\mathcal{K}$, with $\|\psi_0\|_{L^2}=1$. Set  $\lambda_0=0$. A maximizer of  \eqref{infinite_optt}  can be computed as $\thetaL^*=   \sum_{r=0}^R \hat b_r^* \psi_r$ where $\{\hat b_r^*\}_{r=0}^R$ solves 
\begin{equation}\label{infinite_opttt}
\begin{aligned}
 \max_{[\hat b_0, \ldots, \hat b_R]  }&\quad \frac12    \sum_{r=0}^R \frac{(b_r+\hat b_r)^2}{(1- \alpha \lambda_r)^2} ,\\
\textup{s.t.} &\quad \sum_{r=0}^R \hat b_r^2 \le C.
\end{aligned}
\end{equation}

\end{lemma}

It is important to remark that the class of finite rank graphons is quite rich, we provide some examples next.

\begin{example}[Community structure]\label{ex:SBM}

Consider a generalization of the community model  with $K$ communities introduced in  Example 1, where we allow agents across different communities to interact with  different probabilities.  Specifically, let $Q\in[0,1]^{K\times K}$ be a symmetric matrix whose element in position $(k,l)$ denotes the probability that agents of community $k$ and $l$ are interacting (the graphon in Example~1 corresponds to the special case $Q=[g_{\textup{in}} I_K+ g_{\textup{out}}(\mathbbm{1}_K\mathbbm{1}_K^\top -I_K)]$). Let $\mathcal{C}_k$ be the subset of $[0,1]$ associated with community $k$, with $|\mathcal{C}_k|=\pi_k$ and $\sum_k \pi_k=1$, and construct the SBM graphon
\[W_\textup{SBM}(x,y)=Q_{ij} \textup{ for all } x\in \mathcal{C}_i, y\in \mathcal{C}_j.\]
The SBM graphon is finite rank with rank equal to the number of communities. In fact as shown for example in \cite{graphons} 
 eigenvalues and eigenfunctions of the SBM graphon operator can be easily computed by considering the auxiliary matrix
\begin{equation}\label{eq:E}E:=Q\Pi \in\mathbb{R}^{K\times K},\end{equation}
where $\Pi$ is a diagonal matrix whose diagonal elements correspond to the community sizes, that is $\Pi_{kk}=\pi_k$ for all $k$.
 {Lemma \ref{eigSBM} provided in the online Appendix   shows that $\mathbb{W}_{\textup{SBM}}$ and $E$ have the same  eigenvalues and the eigenfunctions of $\mathbb{W}_{\textup{SBM}}$ are piecewise constant over the partition $\{\mathcal{C}_k\}_{k=1}^K$, with constant  value in each community $\mathcal{C}_k$  given by the $k$-th element of the corresponding eigenvector of $E$.\\
If  we assume that all agents within the same community have the same stadalone marginal return  (i.e.,  $\thetaL(x)=\thetaL^{com}_k$ for all $x\in \mathcal{C}_k$) then it is immediate to see that at the graphon equilibrium each agent belonging to the same community has the same strategy $\bar s_k^{com}$ and the vector of such strategies $\bar s^{com}\in\mathbb{R}^K$ satisfies the relation
\begin{equation}\label{eq:com_eq}
 \bar s^{com} = (I_K - \alpha E)^{-1} \thetaL^{com}.
\end{equation}
Turning to the optimal intervention problem \eqref{infinite_optt}, note that since $\thetaL(x)=\thetaL_k^{com}$ for all $x\in \mathcal{C}_k$, $\thetaL$ can be written as a linear combination of the eigenfunctions of $\mathbb{W}$ (i.e., $b_0$ as defined in Lemma~\ref{solvability} is zero). One can then conclude that  $\hat \thetaL_0^*=0$ and that $\thetaL^*=\sum_{r=1}^R \hat \thetaL^*_r \psi_r$ is constant within each community. Hence  in the limit of large number of agents it is sufficient to intervene at the level of  communities instead of individuals (see Section \ref{sec:case_study} for a detailed case study illustrating this point).
 }
\end{example}

The example above considered a graphon with a finite number of types (in that case the number of communities). We next show that a graphon can have finite rank even when there is a continuum of types.
\begin{example}[Rank one graphon]
 Consider the graphon
$$W_\gamma(x,y)=\frac{1}{(xy)^\gamma} \quad \textup{with }  \gamma\in(0,1/2) 
\textup{ and with associated degree} \quad d(x)=\frac{1}{1-\gamma} x^{-\gamma}.$$
This graphon has rank one with normalized eigenfunction $\psi_\gamma(x)=\sqrt{1-2\gamma}x^{-\gamma}$ and eigenvalue $\lambda=\frac{1}{(1-2\gamma)}$. In fact
$$ (\mathbb{W}_\gamma \psi_\gamma)(x)= \int_0^1 W_\gamma(x,y) \psi_\gamma(y)dy = \sqrt{1-2\gamma} \int_0^1 \frac{1}{(xy)^\gamma} \frac{1}{y^{\gamma}}dy =\frac{1}{(1-2\gamma)} \psi_\gamma(x)$$
and for any $f\in\leb$
$$\mathbb{W}f(x)=\int_0^1 \frac{1}{(xy)^\gamma} f(y)dy = \frac{1}{x^\gamma} \int_0^1 \frac{1}{y^\gamma} f(y)dy = \lambda \langle f, \psi_\gamma \rangle \psi_\gamma(x).$$
Turing to the optimal intervention problem \eqref{infinite_optt}, let $b_1=\langle \thetaL, \psi_\gamma\rangle$ and $b_0=\|\thetaL-b_1\psi_\gamma\|_{L^2}$. Then problem \eqref{infinite_optt}  can be solved by solving 

\begin{equation}
\begin{aligned}
\left[\hat b_0,  \hat b_1\right]^*:= \arg\max_{[\hat b_0,  \hat b_1]  }&\quad \frac12   \left[ (b_0+\hat b_0)^2+ \frac{(b_1+\hat b_1)^2}{(1- \alpha \lambda)^2} \right],\\
\textup{s.t.} &\quad \hat b_0^2+\hat b_1^2 \le C,
\end{aligned}
\end{equation}
and then setting $\hat \thetaL(x)= \hat b^*_0 \frac{(\thetaL(x)-b_1\psi_\gamma(x))}{b_0}+ \hat b^*_1\psi_\gamma(x)$.

\end{example}

\section{ \textbf{An illustrative case study}: Interventions in rural villages }
\label{sec:case_study}

In this section we illustrate the differences (in terms of information, computation and optimality) between the intervention procedure described in Section \ref{sec:int} and  a more direct approach based on detailed network information. To this end, we  construct a simulated dataset of $80$ networks (which for example could model  interactions  among the inhabitants of $80$ different rural villages) and we  assume that agents within each network make  strategic decision subject to network externalities (e.g., each individual in a village decides his level of investment  in a microfinance program). Note that we assume the networks to be isolated (e.g., the villages are far apart so that there is no interaction of individuals across villages), hence these can be seen as $80$ independent  network games. Each agent $i$ has payoff
$$\J(s^i,\q^i_\kappa(s),  \thetaL^i)=-\frac12 (s^i)^2+ s^i[\alpha z^i_\kappa(s)+\thetaL^i],$$
where the parameter $\alpha$ is the same for all agents,  $\thetaL^i$ is  agent specific and $\kappa$ is a sparsity parameter as introduced in Section \ref{sec:sparse}. 
We also assume that agents in each network are equally likely to belong to one of $4$ different communities (e.g., $4$ different caste in the case of rural villages) and that the probability of agents interacting depends on  community identity according to the community structure illustrated in Figure \ref{village}. Finally, we  assume that  agents belonging to the same community have the same parameter $\thetaL^i$, which we denote by $\thetaL^{com}_h$ for community $h=1,\ldots,4.$

 \begin{figure}[H]
 \begin{center}
 $\vcenter{\hbox{\includegraphics[width=0.18\textwidth]{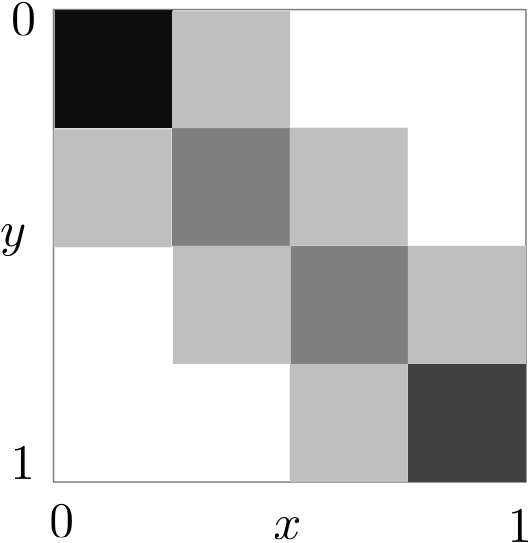}}}$ \qquad
 \begin{minipage}{0.6\textwidth}
  \begin{center}
$\vcenter{\quad\hbox{\includegraphics[width=0.9\textwidth]{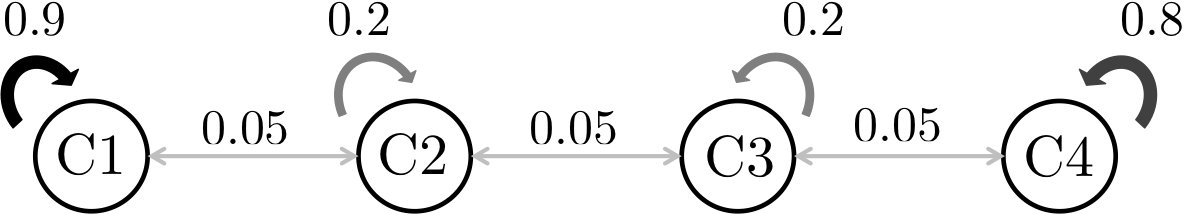}}}$\\[0.2cm]
\end{center}
\end{minipage}
\end{center}
\caption{\blue The SBM graphon used to generate sampled networks    in the case study of Section \ref{sec:case_study}. The  figure on the right is a visualization of the  interactions among the 4 communities (the color  width of the arrows is larger the stronger the interaction), the figure on the left shows the corresponding  graphon (with white representing  $0$ and black representing $1$). Note that community $1$ and $4$ are tightly intra-connected, while $2$ and $3$ are less intra-connected. Communities are only slightly inter-connected with neighboring communities. For sampling,  we used a sparsity parameter $\kappa_N=40/N^{0.8}$ (see Section \ref{sec:sparse}).}
\label{village}
\end{figure}

Our main interest is to understand how a CP can allocate a limited budget in each sampled network (village) to maximize agents welfare,  by designing interventions as discussed in Section \ref{sec:int}.

\subsection{\textbf{Data acquisition}}
Our aim  is to simulate the procedure that the CP would have to follow  in a  field experiment. To this end, from here on we are going to assume that the central planner \textbf{does not} have access to the information detailed above, but instead needs to rely on surveys to reconstruct agents attributes and interactions. 
Regarding the latter, we are going to assume that the CP can use two different types of relational surveys:

\begin{itemize}
\item \textbf{Detailed Relational Data (DRD):} The CP is able to ask to \textit{each agent} in each network (village) \textit{the exact identity of each of his neighbors}. 
\item \textbf{Aggregated Relational Data (ARD):} The CP is able to ask to a \textit{subset of the agents} in each network (village) \textit{how many of his neighbors belong to each community}. 
\end{itemize}

As argued in \cite{breza2017using} aggregated relational data of the second type is much easier to obtain in the field than the information required by the detailed relational survey. Furthermore, the aggregated information required by the second type of survey can allow data acquisition in settings where detailed information is not possible because of proprietary data or privacy concerns (e.g. in the case of financial intermediaries or high-risk populations). 

While relational data is typically hard to obtain,  it is instead  common in empirical studies to collect detailed agent-level information  through an exhaustive census, see for example \cite{banerjee2013diffusion}.
For this reason we are  going to assume that, in addition to one of the relational surveys above, the CP can perform a census of all the agents in each network asking about agent-level information such as:  i)  agent type (e.g., the community to which the agent belongs), ii)   equilibrium strategy before the intervention (e.g., the current level of investment in the microfinance program) and iii) the status quo standalone marginal return $\thetaL^i$.

Note that we do not assume that the CP has any information about the strength  of peer effects, $\alpha$, or about the parameters of the network formation model. Indeed these parameters would not be available in field work and therefore need to be estimated from the relational survey and census data described above.

\subsection{\textbf{Intervention design based on detailed relational data}}

If the CP has access to the information contained in the census and in the \textit{detailed relational data} then he can reconstruct for each village: i) the exact network of interactions among the agents $P^{[N]}_s$, ii) the vector of parameters $\thetaL^{[N]}\in\mathbb{R}^N$ and iii) the  Nash equilibrium $\bar s^{[N]}$ before the intervention. He can then use this information to infer the unknown normalized network parameter $\alpha_\kappa:=\frac{\alpha}{\kappa_N}$ by performing least square regression given the relation
$$   \bar s^{[N]}=\left(I_N- \frac{\alpha_\kappa}{N} P^{[N]}_s\right)^{-1} \thetaL^{[N]} \textup{ or equivalently }  Y=\alpha_\kappa X \textup{ for } Y:=\bar s^{[N]} -\thetaL^{[N]}, X:= \frac{P^{[N]}_s}{N}\bar s^{[N]}, $$
leading to 
$$\hat \alpha_\kappa^{DRD} := (X^\top X)^{-1} X^\top Y,$$
where the superscript $DRD$ denotes the use of the detailed relational data.\footnote{{\color{blue} The regressor $X$ is an endogenous variable (as it depends on $\bar s^{[N]}$) hence one should use regression based on instrumental variables as detailed in  \cite{bramoulle2009identification}. This is not needed in our case study because we assumed that there is no noise in the equation $Y=\alpha_\kappa X$. In this case ordinary least square can be used and  produces the exact parameter  $ \alpha_\kappa$ since
$\hat \alpha_\kappa^{DRD} := (X^\top X)^{-1} X^\top Y= (X^\top X)^{-1} X^\top (\alpha_\kappa X)=\alpha_\kappa (X^\top X)^{-1} (X^\top X)=\alpha_\kappa. $
In this case study we assumed no noise because our objective is  to compare the performance of interventions based on detailed relational data (DRD) versus aggregate relational data (ARD). Assuming no noise corresponds to the best case scenario for the interventions based on detailed relational data and allows us to focus only on the comparison of interest. Finally, note that it is not possible to estimate $\alpha$ and $\kappa_N$ distinctly, but this is  not needed. 
} }
In our simulated data we assumed no noise hence this procedure allows the CP to recover $\alpha_\kappa$ exactly (i.e., $\hat \alpha^{DRD}_\kappa=\alpha_\kappa$).
Using $P^{[N]}_s, \thetaL^{[N]}$ and the estimated parameter $\hat \alpha_\kappa^{DRD}$ the CP can either solve exactly the optimal intervention problem in \eqref{finite_opt} (if this is computationally feasible) or otherwise he can use the  heuristic suggested in \cite{galeotti2017targeting} and allocate the budget according to the dominant eigenvector of  $P^{[N]}_s$. We refer to these two interventions as \textit{network optimal} and \textit{network heuristic}, respectively.\footnote{\blue The CP could also employ an in-between strategy where budget is allocated according to the $r$ dominant eigenvectors for some $r>1$. This strategy  still requires the detailed relation dataset and will have performances that are in between the \textit{network optimal} and \textit{network heuristic}. }

\subsection{\textbf{Intervention design based on aggregated relational data}}\label{sec:ARD}

Suppose instead that the CP cannot access the detailed relational survey, but instead need to rely only on the \textit{aggregated relational data}.  We assume that the CP knows that the networks are drawn from a stochastic block model with 4 communities and use the ARD to estimate the parameters of the SBM model and the peer effect parameter $\alpha_\kappa$.
To this end, for each village the CP can estimate:
\begin{enumerate}
\item The exact proportion of agents in each community (from the census data) as
$$ \pi_h :=\frac{N_h}{N}:=\frac{\textup{number of agents in community } h}{\textup{total number of agents in the census}} .$$
Let $\Pi$ be a diagonal matrix with  $\pi_h$  in position $(h,h)$.
\item The maximum likelihood estimator of the interaction probability of agents of  community $h$ and $h'$ (from the subset of agents interviewed in the aggregated survey) as
$$\hat q^{{ARD}}_{h,h'} :=\frac{ S_{h,h'} + S_{h',h} }{ S_h N_{h'} + S_{h'} N_h },$$
where
$S_h$ is the total number of agents surveyed from community $h$ in the ARD and
$S_{h,h'}$ is the total number of neighbors that they reported having in community $h'$. The superscript \textit{ARD} denotes the use of aggregated data. Let $\hat Q^{{ARD}}_\kappa:=[\hat q^{{ARD}}_{h,h'}]$ be the estimated interaction matrix (see Example \ref{ex:SBM} in Section \ref{sec:tractability}) and $\hat E^{{ARD}}_\kappa:= \hat Q^{{ARD}}_\kappa \Pi.$\footnote{\blue Technically, since we assume sparse networks these are the matrices $Q$ and $E$ as described in Example \ref{ex:SBM} \textit{multiplied by $\kappa_N$}, this is not a problem because we can only estimate $\alpha$ \textit{divided by $\kappa_N$}, hence the (unknown) $\kappa_N$ factor cancels out, that is $\alpha_\kappa E_\kappa=\alpha E$.} 
\item The average strategy of agents in community $h$ before the intervention as 
$$\hat s_h^{{ARD}} := \frac{\textup{sum of effort of agents surveyed from community } h }{N_h}.$$
We show in Appendix  \ref{csA} that for $N\rightarrow \infty$, $\hat s_h^{{ARD}}$  converges   almost surely to the strategy $\bar s^{com}_h$ played by agents of community $h$ in the  graphon game (recall that since the graphon in this case is a SBM, each agent in community $h$ has the same graphon equilibrium strategy, see \eqref{eq:com_eq}).
\item The parameter $\alpha_\kappa$  by 
$$\hat \alpha^{ARD}_\kappa = (\hat X^\top \hat X)^{-1} \hat X^\top \hat Y,$$
with $\hat X:= \hat E_\kappa^{{ARD}} \hat s^{{ARD}}$ and $\hat Y:=\hat s^{{ARD}} - \thetaL^{{com}}$,
where $ \thetaL^{{com}}$ is the vector of marginal return per community (which can be recovered exactly from the census data). We show in Appendix \ref{csA} that $\hat \alpha^{ARD}_\kappa\rightarrow \alpha_\kappa$   almost surely for $N\rightarrow\infty$.
\end{enumerate}}

{\color{blue}
Based on $\thetaL^{{com}}, \hat \alpha^{ARD}$ and $\hat E^{{ARD}}$ the CP can solve Problem \eqref{infinite_optt} (by equivalently solving Problem \eqref{infinite_opttt}) and obtain the optimal graphon intervention. Note that for this case study
Problem  \eqref{infinite_opttt}  is a problem of dimension $4$ and 
outputs the  intervention that the CP should apply in each community. The CP then knows which intervention to apply to each agent because he collected information about agent's type in the census. We refer to this intervention as \textit{graphon optimal}.

\begin{figure}
\includegraphics[width=0.5\textwidth]{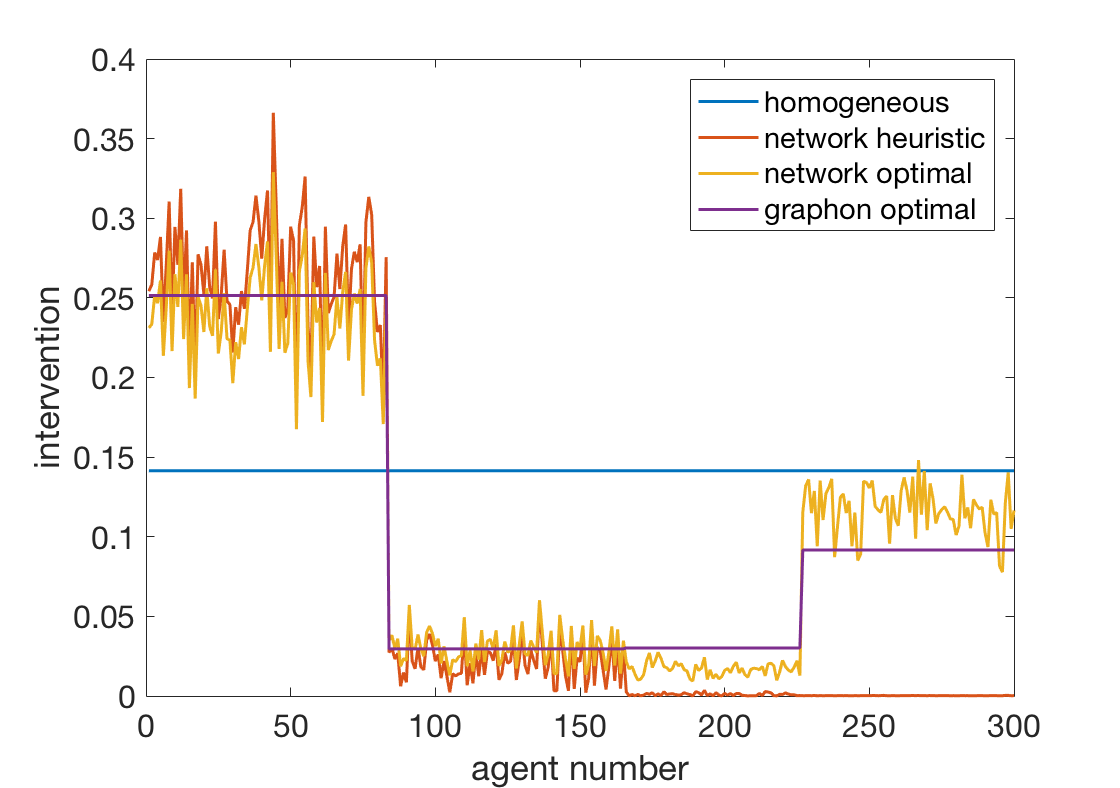}
\includegraphics[width=0.5\textwidth]{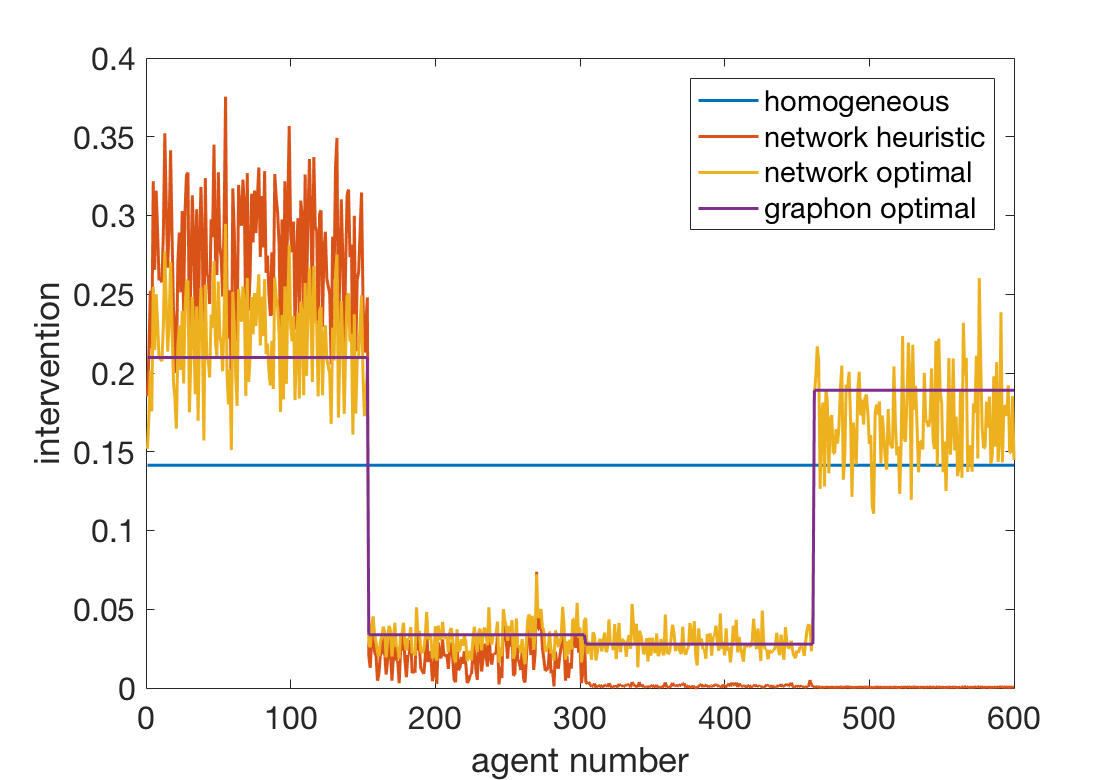}
\caption{Interventions for a sampled network of size $N=300$ (left) and  $N=600$ (right).}
\label{fig:inter2}
\end{figure}

\subsection{\textbf{Comparison}}
Figure \ref{fig:inter2} illustrates the  \textit{network optimal}, \textit{network heuristic} and \textit{graphon optimal} interventions for two sampled networks of size $N=300$ and $N=600$. A first observation is that while the first two interventions are tailored to the specific network realization (and thus prescribe a different intervention to each agent), the graphon intervention prescribes the same intervention to each agent belonging to the same community. 
We finally compare the performances of three type of interventions  in terms of 
optimality, information and computation. 
\begin{enumerate}
\item \textit{Information:} as discussed above the \textit{network optimal} and \textit{network heuristic} interventions require  detailed relational data (DRD), while the \textit{graphon optimal} intervention can be computed based solely on aggregated relational data (ARD);
\item \textit{Computation:} the \textit{network optimal} intervention requires the solution of Problem \eqref{finite_opt} whose complexity is polynomial in  $N$ (we could not find a solution  for $N\ge 1000$),  the \textit{network heuristic} intervention requires the computation of the dominant eigenvector of $P^{[N]}_s\in\mathbb{R}^{N\times N}$ which is again polynomial in $N$, the \textit{graphon optimal} intervention requires the solution of Problem  \eqref{infinite_opttt} which is polynomial in $K=4$.
\item \textit{Optimality:} the following table illustrates the percentage of welfare improvement under the three different policies (averaged over the $80$ networks) with respect to the homogeneous intervention that splits the budget equally for all the agents. Different columns represent repetitions of the same case study for networks with increasing number of agents. 

\begin{table}[H]
\blue
\begin{center}
\begin{tabular}{|c|c|c|c|} \hline
& \textbf{\textit{Case study 1}} &  \textbf{\textit{Case study 2}}  &  \textbf{\textit{Case study 3}}  \\
& \textbf{{($\boldsymbol{N=300}$)}} &  \textbf{{($\boldsymbol{N=600}$)}}  &  \textbf{{($\boldsymbol{N=1200}$)}}  \\ \hline
\textit{Avg. Impr. Network Optimal} & 25.6\% ($\pm$6.6) &23.0\% ($\pm$3.6)&-\\
 &\color{gray}\small{[15.4;61.5]\%}  &\color{gray}\small{[19.0;37.8]\%}&\color{gray} \small{[-]}\\[0.2cm]
\textit{Avg. Impr. Graphon Optimal } & 22.8\% ($\pm$7.1) &20.8\% ($\pm$3.9)&20.6\% ($\pm$2.3)\\
 &\color{gray}\small{[5.6;59.7]\%}  &\color{gray}\small{[14.8;35.7]\%}&\color{gray} \small{[17.6;31.5]\%}\\[0.2cm]
\textit{Avg. Impr. Network Heuristic} & 12.6\% ($\pm$15.7)&4.5\% ($\pm$12.3)&2.8\% ($\pm$10.5)\\
 &\color{gray}\small{[-12.6;60.73]\%}  &\color{gray}\small{[-12.2;33.9]\%}&\color{gray} \small{[-12.4;27.3]\%}\\[0.2cm] \hline
\textit{Avg. Degree Cumulative} &18.5  &21.6 &24.9\\ 
\textit{Avg. Degree Community 1}&27.9   &34.1 &39.2 \\
\textit{Avg. Degree Community 2} &9.5& 10.7  &12.3\\
\textit{Avg. Degree Community 3} &9.2& 10.7  &12.3\\
\textit{Avg. Degree Community 4}& 27.0&  30.5 &35.2\\ \hline 
\end{tabular}
\end{center}
\caption{\blue Comparison of \textit{network optimal} (NO), \textit{graphon optimal} (GO)  and \textit{network heuristic} (NH) intervention for the case study described in Section \ref{sec:case_study}. The average improvement is computed as $Avg. Impr= \frac{1}{80} \sum_{g=1}^{80} \left(\frac{\textup{welfare according to NO/NH/GO intervention in network } g}{\textup{welfare according to homogeneous intervention in network } g}-1\right)$, one standard deviation is reported in round brackets, minimum and maximum are reported in square brackets. We also show the average degree per community and in the entire graph to illustrate that the \textit{graphon optimal} (GO) intervention is a good approximation in a range of degrees that is realistic (and does not increase too quickly in $N$ thanks to the sparsity parameter $\kappa_N$). In all the case studies, we used  $C=0.02N$, $\alpha= 2.65$ and $\hat \thetaL^{com}=[0.1,0.1,0.1,0.25]$. For the ARD, we assume that the aggregated relational survey is completed by \textbf
{$\mathbf{10\%}$ of the agents in each network}.}
\label{default}
\end{table}%

\end{enumerate}
}

\ssection{\textbf{Conclusion}}
\label{conc}

 In this work we  introduced the novel class  of graphon games for modeling  strategic behavior in  infinite populations while accounting for local heterogeneity. We then showed that graphon games can be used to approximate strategic behavior in large but finite sampled network games by interpreting the graphon as a stochastic network formation process. This statistical interpretation of network games allows for the design of simple intervention policies that do not require detailed information about the network realization.  

  We believe that the initial investigation of graphons as a tool to model strategic behavior presented in this work  can  be extended in  a  number of different  directions. 
{\color{blue} First, in this paper   to guarantee uniqueness of the Nash equilibrium we used Assumption \ref{cond}, which is formulated in terms of the maximum eigenvalue of the graphon. Previous works showed that alternative conditions for uniqueness can formulated in finite network games by using conditions involving the maximum degree or the minimum eigenvalue (for games with strategic substitutes). Extending those results to graphon games is an interesting open direction as well as extending our analysis beyond uniqueness. Second, as an application of our framework we showed how the graphon  approach allows the computation of almost optimal targeted  interventions, overcoming the  computational intractability of approaches based on full network information. We believe that our results can be generalized to other type of interventions, such as selecting the key player as introduced in \cite{ballester2006s}. Third, we here defined graphon games as nonatomic games. It might be interesting to extend this framework to allow for a small number of atomic (major) agents that influence a mass of nonatomic (minor) agents interacting over a graphon, similarly to previous results derived for mean field games in \cite{nourian2013}.  Finally, in our case study and in  online Appendix \ref{sec:id}, we hinted at how the graphon game framework could be used to estimate peer effects when information about the realized network is not available. We believe that extending these results would be of practical interest.}

 \appendix
 
\begin{center}\textbf{Appendices}\end{center}
 \renewcommand{\thesection}{{\Alph{section}}}

{\textbf{Summary of Notation:}} We denote by $\R^n$ the  space of $n$-dimensional vectors, by $\leb$ the  space of square integrable functions defined on $[0,1]$ and by  $\lebn$ the space of square integrable vector valued functions defined on $[0,1]$. The norms in these spaces are denoted  by $\|v \|:=\sqrt{\sum_{h=1}^n [v]_h^2}$, $\| f\|_{L^2}:=\sqrt{\int_0^1 f(x)^2 dx}$, $\|g\|_{L^2;\mathbb{R}^n} 
:=  \sqrt{ \int_0^1 \|g(x)\|^2 dx}  $, respectively. $[v]_h$ denotes the $h$-th component of the vector $v$. With the exception of $\N$ and $\R$ (that denote the sets of natural and real numbers, respectively), we use blackboard bold symbols (such as $\mathbb{O}$) to denote  operators acting on $\leb$ or on $\lebn$. The induced operator norm are denoted by $\vertiii{\mathbb{O}}:=\sup_{\{f\mid \|f\|_{L^2}=1\}} \|\mathbb{O}f\|_{L^2}$ and $\vertiii{\mathbb{O}}:=\sup_{\{g\mid \|g\|_{L^2;\mathbb{R}^n}=1\}} \|\mathbb{O}g\|_{L^2;\mathbb{R}^n}$.  We denote by $\lambda_{\textup{max}}(\mathbb{L})$ and by $r(\mathbb{L})$   the largest eigenvalue  and the spectral radius of  the linear integral operator $\mathbb{L}f:=\int_0^1L(x,y)f(y)dy$ with symmetric kernel $L(x,y)=L(y,x)$.  We denote sets by using calligraphic symbols (such as $\mathcal{S}$) and the set of subsets of $\R^n$ by $2^{\R^n}.$ The symbol $\one$ denotes the vector of all ones in $\R^N$ and $1_{[0,1]}(x)$ the function constantly equal to one in $[0,1]$. $\mathbb{I}$ is the identity operator and $I$ the identity matrix.

\ssection{Graphon games: Nash equilibrium as a fixed point of the game operator}
\label{app:graphon_games}

In this appendix we derive an equivalent reformulation of the Nash equilibrium of any graphon game as a fixed point of a suitable operator which we term the \textit{game operator}. The analysis of the properties of such a game operator is key to prove the results of Section \ref{properties} on existence, uniqueness and continuity of the Nash equilibrium. 

\subsection{\textbf{ Reformulation as a fixed point}}

{\color{blue} We here consider the more general  case where strategies are vectors in $\mathbb{R}^n$ instead of scalars and the parameter $\theta$ is a vector of $\mathbb{R}^m$ instead of a scalar. Consequently,  a strategy profile $s:[0,1]\rightarrow \R^n$ is  a \textit{vector valued function}. In other words, $s(x)=[s_1(x),\ldots,s_n(x)]^\top$ for all $x\in[0,1]$. 
In the following, we require that any strategy profile is square integrable, that is $s(x)\in\lebn$.\footnote{ This  implies that
each component  is square integrable, that is, $s_k(x)\in\leb$ for all $k\in\{1,\ldots,n\}$. In fact for any $k\in\{1,\ldots,n\}$ it holds $\|s_k\|_{\lt}=\sqrt{\int_0^1 s_k(x)^2 dx}\le \sqrt{\int_0^1 \|s(x)\|^2 dx}= \|s\|_{\lnn}$. } 
}

To derive a fixed point characterization of the Nash equilibrium, we start by considering any strategy function $s\in \lebn$.  The corresponding local aggregate is

\begin{align*}
z(x\mid s)&:=\int_0^1 W(x,y)s(y)dy\\
&:= \left[\begin{array}{c}\int_0^1 W(x,y)s_1(y)dy \\\vdots \\ \int_0^1 W(x,y)s_n(y)dy \end{array}\right] = \left[\begin{array}{c}(\mathbb{W}s_1)(x) \\\vdots \\ (\mathbb{W}s_n)(x) \end{array}\right] =: (\mathbb{W}_ns)(x),
\end{align*}
where $\mathbb{W}_n:\lebn\rightarrow\lebn$ is defined by applying $\mathbb{W}$ component-wise.\footnote{ Note that $s(x)\in \lebn \Rightarrow z(x\mid s)\in\lebn$.
In fact 
{\scriptsize $\|z(x\!\mid\! s)\|^2_{L^2;\mathbb{R}^n}\! =  \int_0^1 \|z(x\mid s)\|^2 dx\textstyle=  \int_0^1 \sum_k [\int_0^1 W(x,y) s_k(y)dy]^2 dx\le$ $  \int_0^1 \sum_k (\int_0^1 [W(x,y)]^2 dy) (\int_0^1  [s_k(y)]^2dy) dx\textstyle \le \int_0^1  \sum_k (\int_0^1  [s_k(y)]^2dy) dx =  \sum_k (\int_0^1  [s_k(y)]^2dy) = \int_0^1 \sum_k   [s_k(y)]^2dy  $ }$=  \|s(x)\|^2_{L^2;\mathbb{R}^n}$,
where we used Cauchy-Schwartz, $W(x,y)^2\le1$ { and Fubini-Tonelli to switch the sum and integral in the second to last equality.}
}
Let us now define an operator $\mathbb{B}_{\theta}:\lebn\rightarrow \lebn$ defined point-wise as follows 
\begin{equation}
\label{eq:br}
(\mathbb{B}_{\theta}z)(x):= \arg\max_{\tilde s\in\mathtt{S}(x)} \J(\tilde s,z(x), {\blue \theta(x)}),
\end{equation}
 where $z(x)$ is any function of $\lebn$ (i.e., not necessarily $z(x\!\mid\! s)$). In words, $(\mathbb{B}_{\theta}z)(x)$ is the best response of agent $x$ to the \textit{fixed} local aggregate $z(x)$.  Note that, under Assumption~\ref{ass:cost}, such best response operator is well defined since  the maximization problem in \eqref{eq:br} has a unique solution.  The fact that, under the given assumptions, the codomain of the best response operator $\mathbb{B}_{\theta}$ is $\lebn$ will be proven in the next section.

We then see that a strategy profile $\bar s \in\lebn$ is a Nash equilibrium if and only if 
\begin{equation}\label{fixed}
\bar s = \mathbb{B}_{\theta}\mathbb{W}_n\bar s,
\end{equation}
that is, the function $\bar s$ is a fixed point of the composite operator $\mathbb{B}_{\theta}\mathbb{W}_n$, which we term the \textit{game operator}.

\subsection{\textbf{Properties of the game operator}}
Existence and uniqueness of a fixed point solving \eqref{fixed} depend on the properties of the composite game operator $\mathbb{B}_{\theta}\mathbb{W}_n$.    We start by studying the properties of $\mathbb{B}_{\theta}$ and $\mathbb{W}_n$ separately.  Lemmas \ref{lem:wn}, \ref{lem:b} and \ref{lem:ls} are then used in the proofs of Theorems \ref{thm:existence}, \ref{thm:unique} and \ref{thm:comp}.

We start by summarizing in the next proposition the main properties of the graphon operator, which follow from the symmetry of $W$ and   will be used in our subsequent analysis.

\begin{lemma}[Properties of $\mathbb{W}_n$] \label{lem:wn} The following  holds:
\textit{\begin{enumerate}
\item $\mathbb{W}_n$ is a linear, continuous, bounded and compact operator;
\item The eigenvalues of $\mathbb{W}_n$ coincide (besides  multiplicity) with those of $\mathbb{W}$ and are   real; 
\item   $\vertiii{\mathbb{W}_n}=\lambda_{\textup{max}}(\mathbb{W})$.
\end{enumerate}}
\end{lemma}
\begin{proof}
We first consider the case $n=1$ and show that  the statements above follow from well-known results
\begin{enumerate}
\item $\int_0^1\int_0^1 |W(x,y)|^2 dx dy\le \int_0^1\int_0^1 1 dx dy \le 1$, hence $\mathbb{W}$ is a  Hilbert-Schmidt integral operator with Hilbert-Schmidt kernel $W$ (and is thus a continuous, bounded and compact operator). See also \cite[Section 7.5]{Lovasz2012}.  
\item Since $W(x,y)$ is symmetric, $\mathbb{W}$ is self-adjoint   (see e.g. \cite[Example 6.5.9]{hutson2005applications}). The spectrum of bounded self-adjoint operators is real \cite[Theorem 6.6.3]{hutson2005applications}. See also \cite[Section 7.5]{Lovasz2012}. 
\item Let  $r(\mathbb{W})$ and $\lambda_{\textup{max}}(\mathbb{W})$   be  the spectral radius and the largest eigenvalue of  $\mathbb{W}$.
For bounded self-adjoint operators it holds $\vertiii{\mathbb{W}}=r(\mathbb{W})$, \cite[Theorem 6.6.7]{hutson2005applications}. Since $\mathbb{W}$ is linear compact and positive (with respect to the  total order cone of nonnegative functions in $\leb$), by Krein-Rutman theorem $r(\mathbb{W})>0$  is an eigenvalue \cite[Proposition 7.26]{zeidler1985nonlinear}. Since all eigenvalues are real it must be $r(\mathbb{W})=\lambda_{\textup{max}}(\mathbb{W})$. Hence 
$\vertiii{\mathbb{W}}:= \sup_{f\in\leb, \|f\|_{L^2}=1} \|\mathbb{W}f\|_{L^2}= r(\mathbb{W})= \lambda_{\textup{max}}(\mathbb{W})$.
\end{enumerate}
{\color{blue} The extension to $n>1$  is  immediate since $\mathbb{W}_n$ acts independently on each component.}
\end{proof}

\begin{lemma}[Properties of $\mathbb{B}_{\theta}$] \label{lem:b}  Suppose that Assumption \ref{ass:cost}   holds.  Then the following holds:
\textit{\begin{enumerate}
\item {\blue $\mathbb{B}_{\theta}$ is a Lipschitz operator. That is, for any $f_1,f_2\in\lebn$ and $\theta_1,\theta_2 \in \lebm$
\[
 \|\mathbb{B}_{\theta_1}f_1-\mathbb{B}_{\theta_2}f_2\|_{L^2;\mathbb{R}^n}\le \frac{1}{\alpha_\J}(\ell_\J\|f_1-f_2\|_{L^2;\mathbb{R}^n}+\ell_\theta\|\theta_1-\theta_2\|_{L^2;\mathbb{R}^m});
\]}
 \item $\mathbb{B}_{\theta}$ is a  continuous operator;
 \item  Suppose  further that Assumption \ref{ass:constraint}A) holds, then the codomain of $\mathbb{B}_{\theta}$ is $\lebn$;
\item Suppose  further that Assumption \ref{ass:constraint}B) holds, then  the codomain of $\mathbb{B}_{\theta}$ is contained in  
\begin{equation}\label{eq:Ls}
\ls:=\{f\in\lebn \mid \|f\|_{L^2;\mathbb{R}^n}\le s_\textup{max}\},
\end{equation}
where $s_\textup{max}$ is as defined in  Assumption \ref{ass:constraint}.
\end{enumerate}}
\end{lemma}

\begin{proof}
\begin{enumerate}
\item {\blue Take any $f_1,f_2\in\lebn$ and $\theta_1,\theta_2 \in \lebm$.
For any  $x\in[0,1]$ we get

\begin{equation}\begin{aligned}
&\|(\mathbb{B}_{\theta_1}f_1)(x) - (\mathbb{B}_{\theta_2}f_2)(x)\|= \|\arg\max_{\tilde s\in\mathtt{S}(x)} \J(\tilde s,f_1(x),\theta_1(x))- \arg\max_{\tilde s\in\mathtt{S}(x)} \J(\tilde s,f_2(x),\theta_2(x))\| \\
&\le \frac{1}{\alpha_\J}\|-\nabla_s \J((\mathbb{B}_{\theta_2}f_2)(x) ,f_1(x),\theta_1(x))+ \nabla_s \J((\mathbb{B}_{\theta_2}f_2)(x) ,f_2(x),\theta_2(x))\|\\&\le \frac{1}{\alpha_\J}(\ell_\J\|f_1(x)-f_2(x)\|+\ell_\theta\|\theta_1(x)-\theta_2(x)\|).
\end{aligned}\label{eq:continuity}
\end{equation}

The first inequality in \eqref{eq:continuity} can be proven by reformulating the optimization problem in \eqref{eq:br} as the variational inequality VI$(\mathtt{S}(x), -\nabla_s \J(\cdot, z(x), \theta(x)))$. By Assumption~\ref{ass:cost}, the operator  $-\nabla_s \J(\cdot, z,\theta)$ is strongly monotone with constant $\alpha_\J$ for all $z\in\R^n,\theta\in\R^m$, \cite[Equation (12)]{scutari2010convex} . The result then follows from a known bound on the distance of the solution of strongly monotone variational inequalities \cite[Theorem 1.14]{nagurney2013network}. The second inequality in \eqref{eq:continuity} comes from the assumption that $\nabla_s \J(s, z, \theta)$ is uniformly Lipschitz in $[z,\theta]$ with constants  $\ell_\J,\ell_\theta$ for all $s\in\R^n$.
Let us now compute $\|\mathbb{B}_{\theta_1}f_1 - \mathbb{B}_{\theta_2}f_2\|_{L^2;\mathbb{R}^n}.$

For simplicity define $h(x):=\|(\mathbb{B}_{\theta_1}f_1)(x) - (\mathbb{B}_{\theta_2}f_2)(x)\| $, $h_f(x):=\frac{\ell_\J}{\alpha_\J}\|f_1(x)-f_2(x)\|$ and $h_\theta(x):=\frac{\ell_\theta}{\alpha_\J}\|\theta_1(x)-\theta_2(x)\|$ for all $x\in[0,1]$. 
By \eqref{eq:continuity}, $0\le h(x)\le h_f(x)+h_\theta(x)$ for all $x\in[0,1]$. Hence
$$ \|h(x)\|_{L^2} \le \|h_f(x)+h_\theta(x)\|_{L^2} \le \|h_f(x)\|_{L^2}+\|h_\theta(x)\|_{L^2}.$$

The conclusion follows from $\|h(x)\|_{L^2}=\|\mathbb{B}_{\theta_1}f_1 - \mathbb{B}_{\theta_2}f_2\|_{L^2;\mathbb{R}^n}$, $ \|h_f(x)\|_{L^2}=\frac{\ell_\J}{\alpha_\J}\|f_1-f_2\|_{L^2;\mathbb{R}^n}$ and $ \|h_\theta(x)\|_{L^2} =\frac{\ell_\theta}{\alpha_\J}\|\theta_1-\theta_2\|_{L^2;\mathbb{R}^m}. $}
\item 
Lipschitz continuity implies continuity.
\item We need to show that for any $z\in\lebn$, $\|\mathbb{B}_{\theta}z\|_{L^2;\mathbb{R}^n}<\infty$. Consider the function $\hat z(x):=\hat z$ for all $x\in[0,1]$, where $\hat z$ is as in Assumption \ref{ass:constraint}A). Note that  $\hat z\in\lebn$ and

{\blue \begin{align*}
\|\mathbb{B}_{\theta}\hat z\|^2_{L^2;\mathbb{R}^n}&=\int_0^1  \|(\mathbb{B}_{\theta}\hat z)(x)\|^2dx = \int_0^1 \| \arg\max_{\tilde s\in\mathtt{S}(x)} \J(\tilde s,\hat z,\theta(x)) \|^2dx   \le M^2.
\end{align*}}

Consider now any $z\in\lebn$. We have

\begin{align*}
\|\mathbb{B}_{\theta} z\|_{L^2;\mathbb{R}^n} & = \|\mathbb{B}_{\theta} z - \mathbb{B}_{\theta} \hat z + \mathbb{B}_{\theta}\hat  z\|_{L^2;\mathbb{R}^n} \le  \|\mathbb{B}_{\theta} z - \mathbb{B}_{\theta} \hat z\|_{L^2;\mathbb{R}^n} + \|\mathbb{B}_{\theta}\hat  z\|_{L^2;\mathbb{R}^n}\\
&\le  \left(\frac{\ell_\J}{\alpha_\J}\right) \|\hat z-z\|_{L^2;\mathbb{R}^n} +M\le  \left(\frac{\ell_\J}{\alpha_\J}\right) (\|\hat z\|_{L^2;\mathbb{R}^n}+\|z\|_{L^2;\mathbb{R}^n}) +M <\infty,
\end{align*}

where the second inequality follows from statement 1).
\item Under Assumption \ref{ass:constraint}B) for any $x\in[0,1]$, $(\mathbb{B}_{\theta}z)(x) \in\mathtt{S}(x)\subseteq \mathcal{S}$ hence

\begin{align*}
\|\mathbb{B}_{\theta}z\|^2_{L^2;\mathbb{R}^n}&= \int_0^1 \|(\mathbb{B}_{\theta}z)(x)\|^2 dx  \le  \int_0^1 s_\textup{max}^2 dx =s_\textup{max}^2.
\end{align*}

Consequently for any $z\in\lebn$, $\mathbb{B}_{\theta}z\in \ls$.
\end{enumerate}
\end{proof}

Finally, we study the properties of  $\ls$ as defined in \eqref{eq:Ls}.
\begin{lemma}[Properties of $\ls$] \label{lem:ls} \textit{For any non-empty compact set $\mathcal{S}\subset{\R^n}$, the set $\ls$ in \eqref{eq:Ls} is a non-empty, convex, closed and bounded subset of $\lebn$.}
\end{lemma}
\begin{proof}
Since $\mathcal{S}$ is non-empty and compact $s_\textup{max}$ is well defined. This immediately implies that $\ls$ is non-empty. Given two functions $f,g\in\ls$ and any $\mu\in[0,1]$
\[
\|\mu f\!+\!(1\!-\!\mu)g\|_{L^2;\mathbb{R}^n}\le \mu \|f(x)\|_{L^2;\mathbb{R}^n}+(1-\mu)\|g(x)\|_{L^2;\mathbb{R}^n}\le \mu s_\textup{max} +(1-\mu) s_\textup{max} =s_\textup{max}.
\]
 Hence $\mu f+(1-\mu)g \in\ls$ and $\ls$ is convex.   $\ls$ is closed and bounded by definition.
\end{proof}

\ssection{Omitted proofs}
\label{appC}

\subsection{\textbf{Section \ref{graphon_games}: Omitted proofs }} 
\vspace{0.2cm}
\textbf{Proof of Theorem \ref{thm:existence}}

We aim at applying Schauder fixed point theorem \cite[Theorem 4.1.1]{smart1980fixed} to $\mathbb{B}_{\theta}\mathbb{W}_n:\ls \rightarrow K:=(\mathbb{B}_{\theta}\mathbb{W}_n(\ls))^{cl} $. 
\begin{enumerate}
\item By Lemma \ref{lem:ls}, $\ls$ is non-empty, convex, closed and bounded. 
\item In Lemma \ref{lem:wn} and \ref{lem:b} we have proven that both $\mathbb{W}_n$ and $\mathbb{B}_{\theta}$ are continuous operators, hence $\mathbb{B}_{\theta}\mathbb{W}_n$ is continuous.
\item We next show $K\subseteq \ls$. Since $\mathbb{W}_n:\ls \rightarrow \lebn$ and   $\mathbb{B}_{\theta}:\lebn \rightarrow \ls$ it holds $ \mathbb{B}_{\theta}\mathbb{W}_n(\ls) \subseteq \ls$. Since $\ls$ is closed, $ K =(\mathbb{B}_{\theta}\mathbb{W}_n(\ls))^{cl}  \subseteq (\ls)^{cl}=\ls$.
\item Finally, we  show that $K$ is compact. To this end note that  $\mathbb{W}_n$ is a compact operator by Lemma \ref{lem:wn} and $\ls$ is bounded, hence $(\mathbb{W}_n(\ls))^{cl}$ is  compact \cite[Definition 7.2.1]{hutson2005applications}. We proved in Lemma \ref{lem:b} that  $\mathbb{B}_{\theta}$ is Lipschitz  (and thus  continuous), consequently $\mathbb{B}_{\theta}((\mathbb{W}_n(\ls))^{cl})$ is  compact \cite[Theorem 2.34]{aliprantisinfinite}.  Clearly $\mathbb{B}_{\theta}(\mathbb{W}_n(\ls)) \subseteq \mathbb{B}_{\theta}((\mathbb{W}_n(\ls))^{cl})$ and thus $K:=(\mathbb{B}_{\theta}(\mathbb{W}_n(\ls)))^{cl} \subseteq (\mathbb{B}_{\theta}((\mathbb{W}_n(\ls))^{cl}))^{cl}=\mathbb{B}_{\theta}((\mathbb{W}_n(\ls))^{cl}))$. $K$ is thus a closed subset of a compact set, which implies that $K$ is compact \cite[pg. 40]{aliprantisinfinite}. 
\end{enumerate}
Schauder fixed point theorem  thus guarantees the existence of a fixed point.

\vspace{0.2cm}
\textbf{Proof of Theorem \ref{thm:unique}}

 We show that under the assumptions of this theorem the game operator is a contraction (i.e. Lipschitz with constant strictly less than one) in the Hilbert space $\lebn$. The conclusion then follows from Banach fixed point theorem \cite[Theorem 4.3.4]{smart1980fixed}. For any $f,g\in\lebn$, 

  \begin{align*}
 \|\mathbb{B}_{\theta} \mathbb{W}_nf-\mathbb{B}_{\theta}\mathbb{W}_ng\|_{L^2;\mathbb{R}^n}&\le \frac{\ell_\J}{\alpha_\J}\|\mathbb{W}_nf-\mathbb{W}_ng\|_{L^2;\mathbb{R}^n}= \frac{\ell_\J}{\alpha_\J}\|\mathbb{W}_n(f-g)\|_{L^2;\mathbb{R}^n}
 \\&\le  \frac{\ell_\J}{\alpha_\J}\vertiii{\mathbb{W}_n}\|f-g\|_{L^2;\mathbb{R}^n} =  \frac{\ell_\J}{\alpha_\J}\lambda_{\textup{max}}(\mathbb{W})\|f-g\|_{L^2;\mathbb{R}^n},
 \end{align*}
 
 where we used Lemma~\ref{lem:b} for the first inequality, the fact that $\mathbb{W}_n$ is linear in the first equality and the fact that  $\vertiii{\mathbb{W}_n}=\lambda_{\textup{max}}(\mathbb{W})$, as proven in Lemma \ref{lem:wn}, in the last line. The conclusion follows from Assumption \ref{cond}.

\vspace{0.2cm}
\textbf{Proof of Theorem \ref{thm:comp}}

{\color{blue} In the vector case condition \eqref{eq:K} becomes

\begin{equation}
\blue \|\bar s -\tilde s\|_{L^2;\mathbb{R}^n} \le \frac{1/\alpha_\J }{1-\ell_\J/\alpha_\J \lambda_{\textup{max}}(\mathbb{W})}\left(\ell_\J\vertiii{\mathbb{W} -\tilde{\mathbb{W}}}s_\textup{max} +\ell_\theta \|\theta-\tilde \theta\|_{L^2;\mathbb{R}^m}\right).
\end{equation}}

{\blue To prove this condition, note that by the equivalent characterization of Nash equilibria in terms of fixed points, it holds $\bar s =\mathbb{B}_{\theta}\mathbb{W}_n\bar s $ and $\tilde s =\mathbb{B}_{\tilde \theta}\tilde{\mathbb{W}_n}\tilde s $ hence

\begin{align*}
\| \bar s-\tilde s\|_{L^2;\mathbb{R}^n}&=\| \mathbb{B}_{\theta}\mathbb{W}_n\bar s -\mathbb{B}_{\tilde \theta}\tilde{\mathbb{W}}_n\tilde s \|_{L^2;\mathbb{R}^n} \le \frac{\ell_\J}{\alpha_\J}\| \mathbb{W}_n\bar s -\tilde{\mathbb{W}}_n\tilde s \|_{L^2;\mathbb{R}^n} + \frac{\ell_\theta}{\alpha_\J} \|\theta-\tilde \theta\|_{L^2;\mathbb{R}^m}\\&
\le  \frac{\ell_\J}{\alpha_\J}\| \mathbb{W}_n\bar s -\mathbb{W}_n\tilde s\|_{L^2;\mathbb{R}^n} + \frac{\ell_\J}{\alpha_\J}\| \mathbb{W}_n\tilde s -\tilde{\mathbb{W}}_n\tilde s \|_{L^2;\mathbb{R}^n}+ \frac{\ell_\theta}{\alpha_\J} \|\theta-\tilde \theta\|_{L^2;\mathbb{R}^m}\\&\le  \frac{\ell_\J}{\alpha_\J}\vertiii{ \mathbb{W}_n} \|\bar s -\tilde s\|_{L^2;\mathbb{R}^n} + \frac{\ell_\J}{\alpha_\J} \vertiii{\mathbb{W}_n -\tilde{\mathbb{W}}_n}\|\tilde s \|_{L^2;\mathbb{R}^n}+ \frac{\ell_\theta}{\alpha_\J} \|\theta-\tilde \theta\|_{L^2;\mathbb{R}^m}\\
&=  \frac{\ell_\J}{\alpha_\J}\lambda_{\textup{max}}(\mathbb{W}) \|\bar s -\tilde s\|_{L^2;\mathbb{R}^n} +  \left(\frac{\ell_\J}{\alpha_\J}\vertiii{\mathbb{W} -\tilde{\mathbb{W}}}\|\tilde s \|_{L^2;\mathbb{R}^n}+ \frac{\ell_\theta}{\alpha_\J}\|\theta-\tilde \theta\|_{L^2;\mathbb{R}^m}\right)
\end{align*}}
where we used that $\mathbb{B}_{\theta}$ is Lipschitz, as proven in Lemma \ref{lem:b}, the fact that  $\vertiii{ \mathbb{W}_n}=\lambda_{\textup{max}}(\mathbb{W})$  and the fact that $\vertiii{\mathbb{W}_n -\tilde{\mathbb{W}}_n}=\vertiii{\mathbb{W} -\tilde{\mathbb{W}}}$. 
%
%
%
The  conclusion follows from the fact that  $ 1-\ell_\J/\alpha_\J \lambda_{\textup{max}}(\mathbb{W})>0$ by Assumption \ref{cond} hence
\begin{align}\label{cont_step}
\| \bar s-\tilde s\|_{L^2;\mathbb{R}^n}\le \frac{1}{1-  \frac{\ell_\J}{\alpha_\J}\lambda_{\textup{max}}(\mathbb{W})}   \left(\frac{\ell_\J}{\alpha_\J}\vertiii{\mathbb{W} -\tilde{\mathbb{W}}}\|\tilde s \|_{L^2;\mathbb{R}^n}+ \frac{\ell_\theta}{\alpha_\J}\|\theta-\tilde \theta\|_{L^2;\mathbb{R}^m}\right)
\end{align}
and the fact that, under Assumption \ref{ass:constraint}B), $\|\tilde s \|_{L^2;\mathbb{R}^n} \le s_\textup{max}$, as proven in Lemma~\ref{lem:b}.

\subsection{\textbf{Section \ref{sec:asymptotic} and \ref{sec:sparse}: Omitted proofs }} \label{sec:conv_proof}

\vspace{0.5cm}
\textbf{Generalized statement and proof of Theorem \ref{thm:oneNash}}

{\blue
We here report a  more general version of the statement of Theorem \ref{thm:oneNash} that includes the sparsity parameter $\kappa_N$, as discussed in Section \ref{sec:sparse}. Theorem \ref{thm:oneNash} is obtained as a special case by setting $\kappa_N=1$.\\

\textbf{Theorem \ref{thm:oneNash} (generalized).}
\textit{A vector $\bar s_{[N]}\in \R^{Nn}$ is a Nash equilibrium of the  game $\mathcal{G}_\kappa^{[N]}(\{\mathcal{S}^i\}_{i=1}^N, \J,\{\theta^i\}_{i=1}^N, P^{[N]})$ with $N$ players, payoff function $\J$ as in \eqref{gameyG} for some sparsity parameter $\kappa_N$,  strategy sets $\mathcal{S}^i$, parameters $\theta^i$ and graph $P^{[N]}$ if and only if the  corresponding  step function equilibrium $\bar s_{[N]}(x)\in \lebn$ is a  Nash equilibrium  of the graphon game  $\mathcal{G}(\mathtt{S}^{[N]}, \J, \theta^{[N]}, W_\kappa^{[N]})$ with   payoff function as in \eqref{eq:br_naJ}, set valued function  $\mathtt{S}^{[N]}(x):=\mathcal{S}^i$ for all $x\in \mathcal{U}_i^{[N]}$, parameter function $\theta^{[N]}(x):=\theta^i$ for all $x\in \mathcal{U}_i^{[N]}$ and step function graphon $W_\kappa^{[N]}$  corresponding to  $\frac{P^{[N]}}{\kappa_N}$.}

\begin{proof}
 Suppose that $\bar s_{[N]}$ is a Nash equilibrium of  $\mathcal{G}(\mathtt{S}^{[N]}, \J, \theta^{[N]}, W_\kappa^{[N]})$. Since $W_\kappa^{[N]}$ is a step function over the partition $\mathcal{U}^{[N]}$, the    aggregate $\bar z_\kappa(x)=\int_0^1 W_\kappa^{[N]}(x,y)\bar s_{[N]}(y)dy $ is a step function with respect to the same partition. Let $\bar z^i_\kappa$ be the value of $\bar z_\kappa(x)$ in $\mathcal{U}_i^{[N]}$ and recall that $\theta^{[N]}(x)=\theta^i$ in $\mathcal{U}_i^{[N]}$.  From the definition of Nash equilibrium for the graphon game

$$\bar s_{[N]}(x)=\arg\max_{s\in\mathtt{S}^{[N]}(x)} \J(s, \bar z_\kappa(x),\theta^{[N]}(x))=\arg\max_{s\in\mathcal{S}^i} \J(s, \bar z_\kappa^i,\theta^i) \mbox{ for all } x\in \mathcal{U}_i^{[N]}.$$
Consequently, also $\bar s_{[N]}(x)$ is a step function with respect to $\mathcal{U}^{[N]}$. Let $\bar s_{[N]}^i$ be the value of $\bar s_{[N]}(x)$ in $\mathcal{U}_i^{[N]}$. 
Then $\bar z^i_\kappa =\int_0^1 W_\kappa^{[N]}(x,y)\bar s_{[N]}(y)dy =\frac1N \sum_{j=1}^N  \frac{P_{ij}^{[N]}}{\kappa_N} \bar s_{[N]}^j$ and $\bar s_{[N]}(x)$ is a Nash equilibrium of the graphon game if and only if for each $i\in\{1,\ldots,N\}$ it holds

$$\bar s_{[N]}^i=\arg\max_{s\in\mathcal{S}^i} \J(s,\bar z_\kappa^i,\theta^i), \quad \bar z_\kappa^i =\frac{1}{\kappa_N N} \sum_{j=1}^N  P_{ij}^{[N]} \bar s_{[N]}^j.$$
The latter is the definition of Nash equilibrium in the sampled network game with network $P^{[N]}$, thus concluding the proof.
\end{proof}}\\

\textbf{Generalized statement and proof of Theorem \ref{thm:dist}}

{\blue
We here report a  more general version of the statement of Theorem \ref{thm:dist} that includes the sparsity parameter $\kappa_N$, as discussed in Section \ref{sec:sparse}. Theorem \ref{thm:dist} is obtained as a special case by setting $\kappa_N=1$.\\}

\textbf{Theorem \ref{thm:dist} (generalized).} 
\textit{Consider a graphon game $\mathcal{G}(\mathtt{S}, \J, {\blue \theta}, W)$ where each player has  homogeneous strategy set, i.e.,  $\mathtt{S}(x)=\mathcal{S}$ for all $x\in[0,1]$. Suppose that $\mathcal{G}$ satisfies Assumptions \ref{ass:cost}, \ref{ass:constraint}B), \ref{cond} and \ref{lipschitz}.  Let $\bar s$ be its unique Nash equilibrium  and fix  any sequences $\{\delta_N, {\blue \kappa_N}\}_{N=1}^\infty$ such that $\delta_N \le e^{-1}$ and $\frac{\log( N/\delta_N)}{N{ \blue \kappa_N}} \rightarrow 0$.  Let $\bar s^{[N]}_{w}$ be an arbitrary step function equilibrium of the sampled network game $\mathcal{G}^{[N]}(\{\mathcal{S}\}_{i=1}^N, \J,{\blue \{{\theta}(t^i)\}_{i=1}^N}, P^{[N]}_{w})$, as introduced in  Section \ref{step2} and 
 $\bar s^{[N]}_{s}$ be an arbitrary step function equilibrium of the sampled network game $\mathcal{G}_{\kappa}^{[N]}(\{\mathcal{S}\}_{i=1}^N, \J,{\blue \{{\theta}(t^i)\}_{i=1}^N}, P^{[N]}_{s})$, as introduced in  Section \ref{sec:sparse}.
Then with probability at least $1-2\delta_N$, for $N$ large enough,  it holds
\begin{align*}\|\bar s^{[N]}_{w/s}-\bar s\|_{L^2;\mathbb{R}^n}  \le K\rho_{\blue \kappa}(N)
\end{align*}
for $\rho_{\blue \kappa}(N):=\left(2s_\textup{max}  \sqrt{(L^2-\Omega^2)d_N^2+\Omega d_N} +s_\textup{max}\sqrt{\frac{4\log(2N/\delta_N)}{{\blue\kappa_N}N}} +{\blue \sqrt{(Ld_N)^2+8\Omega d_N \theta^2_\textup{max}}}\right)$,\\ $ K=\frac{{\blue \max\{\ell_\J,\ell_\theta\}}/\alpha_\J}{1-\ell_\J/\alpha_\J \lambda_{\textup{max}}(\mathbb{W})} $,
 and $d_N:=\frac1N+\sqrt{\frac{8\log(N/\delta_N)}{N}}\rightarrow 0$ as $N\rightarrow \infty$.  
Consequently, $\|\bar s^{[N]}_{w/s}-\bar s\|_{L^2;\mathbb{R}^n} \rightarrow 0$  almost surely when $N\rightarrow \infty$.
}

\begin{proof}
{\blue Let $\theta^{[N]}$ be the step function corresponding to the vector $[\theta(t^i)]_{i=1}^N$} and $W_w^{[N]}$, $W_s^{[N]}$ be the   step function graphons  corresponding to  $P_w^{[N]}$ and {\blue $\frac{P_s^{[N]}}{\kappa_N}$}, respectively, so that
 $\bar s^{[N]}_{w}$ and $\bar s^{[N]}_{s}$ are the equilibria of the graphon games played over the  graphon $W_w^{[N]}$ and $W_s^{[N]}$ {\blue with parameter function  $\theta^{[N]}$}. By Theorem~\ref{thm:comp} it  follows

\begin{equation}\label{eq:boundA}
\begin{aligned} 
\blue \|\bar s^{[N]}_{w/s} -\bar s\|_{L^2;\mathbb{R}^n} \le  K \left(\vertiii{\mathbb{W}^{[N]}_{w/s}-\mathbb{W}}s_\textup{max}+ \|\theta^{[N]}-\theta\|_{L^2;\mathbb{R}^m}\right). \end{aligned}
\end{equation}

 The bound on $\|\bar s^{[N]}_{w/s}- \bar s\|_{L^2;\mathbb{R}^n}$   follows from  \eqref{eq:boundA} and the fact that for $N$ large enough with probability at least $1-2\delta_N$ 
{\blue  \[\|\theta^{[N]}-\theta\|_{L^2;\mathbb{R}^m}\le \rho_\theta(N),\]}
\[ \vertiii{\mathbb{W}^{[N]}_{w/s}-\mathbb{W}} \le \rho_W(N),\]
for $\rho_\theta(N)$ and $\rho_W(N)$ as defined in Lemma \ref{lem:norm} in online Appendix \ref{aux}.
 The latter fact is proven in \cite[Theorem 1]{graphons} (reported in Lemma \ref{lem:norm} in online Appendix \ref{aux}) and follows from the fact that the $\{\type^i\}_{i=1}^N$ are the ordered statistic of $N$ uniform samples from $[0,1]$ combined with the fact that $W$ is piecewise Lipschitz by Assumption \ref{lipschitz}. 
 
Overall, there exists $M>0$ such that  for $N$ sufficiently large with probability at least $1-2\delta_N$ it holds
\[\blue \|\bar s^{[N]}_{w/s}- \bar s\|_{L^2;\mathbb{R}^n}\le M \left(\left(\frac{\log(N/\delta_N)}{{\blue }N} \right)^\frac{1}{4}+ \left(\frac{\log(N/\delta_N)}{{ \kappa_N}N} \right)^\frac{1}{2}\right)=:\Theta_N.\]
To prove almost sure convergence, let us define the infinite sequence of events
\[\mathcal{E}_N:=\left\{ \|\bar s^{[N]}_{w/s}- \bar s\|_{L^2;\mathbb{R}^n}> {\blue \Theta_N} \right\}.\]
 It follows that $\textup{Pr}\left[ \mathcal{E}_N \right]< 2\delta_N$. 
 {\blue Note that $\frac{\log(N/\delta_N)}{N{\blue \kappa_N}} \rightarrow 0$ implies $\frac{\log(N)}{N{\blue \kappa_N}} \rightarrow 0$ (since $\delta_N\le 1$).  Hence $\delta_N=\frac{1}{N^2}$ is also an admissible choice and leads to}
\[    \textstyle \sum_{N=1}^\infty \textup{Pr}\left[ \mathcal{E}_N \right]<  \sum_{N=1}^\infty \frac{2}{N^2} <\infty.\]
By  Borel-Cantelli lemma there exists  a positive integer $\bar N$ such that for all $N\ge \bar N$, the complement of $\mathcal{E}_N$, i.e., 
$\|\bar s^{[N]}_{w/s}- \bar s\|_{L^2;\mathbb{R}^n}\le {\blue \Theta_N}$, holds almost surely. Since, $ \Theta_N \rightarrow 0$ we obtain $\|\bar s^{[N]}_{w/s}- \bar s\|_{L^2;\mathbb{R}^n} \rightarrow 0 $ almost surely. 
\end{proof}

{\blue \subsection{\textbf{Section \ref{sec:int}: Omitted proofs }} 
\vspace{0.2cm}
\begin{lemma}\label{lem:LQ_bound}
Let $\bar s^{[N]}_{w/s}$ be an arbitrary step function equilibrium of the sampled network game $\mathcal{G}^{[N]}(\{\mathbb{R}_{\ge0}\}_{i=1}^N, \J,{\blue \{{\thetaL}(t^i)\}_{i=1}^N}, P^{[N]}_{w/s})$, as introduced in  Section \ref{step2}, with linear quadratic payoff $\J$ as in  \eqref{costL} and set $0<\alpha<\frac{1}{\lambda_\textup{max}(\mathbb{W})}$. Then there exists $M_s$ such that for any admissible confidence sequence $\{\delta_N\}_{N=1}^\infty$ and $N$ large, with probability $1-2\delta_N$,
$\|\bar s^{[N]}_{w/s}\|_{L^2}\le M_s.$
\end{lemma}
\begin{proof}
Fix $\bar k \in(1,\frac{1}{\alpha \lambda_\textup{max}(\mathbb{W})})$ (this interval has non empty interior by assumption).
Since by Lemma \ref{lem:norm} (given in  online  Appendix \ref{aux})   with probability $1-2\delta_N$ for $N$ large $|\lambda_\textup{max}(\mathbb{W}_{w/s}^{[N]})- \lambda_\textup{max}(\mathbb{W})|\le\rho_W(N)$, it follows that for $N$ large enough, $\lambda_\textup{max}(\mathbb{W}_{w/s}^{[N]})\le \bar k \lambda_\textup{max}(\mathbb{W})$ and by the formula for the equilibrium in linear quadratic graphon games derived in Example \ref{lqgg}
\begin{align*}
\textstyle
\|\bar s^{[N]}_{w/s}\|_{L^2}& \le \| {(\mathbb{I}-\alpha \mathbb{W}_{w/s}^{[N]})^{-1}}\thetaL^{[N]}\|_{L^2} \le  \frac{1}{1-\alpha \lambda_\textup{max}(\mathbb{W}_{w/s}^{[N]})} \| \thetaL^{[N]}\|_{L^2}  \\&\le \frac{1}{1-\alpha\bar k  \lambda_\textup{max}(\mathbb{W})} 2\| \thetaL\|_{L^2} =:M_s,
\end{align*}
where we used that for large $N$, $\| \thetaL^{[N]}\|_{L^2} \le 2\| \thetaL\|_{L^2}$ since  $\| \thetaL^{[N]}- \thetaL\|_{L^2}\rightarrow 0$ as shown in Lemma \ref{lem:norm} in the online Appendix \ref{aux}. 

\end{proof}

 \textbf{Proof of Theorem \ref{thm:inter}}\\
We start by noting that Problem \eqref{finite_opt} can be equivalently reformulated as a problem in the space of functions instead of vectors by using the equivalent reformulation given in Section~\ref{net_graph}. So that 
\begin{equation}\label{finite_opt_fun}
\begin{aligned}
T^{[N]}_{\textup{opt}}:=\max_{\hat\thetaL^{[N]} \in L^{[N]}}&\quad \frac{1}{2}\|\bar s^{[N]}_{\hat\thetaL^{[N]}}\|^2_{L^2},\\
\textup{s.t.}&\quad \bar s^{[N]}_{\hat\thetaL^{[N]}} =\textup{Nash equilibrium of } \mathcal{G}(\{\mathbb{R}_{\ge0}\}_{i=1}^N, \J,\thetaL^{[N]}+\hat \thetaL^{[N]}, W_{w/s}^{[N]}),\\
&\quad \textstyle  \|\hat {\thetaL}^{[N]}\|^2_{L^2}\le C,
\end{aligned}
\end{equation}
where we used $L^{[N]}$ to denote the subspace of $\leb$ composed by functions that are piecewise constant w.r.t. the partition $\{\mathcal{U}^{[N]}_i\}_{i=1}^N$ and for simplicity we used the same symbol $\thetaL^{[N]}$ to denote the  vector $[\thetaL(t^i)]_{i=1}^N$ and its corresponding piecewise constant function. \\ Similarly to the proof of Lemma \ref{lem:LQ_bound}, it can be shown that with probability $1-2\delta_N$ both  $\|\bar s^{[N]}_{\hat\thetaL^{[N]}}\|_{L^2}$ and $\|\bar s_{\hat \thetaL}\|_{L^2}$ (as defined in \eqref{finite_opt_fun} and in \eqref{infinite_optt}, respectively)  can be bounded by some constant $M_s$ for any feasible $\hat\thetaL^{[N]},\hat\thetaL$. 
 
Let $\hat\thetaL^{[N]}_{\textup{opt}}$ be an optimizer of Problem \eqref{finite_opt_fun} and note that in the interval $[0,M_s]$ the function $g(\tau)=\frac12\tau^2$ is Lipschitz continuous with constant $M_s$. Hence

\begin{align*}
\quad T^{[N]}(\hat\thetaL^{[N]}_\textup{graphon})&= g(\|\bar s^{[N]}_{\hat\thetaL^{[N]}_\textup{graphon}}\|_{L^2}) \ge g(\|\bar s_{\hat\thetaL^{*}}\|_{L^2}) -\underbrace{M_s( |\|\bar s^{[N]}_{\hat\thetaL^{[N]}_\textup{graphon}}\|_{L^2} - \|\bar s_{\hat\thetaL^{*}}\|_{L^2} | )}_{T_1} \\
&\ge g(\|\bar s_{\hat\thetaL^{[N]}_{\textup{opt}}}\|_{L^2}) - T_1
 \\
 &\ge g(\|\bar s^{[N]}_{\hat\thetaL^{[N]}_{\textup{opt}}}\|_{L^2}) - T_1- \underbrace{M_s( |\|\bar s^{[N]}_{\hat\thetaL^{[N]}_{\textup{opt}}}\|_{L^2}- \|\bar s_{\hat\thetaL^{[N]}_{\textup{opt}}} \|_{L^2}|)}_{T_2}\\
 &= T^{[N]}_{\textup{opt}}- T_1-T_2,
\end{align*}
where the second inequality comes from the fact that $\hat\thetaL^{[N]}_{\textup{opt}}$ is a feasible point of problem \eqref{infinite_optt} and  $\hat\thetaL^{*}$ is the corresponding optimizer.
To bound the terms $T_1$ and $T_2$, note that for any function $\hat{\thetaL_1},\hat{\thetaL_2}\in \leb$ it holds $|\|\hat{\thetaL_1}\|_{L^2}-\|\hat{\thetaL_2}\|_{L^2}|\le \|\hat{\thetaL_1}-\hat{\thetaL_2}\|_{L^2}$ and by the proof of Theorem~\ref{thm:comp} (Equation \eqref{cont_step}) and Lemma \ref{lem:norm}, with probability $1-2\delta_N$
\begin{align*}
\|\bar s^{[N]}_{\hat{\thetaL_1}}-\bar s_{\hat{\thetaL_2}}\|_{L^2} &\le K \left( \vertiii{\mathbb{W}^{[N]}_{w/s}-\mathbb{W}} \|\bar s^{[N]}_{\hat{\thetaL_1}}\|_{L^2}+ \| (\thetaL^{[N]}+ \hat{\thetaL_1}) - (\thetaL +\hat{\thetaL_2})\|_{L^2}\right)\\
&\le K \left( \rho_W(N) M_s+ \| \thetaL^{[N]} -\thetaL \|_{L^2}+ \|  \hat{\thetaL_1} - \hat{\thetaL_2}\|_{L^2}\right)\\
&\le K \left( \rho_W(N)M_s+ \rho_\theta(N)+ \|  \hat{\thetaL_1} - \hat{\thetaL_2}\|_{L^2}\right)\\
&=: K \left( \rho_M(N)+ \|  \hat{\thetaL_1} - \hat{\thetaL_2}\|_{L^2}\right).\end{align*}
where $K$ is  as defined in Theorem \ref{thm:dist} and  $\rho_W,\rho_\theta$ are as defined in Lemma \ref{lem:norm}.
Hence
$$T_1+T_2\le M_s K \left( 2\rho_M(N) +  \|   \hat\thetaL^{[N]}_\textup{graphon}- \hat\thetaL^{*} \|_{L^2}\right).$$
We finally bound $\|\hat\thetaL^{[N]}_\textup{graphon}-\hat\thetaL^{*}\|_{L^2}$. To this end, let $\hat\thetaL^{[N]*}$ be the piecewise function corresponding to $[\thetaL^*(t^i)]_{i=1}^N$, so that i) by Lemma \ref{lem:norm} with the same probability $1-2\delta_N$, $\|\hat\thetaL^{[N]*}-\hat\thetaL^{*}\|_{L^2}\le \rho^*_\theta(N):=\sqrt{(L^*d_N)^2+4\Omega^*d_N\thetaL^*_{\textup{max}}}$ (where $L^*$ is the Lipschitz constant of $\hat \thetaL^*$, $\Omega^*$ the number of points where $\hat \thetaL^*$ is not Lipschitz continuous and $\hat \thetaL^*(x) \le \thetaL^*_\textup{max}$ for all $x$) and ii) by definition $\hat\thetaL^{[N]}_\textup{graphon}= \hat\thetaL^{[N]*} \frac{\sqrt{C}}{\|\hat\thetaL^{[N]*} \|_{L^2}} =  \hat\thetaL^{[N]*} \frac{\|\hat\thetaL^{*} \|_{L^2}}{\|\hat\thetaL^{[N]*} \|_{L^2}} $. 
Overall
\begin{align*}\|\hat\thetaL^{[N]}_\textup{graphon}-\hat\thetaL^{*}\|_{L^2} &\le \|\hat\thetaL^{[N]}_\textup{graphon}-\hat\thetaL^{[N]*} \|_{L^2} + \|\hat\thetaL^{[N]*}-\hat\thetaL^{*}\|_{L^2}
\\&\le \left|\frac{\|\hat\thetaL^{*} \|_{L^2}}{\|\hat\thetaL^{[N]*} \|_{L^2}}-1\right|\|\hat\thetaL^{[N]*} \|_{L^2} + \rho_\theta^*(N)\\
&= |\|\hat\thetaL^{*} \|_{L^2}- {\|\hat\thetaL^{[N]*} \|_{L^2}}| + \rho_\theta^*(N)
\\& \le \|\hat\thetaL^{*} - {\hat\thetaL^{[N]*} \|_{L^2}} + \rho_\theta^*(N)=2\rho_\theta^*(N).
\end{align*}
Hence $T_1+T_2\le 2M_sK(\rho_M(N)+\rho_\theta^*(N))=:\rho_T(N)$.
}

\vspace{0.2cm}

{\blue \textbf{Proof of Lemma \ref{solvability}} \\
Let $\mathcal{K}$ be the kernel of $\mathbb{W}$, so that $\mathbb{W}\psi=0$ for any function $\psi\in \mathcal{K}$ and let $\mathcal{K}^\perp$ be its orthogonal complement. By the spectral theorem it is possible to construct a orthonormal basis for $\mathcal{K}^\perp$ made of eigenfunctions of $\mathbb{W}$. We denote such basis by $\{\psi_r\}_{r=1}^R$. 
Recall that $\thetaL$ is the status-quo standalone marginal return and let $\hat\thetaL$ be a generic function of $\leb$. 
Let $b_0\psi_0, \hat b_0\hat \psi_0$ be the projection of $\thetaL,\hat \thetaL$ in $\mathcal{K}$ (with $\|\psi_0\|_{L^2}=\| \hat \psi_0\|_{L^2}=1$) and let $b_r = \langle \thetaL, \psi_r \rangle , \hat b_r= \langle \hat \thetaL, \psi_r \rangle $. Since $\mathcal{K}$ is a closed linear subspace it holds 
$$ \thetaL = b_0 \psi_0 + \sum_{r=1}^R b_r \psi_r \quad \textup{and} \quad \hat \thetaL = \hat b_0 \hat\psi_0 + \sum_{r=1}^R \hat b_r \psi_r.$$
Using the fact that $\{\psi_r\}_{r=1}^R$ are orthogonal to each other and orthogonal to any function in $\mathcal{K}$ yields
$$\| \hat \thetaL\|^2_{L^2} = \|  \hat b_0 \hat \psi_0 + \sum_{r=1}^R \hat b_r \psi_r \|^2_{L^2} =  \|  \hat b_0 \hat \psi_0 \|^2_{L^2}+\sum_{r=1}^R \hat b_r^2=  \sum_{r=0}^R \hat b_r^2. $$
Since we are considering linear quadratic games, from the discussion in Example \ref{lqgg}, the equilibrium induced by $\thetaL+\hat \thetaL$ can be rewritten as 
\begin{align*}
\bar s_{\hat \thetaL} &= (\mathbb{I}-\alpha \mathbb{W})^{-1}(\thetaL+\hat \thetaL)= \sum_{h=0}^\infty \alpha^h \mathbb{W}^h(\thetaL+\hat \thetaL)\\
 & =\sum_{h=0}^\infty \alpha^h \mathbb{W}^h ( b_0\psi_0+ \hat b_0\hat \psi_0+  \sum_{r=1}^R (b_r+\hat b_r) \psi_r)\\
  & =\sum_{h=0}^\infty \alpha^h \mathbb{W}^h ( b_0\psi_0+ \hat b_0\hat \psi_0) +  \sum_{r=1}^R (b_r+\hat b_r)\sum_{h=0}^\infty \alpha^h \mathbb{W}^h( \psi_r)\\
    & = ( b_0\psi_0+ \hat b_0\hat \psi_0) +  \sum_{r=1}^R (b_r+\hat b_r)\sum_{h=0}^\infty \alpha^h \lambda_r^h \psi_r\\
        & = ( b_0\psi_0+ \hat b_0\hat \psi_0) +  \sum_{r=1}^R \frac{(b_r+\hat b_r)}{1- \alpha \lambda_r} \psi_r.
\end{align*}
Hence 
$$\|\bar s_{\hat \thetaL}\|_{L^2}^2 =\| b_0\psi_0+ \hat b_0\hat \psi_0\|_{L^2}^2 +  \sum_{r=1}^R \frac{(b_r+\hat b_r)^2}{(1- \alpha \lambda_r)^2}. $$
Note that for any fixed value of $b_0,\psi_0,\hat b_0$ the quantity $\|b_0\psi_0+ \hat b_0\hat \psi_0\|_{L^2}^2$ is maximized when $\hat \psi_0= \psi_0$, in which case $\| b_0\psi_0+ \hat b_0\hat \psi_0\|_{L^2}^2=(b_0+\hat b_0)^2$.
Hence Problem \eqref{infinite_optt} can be equivalently reformulated as \eqref{infinite_opttt}.
}

{\blue \subsection{\textbf{Section \ref{sec:case_study}: Omitted proofs }} \label{csA}

\begin{corollary}\label{cor2}
For all $k=1,\ldots K$, $\hat s^{{ARD}}_k \rightarrow \bar s^{com}_k$ almost surely  as $N\rightarrow \infty$.
\end{corollary}
\begin{proof}
{\blue Consider  for simplicity the case with just one community and suppose that the ARD is administered to all agents. In this case, in the graphon game each agent has the same equilibrium strategy, that is, $\bar s(x)=\bar s^{com}$ for all $x\in[0,1]$ hence
\begin{align*}
\| \bar s^{[N]} - \bar s\|_{L^2}^2 & \textstyle = \int_0^1 (\bar s^{[N]}(x) - \bar s(x))^2 dx =  \sum_{i=1}^N \int_{\mathcal{U}^{[N]}_i} (\bar s^{[N]}_i - \bar s^{com})^2 
dx\\
&\textstyle=\frac 1N \sum_{i=1}^N  (\bar s^{[N]}_i - \bar s^{com})^2 = \frac1N \| \bar s^{[N]} - \bar s^{com}\mathbbm{1}_N \|_2^2.
\end{align*}

This yields
\begin{align*}
|\hat s^{{ARD}} - \bar s^{com} | &\textstyle= \left|\left(\frac1N \sum_{i=1}^N \bar s^{[N]}_i \right)- \bar s^{com}\right| = \left|\frac1N \sum_{i=1}^N (\bar s^{[N]}_i - \bar s^{com})\right| \\
&\textstyle\le \frac1N \sum_{i=1}^N |\bar s^{[N]}_i - \bar s^{com}|  = \frac1N  \|\bar s^{[N]} - \bar s^{com}\mathbbm{1}_N\|_1 \\
&\textstyle\le \frac{\sqrt{N}}N  \|\bar s^{[N]} - \bar s^{com}\mathbbm{1}_N\|_2   = \frac{\sqrt{N}}N \sqrt{N} \| \bar s^{[N]} - \bar s\|_{L^2}=\| \bar s^{[N]} - \bar s\|_{L^2}.
\end{align*}
Since by  Theorem \ref{thm:dist} 
$ \| \bar s^{[N]} - \bar s\|_{L^2} \rightarrow 0$ almost surely,\footnote{\blue Theorem \ref{thm:dist}  requires Assumption \ref{ass:constraint}B) which is not met when $\mathcal{S}=\mathbb{R}_{\ge0}$.  Assumption \ref{ass:constraint}B)   is only used within Theorem \ref{thm:dist} to bound $\|\bar s^{[N]}_{s}\|$. We proved in Lemma \ref{lem:LQ_bound} that for linear quadratic games,  $\|\bar s^{[N]}_{s}\|$ can be bounded, with high probability, even without Assumption \ref{ass:constraint}B). Hence the conclusion of Theorem \ref{thm:dist}  holds.} we finally obtain that 
 $|\hat s^{{ARD}} - \bar s^{com} | \rightarrow 0$ almost surely. A similar proof shows that in the case of $K$ communities $|\hat s^{{ARD}}_k - \bar s^{com}_k | \rightarrow 0$ almost surely for all $k=1,\ldots K$.}
\end{proof}

\begin{corollary}
 $\hat \alpha^{ARD}_\kappa \rightarrow \alpha_\kappa$   almost surely  as $N\rightarrow \infty$.
\end{corollary}
\begin{proof}
{\blue Recall from Corollary \ref{cor2} that $\hat s^{{ARD}} \rightarrow \bar s^{com}$ and from point 2 in Section \ref{sec:ARD} that $\hat E_\kappa^{{ARD}} \rightarrow E_\kappa$.  Hence $\hat X\rightarrow \bar X:=  E_\kappa  \bar s^{com}$, $\hat Y \rightarrow \bar Y:=\bar s^{com} - \thetaL^{{com}}$ and
\begin{equation}\label{limit}
\hat \alpha^{ARD}_\kappa \rightarrow (\bar X^\top \bar X)^{-1} \bar{X}^\top \bar Y.  
\end{equation}
Moreover by \eqref{eq:com_eq}, 
$\bar s^{com}= ( I -\alpha E)^{-1} \thetaL^{com}= ( I -\alpha_\kappa E_\kappa)^{-1} \thetaL^{com}$
or equivalently 
$\bar s^{com} - \thetaL^{com} = \alpha_\kappa E_\kappa \bar s^{com} \textup{ implying } \bar Y= \alpha_\kappa \bar X.$
Substituting in \eqref{limit} yields
$\hat \alpha^{ARD}_\kappa \rightarrow (\bar{X}^\top \bar X)^{-1} \bar{X}^\top \bar Y = \alpha_\kappa(\bar{X}^\top \bar X)^{-1} \bar{X}^\top \bar{X}=\alpha_\kappa,$
as desired.}
\end{proof}}

\begin{center}
\textbf{References}
\end{center}

\bibliographystyle{apalike}
\bibliography{mit.bib}

\newpage

\pagenumbering{arabic}
\renewcommand{\thepage}{OA-\arabic{page}}

\begin{center}
{ONLINE APPENDIX}\\
for\\[0.3cm]
\textbf{``Graphon games: A statistical framework for network games and intervention''}\\[0.6cm]

{Francesca Parise, Asuman Ozdaglar}\\
\textit{Laboratory for Information and Decision Systems,\\ Massachusetts Institute of Technology, Cambridge, MA, USA.}
\end{center}

\ssection{\textbf{Incomplete information in sampled network games}}\label{bayes}

In the main text we   assumed that agents have perfect information about the network $P^{[N]}_{s}$.\footnote{For simplicity we here focus on $0$-$1$ adjacency matrices,  similar results hold for the weighted case.}  In this appendix we   generalize our analysis to sampled network games with \textit{incomplete information}. As in the main text, we  consider sampled network games with $N$ agents whose types $\{t^i\}_{i=1}^N$ are drawn independently and uniformly at random from $[0,1]$ (recall that, e.g., in the community structure model of Example 3 an agent's type represents his community, while in the location model of Example~4 an agent's type is his location in the line segment $[0,1]$). Different from the main text, we here assume that agents do not have access to the exact structure of the sampled network $P^{[N]}_{s}$, but instead each agent $i$ knows the stochastic network formation process (i.e., the graphon $W$ in our framework) and his own type $t^i\in[0,1]$, which determines  the probability $W(t^i,t^j)$ that he will connect to agent $j$ of (random) type $t^j$. We next define a symmetric Bayesian Nash equilibrium for this incomplete information game and show that it is well approximated by the equilibrium of a graphon game with graphon $W$.
 
 Note that the strategy $b(x)$ of each agent in an incomplete information sampled network game specifies the action that the agent  will take as a function of his  type $x$. Assuming that all other agents use the  strategy $b$, the expected payoff of an agent $i$ of type $t^i=x$ playing strategy $s(x)\in\mathtt{S}(x)$ is given by\footnote{ In this section we define the local aggregate by dividing by $N-1$ instead of $N$ to account for the fact that agent $i$ does not consider itself in the local aggregate. This allows us to obtain exact equivalence of the graphon game equilibrium and symmetric Bayesian Nash equilibrium in incomplete information sampled network games with linear quadratic payoffs. With the normalization $\frac1N$ instead of $\frac{1}{N-1}$ the equivalence would hold asymptotically in $N$. }

\begin{align}
\J_\textup{exp}(s(x)\mid b)&=\mathbb{E}_{N, t^{-i}, \textup{links}}\left[ \J\left(s(x), \frac{1}{N-1}\sum_{j\neq i} [P^{[N]}_s]_{ij} b(\type^j), {\blue\theta(x)} \right)\right]
\end{align}
where $\J$ is as in \eqref{gamey} and $\mathbb{E}_{N, t^{-i}, \textup{links}}$ denotes the expectation with respect to the number of agents, their types (each agent  knows its type $\type^i=x$ but has no information about  the other agents types  $t^{-i}:=\{t^j\}_{j\neq i}$, which are independent from $\type^i$) and the link realizations (which are generated according to Bernoulli random variables with probability $\{W(t^i,\type^j)\}_{j\neq i}$). We define a symmetric Bayesian Nash equilibrium as follows.

\begin{definition}[Incomplete information  sampled network game]
  An incomplete information  sampled network game $\mathcal{G}^{in}(\mathtt{S},U,{\blue \theta}, W)$, is a  network game  with a random number $N$ of agents, whose  types $\{t^i\}_{i=1}^N$ are sampled independently and uniformly at random from $[0,1]$, that  interact according to a network $P^{[N]}_s$ sampled from the graphon $W$ according to Definition \ref{sample}.\footnote{The distribution of  $N$ does not matter for our results, with the exception of Theorem \ref{thm:bayes} where we assume that the support of such distribution is bounded from below by $N_\textup{min}$.} Each agent $i$ has information about the graphon $W$, his own type $t^i$,  the strategy sets $\mathtt{S}$, {\blue the function $\theta$} and the payoff function $U$, while is uninformed about $P^{[N]}_s$ and the other agents types $t^{-i}$.
\end{definition}

\begin{definition}[Symmetric Bayesian Nash equilibrium]
 Consider a incomplete information  sampled network game $\mathcal{G}^{in}(\mathtt{S},U,{\blue\theta}, W)$.
  A function $b$ such that $b(x) \in \mathtt{S}(x)$ for all $x\in[0,1]$ is a symmetric $\varepsilon$-Bayesian Nash equilibrium if for all $x\in[0,1]$
\[\J_\textup{exp}(b(x)\mid b) \ge \J_\textup{exp}(\tilde s\mid b) - \varepsilon \textup{ for all } \tilde s \in \mathtt{S}(x).\]
The function $b$  is an exact symmetric Bayesian Nash equilibrium if the previous inequality holds for $\varepsilon=0$.
\end{definition}

\begin{remark}
Note that a strategy profile in both the graphon game and the incomplete information sampled network game is a function that maps  $x\in[0,1]$ into a strategy $s(x)\in\mathtt{S}(x)$. In the graphon game this function specifies the action of a continuum of agents $x\in[0,1]$ interacting according to the graphon $W$, in the incomplete information sampled network game it specifies the action an agent with type $x$ takes if he doesn't know the type of the other sampled agents and the realized links.
\end{remark}

We start by focusing on  linear quadratic games with payoff function as in \eqref{eq:cost_quadratic}.

\begin{theorem} \label{lq_in}Consider a linear quadratic game with payoff as in \eqref{eq:cost_quadratic} and {\blue assume that the peer effect parameter $\alpha$ is the same for each agent while $\thetaL(x)$ is agent specific}. A function $\bar s$ is a Nash equilibrium of the graphon game $\mathcal{G}(\mathtt{S},U,{\blue\thetaL},W)$ if and only if it is a symmetric Bayesian Nash equilibrium for the  incomplete information sampled network game $\mathcal{G}^{in}(\mathtt{S},U,{\blue\thetaL},W)$.
\end{theorem}

\begin{proof}
Let $\bar z(x)=\int_0^1 W(x,y) \bar s(y) dy$. By definition $\bar s$ is a graphon equilibrium if and only if for all $x\in[0,1]$, $\bar s(x)\in\mathtt{S}(x)$ and
\begin{equation}\J( \bar s(x),\bar z(x),{\blue \thetaL(x)})\ge \J(s(x), \bar z(x) ,{\blue \thetaL(x)}) \mbox{ for all } s(x) \in \mathtt{S}(x).\label{c1}\end{equation}

Note that for linear quadratic sampled network games with partial information, the expected payoff of an agent with type $\type^i=x$ is

\begin{align*}
\J_\textup{exp}(s(x)\mid \bar s)&=\mathbb{E}_{N, t^{-i}, \textup{links}}\left[ -\frac{1}{2}(s(x))^2+\left(\alpha \frac{1}{N-1}\sum_{j} [P^{[N]}_s]_{ij} \bar s(\type^j) +{\blue \thetaL(x)} \right)s(x)\right]\\
&= -\frac{1}{2}(s(x))^2+\left(\alpha \mathbb{E}_{N, t^{-i}, \textup{links}}\left[ \frac{1}{N-1}\sum_{j} [P^{[N]}_s]_{ij} \bar s(\type^j)\right] +{\blue \thetaL(x)}\right)s(x),\\
&= \J(s(x), z_{\textup{exp}}(x),{\blue \thetaL(x)})
\end{align*}
where we defined 
$
z_{\textup{exp}}(x):= \mathbb{E}_{N, t^{-i}, \textup{links}}\left[ \frac{1}{N-1}\sum_{j} [P^{[N]}_s]_{ij} \bar s(\type^j)\right].
$\footnote{Note that $\mathbb{E}_{N, t^{-i}, \textup{links}}\left[ \frac{1}{N-1}\sum_{j} [P^{[N]}_s]_{ij} \bar s(\type^j)\right]$ is a function of the type $t^i=x$ of agent $i$ since a link between agent $i$ and $j$ forms (i.e., $[P^{[N]}_s]_{ij}=1$) with Bernoulli  probability $W(t^i,t^j)=W(x,t^j)$.  }
Hence $\bar s$ is a symmetric Bayesian Nash equilibrium if and only if for all $x\in[0,1]$, $\bar s(x)\in\mathtt{S}(x)$ and
\begin{equation}\J( \bar s(x),z_{\textup{exp}}(x),{\blue \thetaL(x)})\ge \J(s(x), z_{\textup{exp}}(x),{\blue \thetaL(x)}) \mbox{ for all } s(x) \in \mathtt{S}(x).\label{c22}\end{equation}

We conclude the proof by showing that $z_{\textup{exp}}(x)=\bar z(x)$ for all $x\in[0,1]$, proving that conditions \eqref{c1} and \eqref{c22} are equivalent. To this end, note that 
\begin{equation}\label{eq:z_exp}
\begin{aligned}
z_{\textup{exp}}(x)&= \mathbb{E}_{N, t^{-i}, \textup{links}}\left[ \frac{1}{N-1}\sum_{j\neq i} [P^{[N]}_s]_{ij} \bar s(\type^j)\right]=\mathbb{E}_\textup{N}\mathbb{E}_{{t^{-i}} \mid \textup{N}}\mathbb{E}_{\textup{ links}\mid {t^{-i}, N}}\left[ \frac{1}{N-1}\sum_{j\neq i} [P^{[N]}_s]_{ij} \bar s(\type^j)\right]  \\
&=\mathbb{E}_\textup{N}\mathbb{E}_{{t^{-i}} \mid \textup{N}}\left[ \frac{1}{N-1}\sum_{j\neq i} W(x,\type^j) \bar s(\type^j)\right]
\end{aligned}
\end{equation}
and for any fixed $N$ 

\begin{align*}
& \mathbb{E}_{{t^{-i}} \mid \textup{N}}\left[  \frac{1}{N-1}\sum_{j\neq i} W(x,\type^j) \bar s(\type^j)\right]=  \frac{1}{N-1}\sum_{j\neq i} \mathbb{E}_{{t^{-i}} \mid \textup{N}}\left[  W(x,\type^j) \bar s(\type^j)\right]\\
&=  \frac{1}{N-1}\sum_{j\neq i} \mathbb{E}_{{\type^j}}\left[  W(x,\type^j) \bar s(\type^j)\right]=\frac{1}{N-1}\sum_{j\neq i}\int_0^1  W(x,y) \bar s(y)dy=\frac{1}{N-1}\sum_{j\neq i}\bar z(x) =\bar z(x)
\end{align*}
where we used the fact that the $\{\type^j\}_{j=1}^N$  are independent and uniformly distributed in $[0,1]$. Hence
\begin{equation}\label{eq:z_exp2}
z_{\textup{exp}}(x)=\mathbb{E}_\textup{N}\mathbb{E}_{{t^{-i}} \mid \textup{N}}\left[ \frac1{N-1}\sum_{j\neq i} W(x,\type^j) \bar s(\type^j)\right]=\mathbb{E}_\textup{N}\bar z(x)=\bar z(x).
\end{equation}
Note that $z_{\textup{exp}}(x)$ does not depend on the distribution of $N$.
\end{proof}

\subsection{\textbf{Generalization to Lipschitz payoff functions}}

In the previous subsection we focused on games with linear quadratic payoff functions and we showed that $\bar s$ is a graphon equilibrium if and only if it is a symmetric Bayesian Nash equilibrium  for an incomplete information sampled network game with any number of agents. We next consider a more general class of payoff functions, satisfying the following assumption.
\begin{assumption} \label{cost2} The payoff function $\J(s,z,{\blue \theta})$ is Lipschitz continuous in $z$ uniformly over $s$ and ${\blue \theta}$, with constant $L_\J$.
\end{assumption}
 The expected payoff for an agent of type $\type^i=x$ in this case is 
 
\begin{align*}
\J_\textup{exp}(s(x)\mid b)&=\mathbb{E}_{N, t^{-i}, \textup{links}}\left[ \J\left(s(x), \frac{1}{N-1}\sum_{j} [P^{[N]}_s]_{ij} b(\type^j),{\blue \theta(x)}\right)\right]\\&=\mathbb{E}_{\zeta_b(x)}\left[ \J\left(s(x), \zeta_b(x),{\blue \theta(x)}\right)\right],
\end{align*}
where $\zeta_b(x)$ is a random variable that describes the possible realizations of local aggregate perceived by an agent  of type $x$ over different network realizations when all agents play according to $b$. 
Note that, when the strategy $b$ equals a graphon game equilibrium $\bar s$, $\mathbb{E}_{\zeta_{\bar s}(x)}[\zeta_{\bar s}(x)]= \bar z(x)=\int_0^1 W(x,y)\bar s(y)dy$ as shown in \eqref{eq:z_exp} and \eqref{eq:z_exp2}. For the payoff functions considered here however
\[\J_\textup{exp}(s(x)\mid \bar s)\!=\! \mathbb{E}_{\zeta_{\bar s}(x)}\left[ \J\left(s(x), \zeta_{\bar s}(x),{\blue \theta(x)}\right)\right]\! \neq\!  \J\left(s(x), \mathbb{E}_{\zeta_{\bar s}(x)}\left[\zeta_{\bar s}(x)\right],{\blue \theta(x)}\right) \! =\!\J(s(x), \bar z(x),{\blue \theta(x)})\]
 since the aggregate enters nonlinearly in the payoff function. Therefore it is not possible to use the argument of Theorem \ref{lq_in} to conclude that $\bar s$ is a symmetric Bayesian Nash equilibrium. 
 Nonetheless, we show in Lemma \ref{concentration} (in Appendix \ref{aux}) that $\zeta_{\bar s}(x)$ concentrates around $ \bar z(x)$  for large population sizes.  Hence, for large populations, $\J(s(x), \bar z(x),{\blue \theta(x)})$  is indeed  a good approximation of $\J_\textup{exp}(s(x)\mid \bar s)$. By exploiting this observation we show in the next theorem that, under the additional assumption that each agent has access to a lower bound ($N_\textup{min}$) on  the population size in any realized sampled network game, 
 the graphon equilibrium $\bar s$ is a symmetric $\varepsilon$-Bayesian Nash equilibrium with  $\varepsilon\rightarrow 0$ as the lower bound on the population size $N_\textup{min}\rightarrow \infty$.

\begin{theorem}\label{thm:bayes}
Consider an incomplete information sampled network game $\mathcal{G}^{in}(\mathtt{S},U,{\blue \theta},W)$ where $\mathtt{S}(x)=\mathcal{S}$ for all $x\in[0,1]$.
Suppose that all the agents know that the population size $N$ is sampled from a distribution whose support is strictly lower bounded by $N_\textup{min}$ and
suppose that Assumptions \ref{ass:cost}, ~\ref{ass:constraint}B),~\ref{cond},  \ref{lipschitz}  (with $\Omega=0$) and  \ref{cost2} hold. Let $\bar s$ be the unique equilibrium of the corresponding graphon game $\mathcal{G}(\mathtt{S},U,{\blue \theta},W)$. Then  $\bar s$ is a symmetric $\varepsilon$-Bayesian Nash equilibrium with 
\[\varepsilon =\mathcal{O}\left( \sqrt{\frac{\log(N_\textup{min})}{N_\textup{min}}}\right).\]
\end{theorem}

\begin{proof}
It follows from the  definition of graphon equilibrium that for all $x\in[0,1]$
\[\J(\bar s(x),\bar z(x),{\blue \theta(x)})\ge \J(s(x),\bar z(x),{\blue \theta(x)}) \quad \forall s(x)\in\mathcal{S},\]
where $\bar z(x)=\int_0^1 W(x,y) \bar s(y)dy$.
Consider an agent of type $t^i$. By the previous inequality specialized for $x=t^i$, it follows that for all $s(\type^i)\in\mathcal{S}$ 
\begin{align*}
\J_\textup{exp}(\bar s(\type^i)\mid \bar s) &=  \mathbb{E}_{\zeta_{\bar s}(\type^i)}[\J(\bar s(\type^i),\zeta_{\bar s}(\type^i),{\blue \theta(\type^i)})]\\&
 = \mathbb{E}_{\zeta_{\bar s}(\type^i)}[\J(\bar s(\type^i),\zeta_{\bar s}(\type^i),{\blue \theta(\type^i)})-\J(\bar s(\type^i),\bar z(\type^i),{\blue \theta(\type^i)})]+ \J(\bar s(\type^i),\bar z(\type^i),{\blue \theta(\type^i)})\\
&\ge -L_\J \mathbb{E}_{\zeta_{\bar s}(\type^i)}[\|\zeta_{\bar s}(\type^i)-\bar z(\type^i)\|]+ \J(s(\type^i),\bar z(\type^i),{\blue \theta(\type^i)})\\
&= -L_\J \mathbb{E}_{\zeta_{\bar s}(\type^i)}[\|\zeta_{\bar s}(\type^i)-\bar z(\type^i)\|]+  \mathbb{E}_{\zeta_{\bar s}(\type^i)}[\J(s(\type^i),\bar z(\type^i),{\blue \theta(\type^i)})- \J(s(\type^i), \zeta_{\bar s}(\type^i),{\blue \theta(\type^i)})] \\&\qquad \hfill{+ \mathbb{E}_{\zeta_{\bar s}(\type^i)}[\J(s(\type^i), \zeta_{\bar s}(\type^i),{\blue \theta(\type^i)})] }\\
&\ge -2L_\J \mathbb{E}_{\zeta_{\bar s}(\type^i)}[\|\zeta_{\bar s}(\type^i)-\bar z(\type^i)\|]+   \mathbb{E}_{\zeta_{\bar s}(\type^i)}[\J(s(\type^i), \zeta_{\bar s}(\type^i),{\blue \theta(\type^i)})]\\&=:-\varepsilon +\J_\textup{exp}( s(\type^i)\mid \bar s).
\end{align*}
The proof is concluded upon showing that $\varepsilon:=2L_\J \mathbb{E}_{\zeta_{\bar s}(\type^i)}[\|\zeta_{\bar s}(\type^i)-\bar z(\type^i)\|]=\mathcal{O}\left( \sqrt{\frac{\log(N_\textup{min})}{N_\textup{min}}}\right)$. Since $\bar z(\type^i)=\mathbb{E}_{\zeta_{\bar s}(\type^i)}[\zeta_{\bar s}(\type^i)]$, we need to show that $\zeta_{\bar s}(\type^i)$ concentrates around its mean when $N_\textup{min}\rightarrow \infty$.
{\blue We show in Lemma \ref{concentration} (given in Appendix \ref{aux}) that for any fixed population of size $N$ and any fixed $\type^i$ with probability at least $1-\frac{2n+1}{({N-1})^2}$  it holds $\|\zeta_{\bar s}(\type^i)-\bar z(\type^i)\|\le \varepsilon'$, with $\varepsilon':=\mathcal{O}\left( \sqrt{\frac{\log(N-1)}{N-1}}\right)$. It follows that
\begin{align*}
\mathbb{E}_{\zeta_{\bar s}(\type^i)\mid N}[\|\zeta_{\bar s}(\type^i)-\bar z(\type^i)\|] &\le \left(1-\frac{2n+1}{({N-1})^2}\right) \varepsilon'+\frac{2n+1}{({N-1})^2} 2s_\textup{max} \\&\le \varepsilon'+\frac{2(2n+1)s_\textup{max}}{({N-1})^2} =\mathcal{O}\left( \sqrt{\frac{\log(N-1)}{N-1}}\right) ,
\end{align*}}
where we used that  $\|\zeta_{\bar s}(\type^i)-\bar z(\type^i)\|\le 2 s_\textup{max}$ for all realizations by Assumption~\ref{ass:constraint}B).
Consequently, if $N>  N_\textup{min}$
\[\mathbb{E}_{\zeta_{\bar s}(\type^i)}[\|\zeta_{\bar s}(\type^i)-\bar z(\type^i)\|] =\mathbb{E}_{N}\mathbb{E}_{\zeta_{\bar s}(\type^i)\mid N} [\|\zeta_{\bar s}(\type^i)-\bar z(\type^i)\|] =\mathcal{O}\left( \sqrt{\frac{\log(N_\textup{min})}{N_\textup{min}}}\right).\]

\end{proof}

{\blue \ssection{\textbf{Identification of unknown parameters}}\label{sec:id}

Consider a setting where agents have payoffs as given in \eqref{eq:br_naJ} which additionally depend on a common parameter $\bar \eta$ (for simplicity assume $n=1$ so that agents strategies are scalars). This could be the case for example in a linear quadratic game with payoff
\begin{equation}\label{eq:Lqeta}
\J_{\bar \eta}(s^i,z^i(s),\thetaL^i)=-\frac12(s^i)^2+( \thetaL^i+\bar \eta z^i(s))s^i,
\end{equation}
where the common parameter $\bar \eta>0$ represents the strength of peer effects.
In this section we consider the problem of identifying the parameter $\bar\eta$ from the observation of a sampled equilibrium $\bar s^{[N]}\in\mathbb{R}^N$. We here assume that the realized network $P^{[N]}$ is unknown   (so that results such as \cite{bramoulle2009identification} cannot be applied). Instead we assume that  the network $P^{[N]}$ is a realization from an underlying known graphon $W$ (e.g., $P^{[N]}$ may be a realization from a stochastic block model).

Let us denote by   $\Xi$ the set of parameters $\eta$ for which the corresponding graphon game $\mathcal{G}(\mathtt{S}, \J_\eta,  \theta, W)$ satisfies Assumptions \ref{ass:cost}, \ref{ass:constraint}A) and \ref{cond}. Moreover 
denote by $\bar s_\eta\in\leb$ the unique equilibrium of $\mathcal{G}(\mathtt{S}, \J_\eta,  \theta, W)$. To identify the parameter $\bar \eta$ from an observation of  $\bar s^{[N]}$, one could solve the following optimization problem
\begin{equation}\label{estimate}
\hat \eta:=\arg\min_{\eta\in\Xi} \| \bar s^{[N]}- \bar s_\eta\|_{L^2}
\end{equation}
where $\bar s^{[N]}$ denotes the step function equilibrium.
Intuitively, one can select as  estimate the parameter $\eta$ that minimizes the distance between the observed sampled equilibrium ($\bar s^{[N]}$) and the equilibrium ($\bar s_\eta$) of a graphon game with parameter $\eta$. We next show that if the parameter $\bar \eta$ is identifiable, as defined next, then $\|\hat \eta-\bar \eta\| \rightarrow 0$ as $N\rightarrow \infty$.

\begin{definition}[Identifiability]\label{def:id}
A parameter $\bar \eta \in \Xi$ is identifiable if there exists $L_{\bar \eta}>0$ such that for any $\eta\in\Xi$ it holds
$$ \|\bar \eta- \eta \| \le L_{\bar \eta}  \| \bar s_{\bar \eta}- \bar s_\eta\|_{L^2}.$$
\end{definition}
Intuitively, a parameter $\bar \eta$ is identifiable if  equilibria that are close to $\bar s_{\bar \eta}$ are generated by parameters that are close to $\bar \eta$.
Under this condition we can prove the following corollary of our main convergence theorem.
\begin{corollary}
Suppose that $\mathcal{G}(\mathtt{S}, \J_{\bar \eta},  \theta, W)$  satisfies  Assumptions \ref{ass:cost}, \ref{ass:constraint}B), \ref{cond} and \ref{lipschitz} and that the parameter $\bar \eta$ is identifiable.  Fix  any admissible confidence sequence $\{\delta_N\}_{N=1}^\infty$.  
Then with probability at least $1-2\delta_N$, for $N$ large enough,  it holds
$$ \|\bar \eta- \hat \eta \| \le 2L_{\bar \eta}\bar K\rho(N), $$
where $\rho(N)\rightarrow 0$,  as by  Theorem \ref{thm:dist}, and  $ \bar K:=\frac{\max\{\ell_\J(\bar \eta),\ell_\theta(\bar \eta)\}/\alpha_\J(\bar \eta)}{1-\ell_\J(\bar \eta)/\alpha_\J(\bar \eta) \lambda_{\textup{max}}(\mathbb{W})} $.
\end{corollary}
\begin{proof}
By Theorem \ref{thm:dist} with probability at least $1-2\delta_N$, for $N$ large enough,  it holds
\begin{align*}
\|\bar s^{[N]}-\bar s_{\bar \eta}\|_{L^2}  \le \bar K\rho(N).
\end{align*}
Since $\bar \eta$ is a feasible point of the optimization problem in \eqref{estimate} and $\hat \eta$ is the optimizer it must be
\begin{align*}
\|\bar s^{[N]}-\bar s_{\hat \eta}\|_{L^2}  \le \bar K\rho(N).
\end{align*}
Combining these two inequalities yields
\begin{align*}
\|\bar s_{\hat \eta}-\bar s_{\bar \eta}\|_{L^2}\le \|\bar s^{[N]}-\bar s_{\hat \eta}\|_{L^2} + \|\bar s^{[N]}-\bar s_{\bar \eta}\|_{L^2}  \le 2\bar K\rho(N).
\end{align*}
The identifiability condition yields
$$ \|\bar \eta- \hat \eta \| \le L_{\bar \eta}  \| \bar s_{\bar \eta}- \bar s_{\hat \eta}\|_{L^2}\le 2L_{\bar \eta}\bar K\rho(N). $$

\end{proof}

Assessing for which parameters and games the identifiability condition in Definition \ref{def:id} holds is an interesting open problem. We here briefly comment on linear quadratic games with payoff as in \eqref{eq:Lqeta}. In this case, we recall from Example \ref{lqgg}, that for any $\eta\in\Xi$, $\eta>0$
\begin{align}
\bar s_\eta     &= (\mathbb{I}-\eta\mathbb{W})^{-1}\thetaL \quad \Leftrightarrow \quad \bar s_\eta - \thetaL  = \eta \bar z_\eta,
\end{align}
where $\bar z_\eta:=\mathbb{W}\bar s_\eta$. 
It follows from $(\bar \eta-\eta) \bar z_{\bar \eta}= (\bar \eta \bar z_{\bar\eta}-\eta \bar z_{\eta}) - \eta(\bar z_{\bar\eta}-\bar z_{\eta})$ that
\begin{align*}
|\bar\eta-\eta|\| \bar z_{\bar\eta}\|_{L^2} &\le  \|\bar\eta \bar z_{\bar\eta}-\eta \bar z_{\eta}\|_{L^2} + \eta\| \bar z_{\bar\eta}-\bar z_{\eta}\|_{L^2}\\
&= \|\bar s_{\bar \eta}-\bar s_\eta\|_{L^2} + \eta\|\mathbb{W}( \bar s_{\bar\eta}-\bar s_{\eta})\|_{L^2}\\
&\le \|\bar s_{\bar \eta}-\bar s_\eta\|_{L^2} + \eta \lambda_{\textup{max}}(\mathbb{W})\| \bar s_{\bar\eta}-\bar s_{\eta}\|_{L^2} \le 2\|\bar s_{\bar \eta}-\bar s_\eta\|_{L^2}, 
\end{align*}
where we used that $\eta\in\Xi$ implies $ \eta <\frac{1}{\lambda_{\textup{max}}(\mathbb{W}) }$.
Hence
$$|\bar\eta-\eta|\le \frac{2}{\| \bar z_{\bar\eta}\|_{L^2} }\| \bar s_{\bar\eta}-\bar s_{\eta}\|_{L^2}=: L_{\bar \eta}\| \bar s_{\bar\eta}-\bar s_{\eta}\|_{L^2}, $$
proving that any $\bar \eta\in\Xi$ is identifiable in linear quadratic network games.

}

 \ssection{\textbf{Auxiliary results}}\label{aux_all}

{\color{blue}
\vspace{0.2cm}
\subsection{\textbf{Statements in support of Section \ref{avg}: Average instead of aggregate}}\label{avgA}

 \begin{lemma}
If $\int_0^1 W(x,y) dy \ge d_\textup{min}>0$ a.e. then $\mathbb{W}_d$ is a linear Hilbert-Schmidt integral operator and $\vertiii{\mathbb{W}_d}\le \frac{ \lambda_{\textup{max}}(\mathbb{W})}{ d_\textup{min}}$.
\end{lemma}
 \begin{proof}
Note that $\mathbb{W}_d$ is a linear integral operator with kernel $W_d(x,y):=\frac{W(x,y)}{\int_0^1 W(x,y)dy}$. Since
\begin{align*}
\int_0^1 \int_0^1 \left(\frac{W(x,y)}{\int_0^1 W(x,y)dy}\right)^2 dy dx &= \int_0^1 \frac{\int_0^1 W(x,y)^2 dy}{(\int_0^1 W(x,y)dy)^2} dx  \\
&\le  \int_0^1 \frac{\int_0^1 1 dy}{(d_\textup{min})^2} dx = \left(\frac{1}{d_\textup{min}}\right)^2 <\infty,
\end{align*}
$\mathbb{W}_d$ is a Hilbert-Schmidt integral operator. Moreover, by definition
\begin{align*}
\vertiii{\mathbb{W}_d}^2&=\sup_{f\in\leb \mid \|f\|_{L^2}\le 1} \|\mathbb{W}_d f\|_{L^2}^2 =\sup_{f\in\leb \mid \|f\|_{L^2}\le 1}\int_0^1 (\mathbb{W}_d f)^2(x) dx \\& =\sup_{f\in\leb \mid \|f\|_{L^2}\le 1}\int_0^1 \left( \frac{\int_0^1 W(x,y) f(y) dy}{\int_0^1 W(x,y)dy}\right)^2 dx \\
&\le \sup_{f\in\leb \mid \|f\|_{L^2}\le 1}\int_0^1 \left( \frac{\int_0^1 W(x,y) f(y) dy}{d_\textup{min}}\right)^2 dx \\& = \frac{1}{(d_\textup{min})^2} \sup_{f\in\leb \mid \|f\|_{L^2}\le 1} \|\mathbb{W} f\|_{L^2}^2 = \left(\frac{\vertiii{\mathbb{W}}}{d_\textup{min}}\right)^2 = \left(\frac{\lambda_\textup{max}(\mathbb{W})}{d_\textup{min}}\right)^2.
\end{align*}
\end{proof}

\begin{lemma}\label{lem:normD}

Consider a graphon  $W$ satisfying Assumption \ref{lipschitz} with $\Omega=0$ and suppose that $\int_0^1 W(x,y) dy \ge d_\textup{min}>0$ a.e. Let $W^{[N]}_{w/s}$ be the step function graphons corresponding to the matrices  $P^{[N]}_{w/s}$, as defined in  Section \ref{step1} and $\theta^{[N]}$ be the step-function corresponding to $[\theta(t^i)]_{i=1}^N$. Let $\mathbb{W}^{[N]}_{sd/nd}$ be the normalized graphon operators corresponding to $W^{[N]}_{w/s}$.  Fix  any sequence $\{\delta_N\}_{N=1}^\infty$ such that $\delta_N \le e^{-1}$ and $\frac{\log( N/\delta_N)}{N} \rightarrow 0$.  Then, for  $N$ large enough,
\begin{enumerate} 
\item with probability at least  $1-\delta_N$,  \eqref{distU} holds, $\|\theta^{[N]}-\theta\|_{L^2;\mathbb{R}^m}\le \rho_\theta(N)$ and  
\[ \vertiii{\mathbb{W}^{[N]}_{wd}-\mathbb{W}_d} =\mathcal{O}\left(d_N\right);\]
\item  with probability at least  $1-4\delta_N$,  \eqref{distU} holds, $\|\theta^{[N]}-\theta\|_{L^2;\mathbb{R}^m}\le \rho_\theta(N)$ and
\begin{equation}\vertiii{\mathbb{W}^{[N]}_{sd}-\mathbb{W}_d} =\mathcal{O}\left(\sqrt{\frac{\log(N/\delta_N)}{N}}\right).
\end{equation}
\end{enumerate}
\end{lemma}

\begin{proof}
 Note that for $N$ large enough the condition $\delta_N\in(Ne^{-N/5},e^{-1})$  is satisfied under the assumptions of this lemma. Hence Lemma \ref{lem:dist}  applies. 
\begin{enumerate}
\item 
Define
$$d^{[N]}_{w}(x):=\int_0^1 W^{[N]}_{w}(x,y) dy \quad \textup{and} \quad d(x):=\int_0^1 W(x,y) dy.$$ 
Then with probability $1-\delta_N$  if $x\in \mathcal{U}_i^{[N]}$
\begin{align*}
|d^{[N]}_{w}(x)- d(x)|&\le \int_0^1 |W^{[N]}_{w}(x,y)-W(x,y) | dy \\
&= \sum_j \int_{\mathcal{U}_j^{[N]}} |W(t^i,t^j)-W(x,y) | dy \le   \sum_j \int_{\mathcal{U}_j^{[N]}} 2Ld_N dy=2Ld_N =: \epsilon_N,
\end{align*}
where we used the Lipschitz property and \eqref{distU} from Lemma \ref{lem:dist} in the last inequality. Hence
$$d^{[N]}_{w}(x)\ge  d(x) -\epsilon_N \ge d_\textup{min}-\epsilon_N.$$
Similarly for $x\in \mathcal{U}_i^{[N]}$ and $y\in \mathcal{U}_j^{[N]}$ we obtain
$$|{W^{[N]}_{w}(x,y)}-{W(x,y)}| \le \epsilon_N.$$
Let $D(x,y):= \frac{W^{[N]}_{w}(x,y)}{d^{[N]}_w(x)}-\frac{W(x,y)}{d(x)}$.  Then
\begin{align*}
|D(x,y)|&=\left|\frac{W^{[N]}_{w}(x,y)}{d^{[N]}_w(x)}-\frac{W(x,y)}{d(x)}\right| = \frac{|{W^{[N]}_{w}(x,y)}{d(x)}-{W(x,y)}{d^{[N]}_w(x)}| }{d^{[N]}_w(x)d(x)} \\
&\le \frac{|{W^{[N]}_{w}(x,y)}{d(x)}-{W(x,y)}{d(x)}|+|{W(x,y)}{d(x)}  -{W(x,y)}{d^{[N]}_w(x)}| }{d_\textup{min} (d_\textup{min}-\epsilon_N)} \\
&\le \frac{|{W^{[N]}_{w}(x,y)}-{W(x,y)}|+|{d(x)}  -{d^{[N]}_w(x)}| }{d_\textup{min} (d_\textup{min}-\epsilon_N)}  = \frac{2\epsilon_N}{d_\textup{min} (d_\textup{min}-\epsilon_N)} =: \gamma_N \rightarrow 0.
\end{align*}
Consider  any $f\in\leb$ such that $\|f\|_{L^2}=1$. Using the inequalities above
\begin{align*}
   & \|\mathbb{W}^{[N]}_{wd}f-\mathbb{W}_df\|^2_{L^2}=\int_0^1 (\mathbb{W}^{[N]}_{wd}f-\mathbb{W}_df)(x)^2 \mathrm{d}x
                                   = \int_0^1 \left(\int_0^1 D(x,y) f(y) \mathrm{d}y \right)^2  \mathrm{d}x \nonumber\\
                                   &\le \int_0^1 \left(\int_0^1 D(x,y)^2 \mathrm{d}y\right) \left(\int_0^1 f(y)^2 \mathrm{d}y\right) \mathrm{d}x\label{eq:p1} =
                                  \int_0^1 \int_0^1 D(x,y)^2 \mathrm{d}y \mathrm{d}x\le \gamma_N^2.\                                                                                                                             
\end{align*}

Hence with probability $1-\delta_N$, \eqref{distU} holds and

$$ \vertiii{\mathbb{W}^{[N]}_{wd}-\mathbb{W}_d} =\sup_{f\in\leb s.t. \|f\|_{L_2}=1} \|\mathbb{W}^{[N]}_{wd}f-\mathbb{W}_df\|_{L^2} \le \gamma_N =\mathcal{O}(d_N).$$
 The bound on 
$\|\theta^{[N]}-\theta\|^2_{L^2;\mathbb{R}^m} $ can be proven as in Lemma \ref{lem:norm}.

\item Note that  
\begin{align*}
&\vertiii{\mathbb{W}^{[N]}_{sd}-\mathbb{W}^{[N]}_{wd} } \le \|P_{sd}^{[N]}-P_{wd}^{[N]}\|,
\end{align*}
where $P_{sd/wd}^{[N]}$ are the degree normalized versions of $P_{s/w}^{[N]}$.
Recall that $P_w^{[N]}=\mathbb{E}[P_s^{[N]}]$, hence we can  bound the term on the right hand side  by employing matrix concentration inequalities.  Define $d^i_{s/w}=\sum_{j=1}^N [P_{s/w}^{[N]}]_{ij}$.
\begin{itemize}
\item By definition and by the previous point 
$$d^i_{w}= N d^{[N]}_w(t^i) \ge N (d_{\textup{min}}-\epsilon_N) $$
$$ \|P_{w}^{[N]}\| = \vertiii{\mathbb{W}_w^{[N]}} N \le N $$
\item By Hoeffding inequality for any fixed $i$ and $t>0$
$$\textup{Pr}[|d^i_s-d^i_w|> t] < 2 \textup{exp} \left( -\frac{2t^2}{N}\right).$$
Setting $t=\sqrt{\frac{N}{2} \textup{log}\left(\frac{2N}{\delta_N}\right)}$ yields
$$\textup{Pr}\left[|d^i_s-d^i_w|> \sqrt{\frac{N}{2} \textup{log}\left(\frac{2N}{\delta_N}\right)}\right] < 2 \frac{\delta_N}{2N}=\frac{\delta_N}{N}$$
and by the union bound with probability at least $1-\delta_N$
$$ |d^i_s-d^i_w|\le \sqrt{\frac{N}{2} \textup{log}\left(\frac{2N}{\delta_N}\right)}=t \quad \textup{for all }\quad  i\in\{1,\ldots,N\}.$$
Let $D_{s/w}:=\textup{diag}([{d^i_{s/w}}]_{i=1}^N)$. With the same probability 
$$\|D_s^{-1}\| =\max_i \frac{1}{d^i_s} \le \max_i \frac{1}{d^i_w - t}  \le   \frac{1}{(d_{\textup{min}}-\epsilon_N) N - t}  $$
and
\begin{align*}
\|D_s^{-1}-D_w^{-1}\|& =\max_i \left|\frac{1}{d^i_s} -\frac{1}{d^i_w}\right| =  \max_i \frac{|d^i_s-d^i_w|}{d^i_wd^i_s} \le \frac{t}{ (d_{\textup{min}}-\epsilon_N) N }  \frac{1}{(d_{\textup{min}}-\epsilon_N) N - t}
\end{align*}\item The maximum expected degree $C^d_N:=\max_i( \sum_{j=1}^N [P_w^{[N]}]_{ij}   )$ grows  as order $N$. Hence for  $N$ large enough,  it is greater than $\frac49 \log(\frac{2N}{\delta_N})$ since $\frac{\log( N/\delta_N)}{N} \rightarrow 0$ by assumption.
Consequently, all the conditions of \cite[Theorem 1]{chung2011spectra} 
are met and with probability $1-\delta_N$ 
\begin{align*}
\|P_s^{[N]}-P_w^{[N]}\|\le   \sqrt{4C^d_N \log(2N/\delta_N)} \le   \sqrt{4N \log(2N/\delta_N)} ,
\end{align*}
where we used that  $C^d_N\le N$ since each element in $P_w^{[N]}$ belongs to $[0,1]$.
\item Combining the previous results yields that with probability $1-3\delta_N$
\begin{align*}
\|P_{sd}^{[N]}-P_{wd}^{[N]}\| &= \|D_s^{-1}P_{s}^{[N]}-D_w^{-1}P_{w}^{[N]}\| \\&
\le \|D_s^{-1}P_{s}^{[N]}-D_s^{-1}P_{w}^{[N]} \|+\|D_s^{-1}P_{w}^{[N]}  -D_w^{-1}P_{w}^{[N]}\| \\&
 \le  \|D_s^{-1}\| \|P_{s}^{[N]}-P_{w}^{[N]} \|+\|D_s^{-1} -D_w^{-1}\|\|P_{w}^{[N]}\| \\&
\le  \frac{\sqrt{4N \log(2N/\delta_N)}}{(d_{\textup{min}}-\epsilon_N) N - t}   + \frac{t}{ (d_{\textup{min}}-\epsilon_N) N } \cdot \frac{N}{(d_{\textup{min}}-\epsilon_N) N - t}\\
&=   \frac{\sqrt{8} t/N}{(d_{\textup{min}}-\epsilon_N)  - t/N}   + \frac{t/N}{ (d_{\textup{min}}-\epsilon_N) }  \cdot \frac{1}{(d_{\textup{min}}-\epsilon_N)  - t/N}
\end{align*}
Since $t/N = \sqrt{\frac{ \textup{log}\left({2N}/{\delta_N}\right)}{2N}} \rightarrow 0 $, we obtain
$$ \vertiii{\mathbb{W}^{[N]}_{sd}-\mathbb{W}^{[N]}_{wd} } \le \|P_{sd}^{[N]}-P_{wd}^{[N]}\| = \mathcal{O}\left(t/N\right) = \mathcal{O}\left(\sqrt{\frac{ \textup{log}\left({N}/{\delta_N}\right)}{N}}  \right).  $$
\end{itemize} 

Using the fact that 
$$\vertiii{\mathbb{W}^{[N]}_{s}-\mathbb{W} }\le \vertiii{\mathbb{W}^{[N]}_{s}-\mathbb{W}^{[N]}_{w} }+\vertiii{\mathbb{W}^{[N]}_{w}-\mathbb{W} }$$
and the first statement concludes the proof.
\end{enumerate}

\end{proof}

}

{\color{blue}
\vspace{0.2cm}
\subsection{\textbf{Statements in support of Section \ref{sec:dir}: Directed networks}}\label{sec:dirA}
\vspace{0.2cm}

\begin{lemma}\label{lem:asy}
Consider a matrix $P^{[N]}_w\in [0,1]^{N\times N}$ with $\|P^{[N]}_w\|_\infty$ of order $N$ and a random matrix $P^{[N]}_s\in \{0,1\}^{N\times N}$ such that 
$$[P^{[N]}_s]_{ij}=Ber([P^{[N]}_w]_{ij}).$$
With probability $1-\delta_N$ for $N$ large enough
$$\frac1N\| P^{[N]}_s- P^{[N]}_w\| \le \sqrt{4 \frac{\log(4N/\delta_N)}{N}}.$$
\end{lemma}
\begin{proof}
Construct the symmetric matrix 
$$A^{[2N]}_{s/w} = \left[\begin{array}{cc}0 & P^{[N]}_{s/w} \\  (P^{[N]}_{s/w})^T& 0\end{array}\right]\in\mathbb{R}^{2N\times 2N}$$
and note that 
\begin{enumerate}
\item $\| P^{[N]}_s- P^{[N]}_w\| = \| A^{[2N]}_s- A^{[2N]}_w\|$;
\item $\mathbb{E}[P^{[N]}_s ]= P^{[N]}_w$ implies  $\mathbb{E}[A^{[2N]}_s ]= A^{[2N]}_w$;
\item the maximum degree $\Delta_A$ of $A^{[2N]}_w$ is of order $N$ and is therefore greater than $\frac{4}{9} \log(4N/\delta_N)$ for $N$ large enough.
\end{enumerate}
Then by \cite[Theorem 1]{chung2011spectra} with probability $1-\delta_N$ for $N$ large enough
$$\frac1N\| P^{[N]}_s- P^{[N]}_w\| =\frac1N \| A^{[2N]}_s- A^{[2N]}_w\| \le \frac1N\sqrt{4\Delta_A \log(4N/\delta_N)}\le \sqrt{4 \frac{\log(4N/\delta_N)}{N}}.$$
\end{proof}

}

\vspace{0.2cm}
\subsection{\textbf{Auxiliary results}}\label{aux}
\vspace{0.2cm}

 We report here some auxiliary lemmas.  Specifically, 
\begin{itemize}
\item[-] Lemma \ref{eigSBM},  \ref{lem:dist} and \ref{lem:norm} are immediate extensions of results in \cite{graphons};
\item[-] Lemma \ref{lem:lip} derives sufficient conditions for the equilibrium of a graphon game to be Lipschitz continuous;
\item[-] Lemma \ref{concentration} provides a concentration result for the local aggregate in incomplete information sampled network games;
\item[-]  Lemma \ref{thm:eps} proves that the graphon equilibrium is an $\epsilon$-Nash equilibrium under additional regularity assumptions.
\end{itemize}

\begin{lemma}[\cite{graphons}]\label{eigSBM}  Consider a SBM graphon $\mathbb{W}_{\textup{SBM}}$ which is piecewise constant over the partition $\{\mathcal{C}_k\}_{k=1}^K$.
If $(\lambda,\psi)$ is an eigenpair of $\mathbb{W}_{\textup{SBM}}$, then there exists $v\in\R^K$ such that $(\lambda,v)$ is an eigenpair of the matrix $E\in\mathbb{R}^{K\times K}$ defined in \eqref{eq:E} and
\begin{equation}\label{phi_v}\psi(x)= \gamma v_k, \mbox{ for all } x\in \mathcal{C}_k\end{equation}
where  $\gamma>0$ is a normalization parameter. Conversely, if $(\lambda,v)$ is an eigenpair of the matrix $E\in\mathbb{R}^{K\times K}$ then $(\lambda,\psi)$  is an eigenpair of $\mathbb{W}_{\textup{SBM}}$ with $\psi$ constructed from $v$ as in \eqref{phi_v}.
\end{lemma}

\begin{lemma}[\cite{graphons}]\label{lem:dist}
Let $\{\type^i\}_{i=1}^N$ be the ordered statistics of $N$ random samples from $\mathcal{U}[0,1]$. For any $\delta_N\in (Ne^{-N/5},e^{-1})$ and $N$ large, with probability at least $1-\delta_N$ it holds
\begin{equation}\label{distU}
|\type^i-x|\leq d_N \mbox{ for any } i\in\{1,\ldots,N\} \mbox{ and any }\textstyle x\in\mathcal{U}_i^{[N]}=[\frac{i-1}{N},\frac iN),
\end{equation}
where $d_N:=\frac1N+\sqrt{\frac{8\log(N/\delta_N)}{N}}\rightarrow 0$.
\end{lemma}

\begin{lemma}[\cite{graphons}]\label{lem:norm}

Consider a graphon  $W$ satisfying Assumption \ref{lipschitz}. Let $W^{[N]}_{w/s}$ be the step function graphons corresponding to the matrices  $P^{[N]}_{w}$ and $\frac{P^{[N]}_{s}}{\blue \kappa_N}$, as defined in  Sections \ref{step1} {\blue and \ref{sec:sparse}. Let $\theta^{[N]}$ be the step-function corresponding to $[\theta(t^i)]_{i=1}^N$. Fix  any sequence $\{\delta_N,\kappa_N\}_{N=1}^\infty$ such that $\delta_N \le e^{-1}$ and $\frac{\log( N/\delta_N)}{N\kappa_N} \rightarrow 0$. } Then, for  $N$ large enough,
\begin{enumerate} 
\item with probability at least  $1-\delta_N$,  \eqref{distU} holds, 
{\blue \[\|\theta^{[N]}-\theta\|_{L^2;\mathbb{R}^m}\le \rho_\theta(N):=\sqrt{(Ld_N)^2+8\Omega d_N \theta^2_\textup{max}}\] } and  
\[ {\blue|\lambda_{\textup{max}}(\mathbb{W}^{[N]}_{w})-\lambda_{\textup{max}}(\mathbb{W}) |\le } \vertiii{\mathbb{W}^{[N]}_{w}-\mathbb{W}} \le \tilde \rho(N):= 2 \sqrt{(L^2-\Omega^2)d_N^2+\Omega d_N};\]
\item  with probability at least  $1-2\delta_N$,  \eqref{distU} holds, $\|\theta^{[N]}-\theta\|_{L^2;\mathbb{R}^m}\le \rho_\theta(N)$ and
\begin{equation}
{\blue |\lambda_{\textup{max}}(\mathbb{W}^{[N]}_{s})-\lambda_{\textup{max}}(\mathbb{W}) |\le} \vertiii{\mathbb{W}^{[N]}_{s}-\mathbb{W}}\le \tilde \rho(N)+ \sqrt{\frac{4\log(2N/\delta_N)}{N{\blue \kappa_N}}}=:\rho_W(N).
\end{equation}
\end{enumerate}
\end{lemma}

\begin{proof}
 Note that for $N$ large enough the condition $\delta_N\in(Ne^{-N/5},e^{-1})$  is satisfied under the assumptions of this lemma. In fact, if  $\delta_N \le Ne^{-N/5}$ infinitely often then $\frac{\log(N/\delta_N)}{N \color{blue}\kappa_N}\ge \frac{\log(N/N \cdot e^{N/5})}{N} =\frac{1}{5}$ infinitely often and the assumption $\frac{\log(N/\delta_N)}{N {\blue \kappa_N}}\rightarrow 0$ would be violated. Hence Lemma \ref{lem:dist}  applies and  the result for piecewise Lipschitz graphons follows from \cite[Theorem 1]{graphons}.
 We here report a simplified  proof for  Lipschitz continuous graphons (i.e. for the case $\Omega=0$).
\begin{enumerate}
\item Consider  any $f\in\leb$ such that $\|f\|_{L^2}=1$. Let $D(x,y):= W^{[N]}_{w}(x,y)-W(x,y)$. Then with probability $1-\delta_N$ (independent of $f$)

\begin{align*}
    \|\mathbb{W}^{[N]}_{w}f-\mathbb{W}f\|^2_{L^2}&=\int_0^1 (\mathbb{W}^{[N]}_{w}f-\mathbb{W}f)(x)^2 \mathrm{d}x
                                   = \int_0^1 \left(\int_0^1 D(x,y) f(y) \mathrm{d}y \right)^2  \mathrm{d}x \nonumber\\
                                   &\le \int_0^1 \left(\int_0^1 D(x,y)^2 \mathrm{d}y\right) \left(\int_0^1 f(y)^2 \mathrm{d}y\right) \mathrm{d}x\label{eq:p1} \\&=  \int_0^1 \left(\int_0^1 D(x,y)^2 \mathrm{d}y\right)\|f\|^2_{L^2}  \mathrm{d}x   = \int_0^1 \int_0^1 D(x,y)^2 \mathrm{d}y \mathrm{d}x\\& = \sum_{i}\sum_j \int_{\mathcal{U}_i^{[N]}} \int_{\mathcal{U}_j^{[N]}} (W^{[N]}_{w}(x,y)-W(x,y))^2 \mathrm{d}y \mathrm{d}x\\   
                                & = \sum_{i}\sum_j \int_{\mathcal{U}_i^{[N]}} \int_{\mathcal{U}_j^{[N]}} (W(\type^i,\type^j)-W(x,y))^2 \mathrm{d}y \mathrm{d}x \\&\le L^2 \sum_{i}\sum_j \int_{\mathcal{U}_i^{[N]}} \int_{\mathcal{U}_j^{[N]}} (|\type^i-x|+|\type^j-y|)^2 \mathrm{d}y \mathrm{d}x \\       
                                & \le L^2 \sum_{i}\sum_j \int_{\mathcal{U}_i^{[N]}} \int_{\mathcal{U}_j^{[N]}} (2d_N)^2 \mathrm{d}y \mathrm{d}x = (2Ld_N)^2\                                                                                                                             
\end{align*}

where we used \eqref{distU} from Lemma \ref{lem:dist} in the last inequality. Hence with probability $1-\delta_N$, \eqref{distU} holds and

$$ \vertiii{\mathbb{W}^{[N]}_{w}-\mathbb{W}} =\sup_{f\in\leb s.t. \|f\|_{L_2}=1} \|\mathbb{W}^{[N]}_{w}f-\mathbb{W}f\|_{L^2} \le 2Ld_N.$$
{\blue 
The fact that 
$|\lambda_{\textup{max}}(\mathbb{W}^{[N]}_{w})-\lambda_{\textup{max}}(\mathbb{W}) |\le  \vertiii{\mathbb{W}^{[N]}_{w}-\mathbb{W}} $ can be proven by inverse triangular inequality upon noting that $\lambda_{\textup{max}}(\mathbb{W}^{[N]}_{w})= \vertiii{\mathbb{W}^{[N]}_{w}} $ and $\lambda_{\textup{max}}(\mathbb{W}) = \vertiii{\mathbb{W}} $.

Similarly,  
\begin{align}
\|\theta^{[N]}-\theta\|^2_{L^2;\mathbb{R}^m} &=\int_0^1 \|\theta^{[N]}(x)-\theta(x)\|^2 dx = \sum_i \int_{\mathcal{U}_i^{[N]}} \|\theta(t^i)-\theta(x)\|^2 dx \\
& \le \sum_i \int_{\mathcal{U}_i^{[N]}} L^2 |t^i-x|^2 dx \le \sum_i \int_{\mathcal{U}_i^{[N]}} (L d_N)^2dx = (Ld_N)^2. 
\end{align}
}
\item The operator $\mathbb{W}^{[N]}_{s}-\mathbb{W}^{[N]}_{w}$ can be seen as the graphon operator of an SBM  graphon with matrix $\frac{P_s^{[N]}}{\blue \kappa_N}-P_w^{[N]}$ over the uniform partion $\{\mathcal{U}_i^{[N]}\}_{i=1}^N$. Note that for any graphon operator $\mathbb{A}$  over such partition (i.e. $\mathbb{A}(x,y)=A_{ij}$ for $x\in\mathcal{U}^{[N]}_i,y\in\mathcal{U}^{[N]}_i$)  it holds 
$\vertiii{\mathbb{A}}\le\frac1N\|A\|$.
%
%
Consequently,  
\begin{align*}
&\vertiii{\mathbb{W}^{[N]}_{s}-\mathbb{W}^{[N]}_{w} } \le \frac{1}{N}\|\frac{P_s^{[N]}}{\blue \kappa_N}-P_w^{[N]}\|={\blue \frac{1}{N\kappa_N}\|P_s^{[N]}-\kappa_N P_w^{[N]}\|.}
\end{align*}
Recall that ${\blue \kappa_N}P_w^{[N]}=\mathbb{E}[P_s^{[N]}]$, hence we can  bound the term on the right hand side  by employing matrix concentration inequalities. 

The maximum expected degree $C^d_N:=\max_i( \sum_{j=1}^N {\blue \kappa_N} [P_w^{[N]}]_{ij}   )$ grows  as order ${\blue \kappa_N} N$. Hence for  $N$ large enough,  it is greater than $\frac49  \log(\frac{2N}{\delta_N})$ since $\frac{\log( N/\delta_N)}{N {\blue \kappa_N}} \rightarrow 0$ by assumption.
Consequently, all the conditions of \cite[Theorem 1]{chung2011spectra} 
are met and with probability $1-\delta_N$ 
\begin{align*}
& \frac{1}{N{\blue \kappa_N}}\|P_s^{[N]}-{\blue \kappa_N}P_w^{[N]}\|\le   \frac{1}{N{\blue \kappa_N}}\sqrt{4C^d_N \log(2N/\delta_N)} \le   \sqrt{\frac{4  \log(2N/\delta_N)}{N {\blue \kappa_N}}},
\end{align*}
where we used that  $C^d_N\le {\blue \kappa_N} N$ since each element in $P_w^{[N]}$ belongs to $[0,1]$. Using the fact that 
$$\vertiii{\mathbb{W}^{[N]}_{s}-\mathbb{W} }\le \vertiii{\mathbb{W}^{[N]}_{s}-\mathbb{W}^{[N]}_{w} }+\vertiii{\mathbb{W}^{[N]}_{w}-\mathbb{W} }$$
and the first statement concludes the proof. {\blue The fact that 
$|\lambda_{\textup{max}}(\mathbb{W}^{[N]}_{s})-\lambda_{\textup{max}}(\mathbb{W}) |\le  \vertiii{\mathbb{W}^{[N]}_{s}-\mathbb{W}} $ can be proven as in the previous point}.
\end{enumerate}

\end{proof}

{\blue \begin{lemma}\label{lem:lip}
Consider a graphon game satisfying Assumptions  \ref{ass:cost}, \ref{ass:constraint}B), \ref{cond} and \ref{lipschitz} with $\Omega=0$ and suppose that $\mathtt{S}(x)=\mathcal{S}$ for all $x$. Then the unique graphon equilibrium is Lipschitz continuous with constant $L_{ s}=\frac{\max\{\ell_\J,\ell_\theta\} L (s_\textup{max}+1) }{\alpha_\J}$. 
\end{lemma}
\begin{proof}
Let $\bar s$ be the unique graphon equilibrium and $\bar z=\int_0^1 W(x,y)\bar s(y)dy$.
For any $x_1,x_2\in[0,1]$ it holds
\begin{equation}\label{lip1}
\begin{aligned}
\|\bar s(x_1)-\bar s(x_2)\|&=\| \arg\max_{s\in\mathcal{S}} \J(s,\bar z(x_1), \theta(x_1))-\arg\max_{s\in\mathcal{S}} \J(s,\bar z(x_2),\theta(x_2))\|\\
&\le \frac{1}{\alpha_\J}\|\nabla_s \J(\bar s(x_1),\bar z(x_1), \theta(x_1))- \nabla_s \J(\bar s(x_1),\bar z(x_2), \theta(x_2))\|\\& \le \frac{\max\{\ell_\J,\ell_\theta\}}{\alpha_\J}\left(\|\bar z(x_1)-\bar z(x_2)\|+\|\theta(x_1)-\theta(x_2)\|\right).
\end{aligned}
\end{equation}
Moreover,
\begin{equation}\label{lip2}
\begin{aligned}
\|\bar z(x_1)-\bar z(x_2)\|&=\| \int_0^1 W(x_1,y)\bar s(y)dy- \int_0^1 W(x_2,y)\bar s(y)dy\|\\
&\le \int_0^1 |W(x_1,y)-W(x_2,y)|\| \bar s(y) \|dy \\&\le \int_0^1 L |x_1-x_2| s_\textup{max}dy=L |x_1-x_2| s_\textup{max},\\
\|\theta(x_1)-\theta(x_2)\|&\le L|x_1-x_2|.
\end{aligned}
\end{equation}
Combining \eqref{lip1} and \eqref{lip2} yields
\[\|\bar s(x_1)-\bar s(x_2)\| \le \frac{\max\{\ell_\J,\ell_\theta\} L( s_\textup{max}+1) }{\alpha_\J} |x_1-x_2|.\]
\end{proof}}

{\blue \begin{lemma}\label{concentration}

Suppose that the assumptions of Theorem \ref{thm:bayes} hold.
Consider a fixed population size $N$, a fixed $\type^i\in[0,1]$ and let  $\zeta_{\bar s}(\type^i)$ be a  realization of $\frac{1}{N-1}\sum_{j} [P^{[N]}_s]_{ij} \bar s(\type^j)$,
where $[P^{[N]}_s]$ is sampled from the graphon $W$ according to Definition \ref{sample}. Then with probability at least $1-\frac{2n+1}{(N-1)^2}$  it holds $\|\zeta_{\bar s}(\type^i)-\bar z(\type^i)\|\le \varepsilon'$, with $\varepsilon':=\mathcal{O}\left( \sqrt{\frac{\log(N-1)}{N-1}}\right)$. By the union bound  with probability at least $1-\frac{(2n+1)N}{(N-1)^2}$  it holds $\|\zeta_{\bar s}(\type^i)-\bar z(\type^i)\|\le \varepsilon'$ for all $i\in\{1,\ldots,N\}$.
\end{lemma}}
\begin{proof}
Let $t^{-i}$ be the types of all the agents except for agent $i$. For each realization of $t^{-i}$ we have

\begin{align*}
\|\zeta_{\bar s}(\type^i)-\bar z(\type^i)\|&=\|\frac{1}{N-1}\sum_{j\neq i} [P^{[N]}_s]_{ij} \bar s(\type^j)  -\bar z(\type^i)\|
\\&=\| \frac{1}{N-1}\sum_{j\neq i}\left( [P^{[N]}_s]_{ij} \bar s(\type^j) -W(\type^i,\type^j) \bar s(\type^j) +W(\type^i,\type^j) \bar s(\type^j) \right) -\bar z(\type^i)\|
\\&\le  \underbrace{ \|\frac{1}{N-1} \sum_{j\neq i}( [P^{[N]}_s]_{ij} -W(\type^i,\type^j) )\bar s(\type^j) \| }_{\textup{Term 1}}+\underbrace{\| \frac{1}{N-1}\sum_{j \neq i} W(\type^i,\type^j) \bar s(\type^j) -\bar z(\type^i)\|}_{\textup{Term 2}}.
\end{align*}

We can bound the two terms separately.
\begin{itemize}
\item Term 1: Note that $ \sum_{j\neq i}( [P^{[N]}_s]_{ij} -W(\type^i,\type^j) )\bar s(\type^j)\in\R^n$. For each $h\in\{1,\ldots,n\}$, we denote by $S_h:= \sum_{j\neq i}( [P^{[N]}_s]_{ij} -W(\type^i,\type^j) )\bar s_h(\type^j)$ the $h$-th component of the previous vector and analyze each component separately.

Let $X^h_j=( [P^{[N]}_s]_{ij} -W(\type^i,\type^j) ) \bar s_h(\type^j)$ and note that for a fixed $h$ the random variables $\{X^h_j\}_{j\neq i}$ are independent, zero mean and $-s_\textup{max} \le X^h_j\le s_\textup{max}$ for all $j\neq i$. Moreover, by definition $S_h=\sum_{j\neq i} X^h_j$. Note that $\mathbb{E}[S_h]=0$. 
The Hoeffding's inequality then  yields 

\begin{align*}
&\textup{Pr}\left[\frac{|S_h|}{N-1}>s_\textup{max} \sqrt{\frac{4\log(N-1)}{N-1}}\right] =\textup{Pr}\left[{|S_h|}>s_\textup{max} \sqrt{{4\log(N-1)}{(N-1)}}\right] \\&< 2 \exp \left( - \frac{2 s_\textup{max}^2 4\log(N-1)(N-1)}{(N-1)(2s_\textup{max})^2} \right)= 2 \exp \left( -2\log(N-1) \right) = \frac{2}{(N-1)^2}.
\end{align*}

 Hence for any $h\in\{1,\ldots,n\}$, with probability at least $1-\frac{2}{(N-1)^2}$, it holds $\frac{|S_h|}{N-1}= \mathcal{O}\left( \sqrt{\frac{\log(N-1)}{N-1}}\right)$. By the union bound, with probability at least $1-\frac{2n}{(N-1)^2}$, it holds $\frac{|S_h|}{N-1}= \mathcal{O}\left( \sqrt{\frac{\log(N-1)}{N-1}}\right)$ for all $h\in\{1,\ldots,n\}$. With the same probability
\[[\textup{term 1}]=\sqrt{ \sum_{h=1}^n \left(\frac{S_h}{N-1}\right)^2}= \mathcal{O}\left( \sqrt{\frac{\log(N-1)}{N-1}}\right).\]

\item Term 2: \\ Define $\delta_N$ and  $d_N$ as in Theorem \ref{thm:dist}.   Let $\{t^{-i}_{(k)}\}_{k=1}^{N-1}$ be the ordered statistics of $\{\type^j\}_{j\neq i}$ so that $t_{(1)}^{-i}\le \ldots \le t_{(N-1)}^{-i}$. By \cite[Proposition 3]{graphons}  (see also Lemma \ref{lem:dist}) the set   of realizations of  $\{\type^j\}_{j\neq i}$ such that $|t^{-i}_{(k)}-y|\le d_{N-1}$ for all $y\in\mathcal{U}^{[N-1]}_k:=[\frac{k-1}{N-1}, \frac{k}{N-1})$ and for all $k\in\{1,\ldots,N-1\}$ 
 has measure at least $1-\delta_{N-1}$. Consequently, with this probability it holds
 \begin{align*}
& [\textup{term 2}]=\| \frac{1}{N-1}\sum_{j\neq i} W(\type^i,\type^j) \bar s(\type^j) -\bar z(\type^i)\|=\| \frac{1}{N-1}\sum_{k=1}^{N-1} W(\type^i,t^{-i}_{(k)}) \bar s(t^{-i}_{(k)}) -\bar z(\type^i)\|\\&=  \| \frac{1}{N-1}\sum_{k=1}^{N-1} W(\type^i,t^{-i}_{(k)}) \bar s(t^{-i}_{(k)}) - \sum_{k=1}^{N-1} \int_{\mathcal{U}^{[N-1]}_k} W(\type^i,y)\bar s(y)dy\|
 \\&=  \| \sum_{k=1}^{N-1}\int_{\mathcal{U}^{[N-1]}_k} [ W(\type^i,t^{-i}_{(k)}) \bar s(t^{-i}_{(k)})- W(\type^i,y)\bar s(y) ] dy \|
  \\&\le   \sum_{k=1}^{N-1}\int_{\mathcal{U}^{[N-1]}_k} \| W(\type^i,t^{-i}_{(k)}) \bar s(t^{-i}_{(k)})- W(\type^i,t^{-i}_{(k)})\bar s(y)\| + \| W(\type^i,t^{-i}_{(k)})\bar s(y) -  W(\type^i,y)\bar s(y) \| dy 
    \\&\le   \sum_{k=1}^{N-1}\int_{\mathcal{U}^{[N-1]}_k}  (L_s + L  s_\textup{max}) |t^{-i}_{(k)}- y| dy  \le (L_s + L  s_\textup{max}) d_{N-1},
 \end{align*}
 
 where the second to last inequality follows from the fact that, under the given assumptions, $\bar s$ is Lipschitz continuous with constant $L_s$ (see Lemma \ref{lem:lip}), $\|\bar s\|\le s_\textup{max}$ and $W$ is Lipschitz continuous with constant $L$.
 By selecting $\delta_{N-1}=\frac{1}{(N-1)^2}$ with probability at least $1-\frac{1}{(N-1)^2}$, [term 2]$= \mathcal{O}(d_{N-1})=\mathcal{O}\left( \sqrt{\frac{\log(N-1)}{N-1}}\right)$.
 \end{itemize}
 By the union bound with probability at least  $1-\frac{2n+1}{(N-1)^2}$ it holds $\|\zeta_{\bar s}(\type^i)-\bar z(\type^i)\|=\mathcal{O}\left( \sqrt{\frac{\log(N-1)}{N-1}}\right).$

\end{proof}

{\blue \begin{lemma}\label{thm:eps}
Consider a graphon game $\mathcal{G}(\mathtt{S},U, \theta,W)$ where $\mathtt{S}(x)=\mathcal{S}$ for all $x\in[0,1]$. Suppose that Assumptions \ref{ass:cost}, ~\ref{ass:constraint}B),~\ref{cond},  \ref{lipschitz}  (with $\Omega=0$) and  \ref{cost2} hold. Let $\bar s$ be the unique equilibrium of the graphon game. Then  with probability $1-\frac{(2n+1)}{N}$, the set $\{\tilde s^i:=\bar s(t^i)\}
_{i=1}^N$ is an $\varepsilon$ Nash equilibrium of the sampled network game $\mathcal{G}^{[N]}(\{\mathcal{S}\}_{i=1}^N, \J,{ \{{\theta}(t^i)\}_{i=1}^N}, P^{[N]}_{s})$ with 
\[\varepsilon =\mathcal{O}\left( \sqrt{\frac{\log(N)}{N}}\right).\]
\end{lemma}
\begin{proof}
For  any agent $i$, let $\tilde z^i=\frac1N\sum_j [P^{[N]}_{s}]_{ij} \tilde s^j = \frac1N\sum_j [P^{[N]}_{s}]_{ij} \bar s(t^j) =\zeta_{\bar s}(t^i)$ then for any $s^i\in \mathcal{S}$
\begin{align*}
\J(\tilde s^i, \tilde z^i, { \theta^i}) & = \J(\bar s(t^i),\zeta_{\bar s}(t^i), { \theta(t^i)}) \ge \J(\bar s(t^i),\bar z(t^i), { \theta(t^i)}) - L_U\| \zeta_{\bar s}(t^i)- \bar z(t^i)\| \\
&\ge \J(s^i,\bar z(t^i), { \theta(t^i)}) - L_U\| \zeta_{\bar s}(t^i)- \bar z(t^i)\| \\
&\ge \J(s^i,\zeta_{\bar s}(t^i), { \theta(t^i)}) - 2L_U\| \zeta_{\bar s}(t^i)- \bar z(t^i)\| =\J(s^i, \tilde z^i, { \theta^i})- 2L_U\| \zeta_{\bar s}(t^i)- \bar z(t^i)\|.
\end{align*}
The proof is concluded by noting that by Lemma \ref{concentration}, with probability  $1-\frac{(2n+1)}{N}$, $\| \zeta_{\bar s}(t^i)- \bar z(t^i)\|= \mathcal{O}\left( \sqrt{\frac{\log(N)}{N}}\right)$ for all agents $i=1,\ldots,N$ (note that Lemma \ref{concentration} is proven for normalization $\frac{1}{N-1}$ but similar arguments apply to $\frac1N$).
\end{proof}
}

\end{document}